\newif\ifeprint
\eprinttrue

\documentclass{article}
\usepackage{fullpage,parskip}
\usepackage{amsmath, amsthm, amssymb,bbm,enumitem}
\usepackage[ruled,linesnumbered,noend]{algorithm2e}
\usepackage[colorlinks,citecolor=blue]{hyperref}
\usepackage{cleveref}
\usepackage{threeparttable,caption}
\usepackage[
	lambda,
	operators,
	advantage,
	sets,
	adversary,
	landau,
	probability,
	notions,	
	logic,
	ff,
	mm,
	primitives,
	events,
	complexity,
	asymptotics,
	keys]{cryptocode}
\usepackage{tikz}
\usetikzlibrary{decorations.pathreplacing, angles, quotes}
\usetikzlibrary{shapes.geometric}
\usetikzlibrary{patterns}
\usetikzlibrary{positioning,calc}

\usepackage{pifont}
\newcommand{\cmark}{\ding{51}}%
\newcommand{\xmark}{\ding{55}}%
\usepackage[normalem]{ulem}

\makeatletter
\newtheorem*{rep@theorem}{\rep@title}
\newcommand{\newreptheorem}[2]{%
\newenvironment{rep#1}[1]{%
 \def\rep@title{#2 \ref{##1}}%
 \begin{rep@theorem}}%
 {\end{rep@theorem}}}
\makeatother

\makeatletter
\newtheorem*{rep@definition}{\rep@title}
\newcommand{\newrepdefinition}[2]{%
\newenvironment{rep#1}[1]{%
 \def\rep@title{#2 \ref{##1}}%
 \begin{rep@definition}}%
 {\end{rep@definition}}}
\makeatother

\definecolor{ceil}{rgb}{0.57, 0.63, 0.81}

\newcommand{\eps}{\varepsilon}
\renewcommand{\epsilon}{\varepsilon}
\newcommand{\F}{\mathbb{F}}
\newcommand{\N}{\mathbb{N}}
\newcommand{\Z}{\mathbb{Z}}
\newcommand{\R}{\mathbb{R}}
\newcommand{\from}{\gets}
\newcommand{\LPN}{\mathsf{LPN}}
\renewcommand{\XOR}{\mathsf{XOR}}

\newcommand{\E}{\mathop{{}\mathbb{E}}}
\DeclareMathOperator{\calA}{\mathcal{A}}

\DeclareMathOperator{\calC}{\mathcal{C}}
\DeclareMathOperator{\calD}{\mathcal{D}}
\DeclareMathOperator{\calE}{\mathcal{E}}
\DeclareMathOperator{\calF}{\mathcal{F}}

\DeclareMathOperator{\calP}{\mathcal{P}}

\DeclareMathOperator{\calO}{\mathcal{O}}
\DeclareMathOperator{\calU}{\mathcal{U}}
\DeclareMathOperator{\calS}{\mathcal{S}}
\DeclareMathOperator{\calX}{\mathcal{X}}
\DeclareMathOperator{\calR}{\mathcal{R}}
\DeclareMathOperator{\calM}{\mathcal{M}}
\DeclareMathOperator{\calQ}{\mathcal{Q}}
\DeclareMathOperator{\calW}{\mathcal{W}}
\DeclareMathOperator{\calT}{\mathcal{T}}

\DeclareMathOperator{\rank}{\mathrm{rank}}
\DeclareMathOperator{\Ber}{\mathrm{Ber}}
\DeclareMathOperator{\Hyp}{\mathrm{Hyp}}
\DeclareMathOperator{\wt}{\mathrm{wt}}
\DeclareMathOperator{\Maj}{\mathrm{Maj}}
\DeclareMathOperator{\len}{\mathrm{len}}

\newcommand{\Model}{\mathsf{Model}}
\newcommand{\RModel}{\overline{\mathsf{Model}}}

\newcommand{\Detect}{\mathsf{Detect}}
\newcommand{\prompt}{\textsc{prompt}}
\newcommand{\Wat}{\mathsf{Wat}}
\newcommand{\Setup}{\mathsf{Setup}}
\newcommand{\done}{\texttt{done}}
\newcommand{\empH}{H_e}
\newcommand{\empHtr}{\bar{H_e}}
\newcommand{\Red}{\mathsf{Red}}
\newcommand{\KeyGen}{\mathsf{KeyGen}}
\newcommand{\Encode}{\mathsf{Encode}}
\newcommand{\Decode}{\mathsf{Decode}}

\newcommand{\Enc}{\mathsf{Enc}}
\newcommand{\Dec}{\mathsf{Dec}}
\newcommand{\m}{{\sf m}}
\newcommand{\p}{{\sf p}}
\newcommand{\q}{{\sf q}}
\renewcommand{\t}{{\sf t}}
\renewcommand{\b}{{\sf b}}
\newcommand{\KeyGendel}{\mathsf{KeyGen}_{\sf del}}
\newcommand{\Encodedel}{\mathsf{Encode}_{\sf del}}
\newcommand{\Decodedel}{\mathsf{Decode}_{\sf del}}

\newcommand{\Generate}{\mathsf{Generate}}

\newcommand{\MajEncode}{\mathsf{MajEnc}}
\newcommand{\RepEncode}{\mathsf{RepEnc}}
\newcommand{\MajDecode}{\mathsf{MajDec}}

\newcommand{\PRG}{\mathsf{PRG}}
\newcommand{\PRC}{\mathsf{PRC}}
\newcommand{\LDPC}{\mathsf{LDPC}}
\newcommand{\LDPCPRC}{\mathsf{LDPC\text-PRC}}
\newcommand{\PermPRC}{\mathsf{Permuted\text-PRC}}

\newcommand{\PRCdel}{\mathsf{PRC}_{\sf del}}

\newcommand{\Gen}{\mathsf{Gen}}
\newcommand{\Sign}{\mathsf{Sign}}
\newcommand{\Vrfy}{\mathsf{Vrfy}}
\newcommand{\SigForge}{\mathsf{SigForge}}

\newcommand{\AttrForge}{\mathsf{AttrForge}}
\newcommand{\ForgeDetect}{\mathsf{AttrText}}

\newcommand{\Watt}{\calW_{\sf att}}

\newcommand{\sigsk}{\mathsf{Sig.sk}}
\newcommand{\sigpk}{\mathsf{Sig.pk}}

\newcommand{\prcsk}{\mathsf{PRC.sk}}

\newcommand{\Sig}{\mathsf{Sig}}
\newcommand{\Steg}{\mathsf{Steg}}

\newcommand{\BDC}{\text{BDC}}
\newcommand{\BSC}{\text{BSC}}
\newcommand{\Eemb}{\calE_{\mathsf{Emb}}}
\newcommand{\Eadv}{\calE_{\sf adv}}

\newtheorem{theorem}{Theorem}
\newtheorem*{theorem*}{Theorem}
\newreptheorem{theorem}{Theorem}

\newtheorem*{corollary*}{Corollary}

\newtheorem{lemma}{Lemma}
\newtheorem{numberedclaim}[lemma]{Claim}

\newtheorem*{remark*}{Remark}

\theoremstyle{definition}
\newtheorem{definition}{Definition}
\newtheorem*{definition*}{Definition}
\newreptheorem{definition}{Definition}
\newtheorem{construction}{Construction}
\newtheorem{assumption}{Assumption}
\newtheorem*{assumption*}{Assumption}

\crefname{construction}{construction}{constructions}
\crefname{numberedclaim}{claim}{claims}

\DeclareMathOperator{\Tr}{Tr}
\DeclareMathOperator{\img}{Im}
\DeclareMathOperator{\sgn}{sgn}

\newcommand{\SD}{\mathsf{SD}}

\newcommand{\gamesfontsize}{\small}

\newcommand{\hfpages}[2]{\framebox{\begin{minipage}[t]{#1\textwidth}\gamesfontsize #2 \end{minipage}}}

\title{Pseudorandom Error-Correcting Codes}
\author{Miranda Christ\thanks{Equal contribution. Email addresses: \texttt{mchrist@cs.columbia.edu}, \texttt{gunn@berkeley.edu}}\\ 
        \small Columbia University 
        \and 
        Sam Gunn$^{*}$\thanks{Supported by a Google PhD Fellowship.}\\ 
        \small UC Berkeley}

\begin{document}
\date{}

\maketitle

\begin{abstract}
    We construct \emph{pseudorandom error-correcting codes} (or simply \emph{pseudorandom codes}), which are error-correcting codes with the property that any polynomial number of codewords are pseudorandom to any computationally-bounded adversary.
    Efficient decoding of corrupted codewords is possible with the help of a decoding key.

    We build pseudorandom codes that are robust to substitution and deletion errors, where pseudorandomness rests on standard cryptographic assumptions.
    Specifically, pseudorandomness is based on either $2^{O(\sqrt{n})}$-hardness of LPN, or polynomial hardness of LPN and the planted XOR problem at low density. 

    As our primary application of pseudorandom codes, we present an undetectable watermarking scheme for outputs of language models that is robust to cropping and a constant rate of random substitutions and deletions.
    The watermark is undetectable in the sense that any number of samples of watermarked text are computationally indistinguishable from text output by the original model.
    This is the first undetectable watermarking scheme that can tolerate a constant rate of errors.

    Our second application is to steganography, where a secret message is hidden in innocent-looking content.
    We present a constant-rate stateless steganography scheme with robustness to a constant rate of substitutions.
    Ours is the first stateless steganography scheme with provable steganographic security and \emph{any} robustness to errors.
\end{abstract}

\thispagestyle{empty}
\newpage
{\small\tableofcontents}
\newpage

\section{Introduction} \label{sec:intro}
The proliferation of AI-generated content is one of the biggest issues facing the internet today.
Recent breakthroughs in large language models have made it increasingly difficult to distinguish this influx of AI-generated content from human-generated content.

A promising solution for detecting AI-generated content is \emph{watermarking}, where a hidden signal is embedded in the content.
Several major companies, including OpenAI and Google, have pledged to embed watermarks in content output by their models \cite{BH23}.
Despite this explosion of interest in watermarking, there are very few techniques for building watermarking schemes that do not noticeably alter the generated content.
Existing schemes incur trade-offs between the quality of generated content, the robustness of the watermark, and the computational complexity of detection.

In this work, we take a new cryptographic approach to this problem that allows us to avoid some of these trade-offs.
Our approach is based on a new cryptographic primitive that we call a pseudorandom error-correcting code, or simply a pseudorandom code (PRC).
A PRC is an error-correcting code that is parameterized by a decoding key.
The pseudorandomness property states that, without this decoding key, any polynomial number of codewords are pseudorandom.

We find that the problem of building robust, quality-preserving watermarks reduces to the problem of building PRCs. Essentially, the watermarking strategy is to replace some of the randomness used by the generative algorithm with outputs from a PRC.

Building PRCs is challenging: Error-correcting codes are typically highly structured, while pseudorandomness implies a lack of discernible structure.
Indeed, a priori it is not clear that such objects should exist.
Nonetheless, we construct PRCs from standard (subexponential) cryptographic assumptions.
Our constructions are related to low-density parity-check codes, and we base pseudorandomness on the Learning Parity with Noise assumption.
We construct PRCs with strong robustness properties, including robustness to any constant rate of substitution and deletion errors.

Applying these PRCs to watermarking for language models, we obtain the first quality-preserving language model watermarking schemes that are robust to cropping and any constant rate of substitution and deletion errors.
That is, the watermark detector will work as long as it is provided any sufficiently-long sequence of text, even if that text is subjected to any constant (less than $1/2$) rate of random substitutions and any constant (less than $1$) rate of random deletions.

\subsection{An approach to watermarking}
\label{subsec:how-to-water}
In this subsection we present a simple, general template for watermarking AI-generated content.
The template described here can in principle be used to watermark arbitrary media, but we only present concrete instantiations in certain contexts (\Cref{sec:water,sec:stego}).

In all generative AI settings there is a generative algorithm, $\Generate$, that defines the behavior of the AI.
A user provides a prompt, and $\Generate$ outputs some randomly-generated content.
A watermarking scheme modifies $\Generate$ so that the generated content contains a hidden pattern, called a watermark.
There are two essential requirements that any watermarking scheme should satisfy:
\begin{itemize}
    \item \textit{Quality}: embedding the watermark should not reduce the quality of generative algorithm; and
    \item \textit{Robustness}: the watermark should be detectable in generated content, even if this content is corrupted.
\end{itemize}
Achieving both of these properties simultaneously is the central challenge of watermarking. Quality means that the watermark should not significantly alter the generated content, while robustness seems to require the watermark to change the content a great deal.

In this work, we propose a new strategy for watermarking: \emph{replacing the randomness used by $\Generate$ with codewords from a pseudorandom error-correcting code}.
A pseudorandom error-correcting code, or simply pseudorandom code (PRC), is a new cryptographic primitive that we introduce in this work.
A PRC is defined by algorithms $\Encode, \Decode$ satisfying two properties:
\begin{itemize}
    \item \emph{Pseudorandomness}: Any efficient adversary, without knowledge of the decoding key, cannot distinguish between oracle access to $\Encode$ and an oracle that always outputs a fresh random string; and
    \item \emph{Error correction} or \emph{robustness}: For any message $\m$, if $x \from \Encode(\m)$ and $x'$ is a ``corrupted'' version of $x$ where the amount of corruption is bounded, then $\Decode(x') = \m$.
\end{itemize}
For watermarking the message can be simply $\m = 1$, indicating the presence of a watermark.\footnote{By instead encoding a longer message in the PRC, this technique extends to steganography --- where messages are secretly communicated in innocent-looking content --- as well.}

In order to make detection possible, we specify $\Generate$ in such a way that the detector can approximate the randomness used to produce any given content.
To test for the presence of a watermark, the detector computes this approximation and then applies $\Decode$ to the result.
If the content is watermarked and the approximation is close enough to the true randomness, then robustness of the PRC ensures that $\Decode$ returns 1. This indicates to the detector that the watermark is present.
Stronger robustness of the PRC translates to stronger robustness of the watermark.

Pseudorandomness of the PRC guarantees that, without the decoding key, watermarked content is indistinguishable from content produced by the original generative algorithm --- even if one is allowed to see many outputs.
In particular, the quality of the generative algorithm is not deteriorated by the watermark.
In \cite{CGZ23}, such a watermark is referred to as \emph{undetectable}.

Therefore, the problem of building robust, quality-preserving watermarks reduces to building PRCs (and appropriately specifying $\Generate$).
Below, we describe the application of this template to watermarking for language models.

\paragraph{Watermarks for language models.}
A generative language model is a randomized algorithm that takes as input a prompt, and samples text constituting a response.
This text consists of \emph{tokens}.
For simplicity, we assume here that tokens are binary and that the full response is of a fixed length $n$.
Neither of these assumptions is important for our results, as discussed in \Cref{sec:water}.

Given any generative language model, it is not difficult to define an algorithm $\Generate$ that takes a prompt and a random seed $x \in \{0,1\}^n$, and samples a response $\t \in \{0,1\}^n$ such that
\begin{itemize}
    \item if $x$ is uniformly random, then $\t$ is distributed identically to the original model, and
    \item each bit of $\t$ is correlated with the corresponding bit of $x$.
\end{itemize}
See \Cref{subsec:techo-water} for an example of such an algorithm.
Now it is easy to recover an approximation of the randomness $x$ that was used to produce a given text: just use the text $\t$ itself.

One natural watermarking strategy is to use the \emph{same} random seed $x \in \{0,1\}^n$ for every call to $\Generate$, storing $x$ itself as the watermarking key.
This is essentially the strategy used by \cite{KTHL23}, and it yields a highly robust watermark because the detector can compute the edit distance between the given text and $x$.
The resulting quality guarantee is that a single response from the watermarked model is distributed identically to a single response from the original model.

However, this strategy results in redundancy of responses because the $i^{\text{th}}$ token of every response is biased towards the $i^{\text{th}}$ bit of $x$.
This is problematic as it limits the variability of text that the language model generates.
For instance, it should not be the case that certain words are preferred as the first word of every response.
One can mitigate this issue by storing a family of seeds $x_1, \dots, x_\ell \in \{0,1\}^n$ and randomly choosing one such seed for each response.
Increasing $\ell$ improves the variability, but comes at the cost of a corresponding increase in both watermark key length \emph{and} detector runtime.
Now, the detector must compute the edit distance for each seed, resulting in a runtime of $O(n^2 \ell)$.
In particular, for any polynomial-time watermarking scheme using this approach, there exists a polynomial-time adversary that can distinguish watermarked content from un-watermarked content without needing the watermark detection key.

A PRC is exactly the object needed to avoid this tradeoff between response variability and efficiency. 
Our watermark detection key will be the decoding key for a PRC, and we will embed the watermark by sampling a fresh pseudorandom seed $x \from \Encode(1)$ for each query to $\Generate$.
This results in no observable correlations between responses, regardless of the number of queries --- i.e., the watermark is undetectable.
Since the same hidden structure is present in every sample from $\Encode(1)$, our detector can simply apply $\Decode$ to check for this structure.
Now, the detector's runtime has no dependence on the response variability. 
Using PRCs that we construct in \Cref{sec:ldpc-prcs,sec:improved-prcs}, we also find that our watermarking scheme can be made robust to a constant fraction of random substitution and deletion errors.

\paragraph{A note on undetectability.}
Undetectability is a strong quality guarantee for watermarking.
Outputs of the watermarked model must be \emph{computationally indistinguishable} from outputs of the original model, even to a user who is allowed to make many queries.
While this is imperative for steganography, its necessity in watermarking is less clear, as some noticeable changes to the model may be permissible as long as the outputs remain ``high-quality.''

However, measuring the quality of watermarked content is often challenging or impossible --- especially when the content is used in a wide range of applications.
Computational indistinguishability is a strong quality guarantee that applies uniformly to every application: it implies that the watermark causes no observable loss in \emph{any} efficiently-computable quality metric.
Without such a guarantee, it is impossible to verify that a watermark is quality-preserving across all applications.
We therefore focus on undetectability in this work.

\subsection{Our contributions}
\label{subsec:contribs}

\paragraph{Pseudorandom codes (\Cref{sec:prc-basics,sec:ldpc-prcs,sec:improved-prcs}).}
Our first contribution is to identify PRCs as an interesting cryptographic primitive, with applications to robustly hiding messages in AI-generated content.
Very roughly, the definition is as follows --- for more details, see either the technical overview (\Cref{subsec:techo-prc-basics}) or the formal definitions (\Cref{subsec:prc-defs}).

\begin{definition*}[\Cref{def:skPRC,def:pkPRC}]
    A \emph{pseudorandom code} (PRC) is an error-correcting code where codewords are pseudorandom to any computationally-bounded adversary who doesn't hold the decoding key.
\end{definition*}

We consider both public- and secret-key variants of PRCs.
For a public-key PRC, the encoding algorithm can be made public without sacrificing pseudorandomness.
When the message space consists of only a single message, we call it a \emph{zero-bit} PRC.
Zero-bit PRCs can also be viewed as robust \emph{backdoored (or trapdoor) pseudorandom generators} \cite{trapdoor,dodis2015formal}, and they are sufficient for our applications to watermarking.

We show how to build zero-bit public-key PRCs related to low-density parity-check (LDPC) codes, where pseudorandomness rests on standard (subexponential) cryptographic assumptions.
All of the PRCs we construct are over a binary alphabet.
Depending on the parameter choices, the pseudorandomness of these LDPC-related codes is based on either of two assumptions, which we state together as \Cref{assumption:combined}. 

\begin{assumption*}[\Cref{assumption:combined}]
    Either
    \begin{itemize}
        \item LPN is hard for any $2^{O(\sqrt{n})}$-time adversary, or
        \item LPN and planted XOR are both hard for any polynomial-time adversary.
    \end{itemize}
\end{assumption*}

For descriptions of the LPN and planted XOR assumptions, see either \Cref{subsec:techo-ldpc-prcs} or \Cref{sec:ldpc-prcs}. Under \Cref{assumption:combined}, we prove that there exist PRCs with robustness to channels that introduce a bounded number of substitution errors. We say that any channel that introduces at most a $p$ fraction of substitution errors (and no other types of errors) is $p$-\emph{bounded}.

\begin{theorem*}[\Cref{theorem:ldpc-prc-lpn,theorem:ldpc-prc-xor}]
    Let $p \in (0,1/2)$ be any constant. Under \Cref{assumption:combined}, there exists a zero-bit public-key PRC that is robust to every $p$-bounded channel.
\end{theorem*}

For some applications it will be useful to have multi-bit PRCs.
Any such construction should ideally have a high \emph{rate}, which is the ratio of the number of message bits to the number of codeword bits.
We prove that any zero-bit PRC can be combined with any error correcting code to give a multi-bit PRC.

\begin{theorem*}[\Cref{theorem:constant-rate-prcs}]
    Suppose that there exists a zero-bit (public-key) PRC and a rate-$R$ error-correcting code, that are both robust to every $p$-bounded channel. Then there exists a (public-key) PRC of rate $R-o(1)$ that is robust to every $(p-\varepsilon)$-bounded channel, for every constant $\varepsilon > 0$.
\end{theorem*}

Applying this theorem with our zero-bit LDPC-based PRCs and the binary error-correcting codes of \cite{ABN+92,NN93,TaShma17}, we have the following corollary.

\begin{corollary*}
    Let $p \in (0,1/2)$ be any constant. Under \Cref{assumption:combined}, there exists a constant-rate PRC that is robust to every $p$-bounded channel.
\end{corollary*}

None of the PRCs mentioned so far can handle deletions.
Deletions are a particularly important type of edit for text watermarks, because an adversary may try to remove the watermark by simply deleting some of the words.

Deletions are notoriously difficult to handle in error correction, and standard techniques involve significant structure --- thus violating pseudorandomness.
Nonetheless, we show that if a PRC has sufficiently strong robustness to substitutions, then it can be converted to a PRC with robustness to deletions (at the cost of a decreased rate).

Let $\BSC_p$ be the binary symmetric channel with error rate $p$, and $\BDC_q$ be the binary deletion channel with deletion rate $q$. That is, $\BSC_p$ randomly flips each bit with probability $p$ and $\BDC_q$ randomly deletes each bit with probability $q$.

\begin{theorem*}[\Cref{theorem:deletion-code}]
    For any constants $p \in (0,1/2)$ and $q \in (0,1)$, there exists $p' \in (0,1/2)$ such that the following holds.
    If there exists a zero-bit (public-key) PRC with robustness to every $p'$-bounded channel, then there exists a zero-bit (public-key) PRC that is robust to the channel $\BSC_p \circ \BDC_q$. 
\end{theorem*}

Together with our LDPC-based PRCs, we obtain the following result.

\begin{corollary*}
    Let $p \in (0,1/2)$ and $q \in (0,1)$ be any constants. 
    Under \Cref{assumption:combined}, there exists a zero-bit public-key PRC that is robust to the channel $\BSC_p \circ \BDC_q$.
\end{corollary*}

\paragraph{Watermarking and steganography for language models (\Cref{sec:water}).}
We apply our zero-bit PRCs to build the first undetectable watermarking scheme for language models with robustness to a constant rate of random substitutions and deletions.
In this section, we assume that text output by the language model is represented as a bitstring.
Arbitrary text can be mapped to a bitstring by either randomly assigning a single bit to each token, or by expanding the tokens into a binary representation.

\begin{theorem*}[\Cref{theorem:robust-water}]
    Let $p \in (0,1/2)$ be any constant. Under \Cref{assumption:combined}, there exists an undetectable watermarking scheme $\calW$ such that the watermark appears in any sufficiently high-entropy text, even if the text is subjected to the channel $\BSC_p$.
\end{theorem*}

Under an extra assumption about the generated text, which roughly corresponds to the text having few repeated words, we can strengthen this to handle deletions as well.

\begin{theorem*}[\Cref{theorem:deletion-robust-water}]
    Let $p \in (0,1/2)$ and $q \in (0,1)$ be any constants. Under \Cref{assumption:combined}, there exists an undetectable watermarking scheme $\calW$ such that the watermark appears in any sufficiently high-entropy and ``variable'' text, even if the text is subjected to the channel $\BSC_p \circ \BDC_q$.
\end{theorem*}

In all of our theorems the text can be cropped, as long as the remaining text is sufficiently high-entropy.

We also construct undetectable watermarking schemes with \emph{unforgeable public attribution} and the same robustness as $\calW$.
Public attribution means that there is a public algorithm to identify which portion of a given text was output by the model.
Unforgeability means that no efficient user can produce text that the attribution algorithm identifies as model-generated, but that was not output by the model.
Interestingly, our schemes retain the standard robust secret-key detector in addition to this public attribution algorithm.

\begin{theorem*}[\Cref{theorem:watermark-att}]
    Under \Cref{assumption:combined}, there exists a watermarking scheme $\Watt$ that retains all properties of $\calW$ from \Cref{theorem:robust-water} or \Cref{theorem:deletion-robust-water}, and additionally satisfies unforgeable public attribution.
\end{theorem*}

Finally, using PRCs with constant rate (\Cref{theorem:constant-rate-prcs}), we also obtain the first \emph{robust} language model steganography scheme.\footnote{Since this scheme doesn't rely on the decoder having access to the prompt, it can also be seen as an undetectable ``multi-bit watermarking scheme.''}

\begin{theorem*}[\Cref{theorem:language-model-stego}]
    Let $p \in (0,1/2)$ be any constant. Under \Cref{assumption:combined}, there exists a language model steganography scheme with constant information rate, such that the message can be recovered from any sufficiently high-entropy text, even if the text is subjected to the channel $\BSC_p$.
\end{theorem*}

\paragraph{Universal steganography (\Cref{sec:stego}).}
In \Cref{sec:stego} we show that PRCs can be used to solve a long-standing open question in steganography: A simple application of PRCs yields the first robust, stateless universal steganography scheme.
Universal steganography can be used for language model steganography, but it is more general \cite{vAH04}.
We take the rate of the steganography scheme to be the ratio of the number of stegotext symbols to the number of bits in the message being encoded.

\begin{theorem*}[\Cref{theorem:stego}]
    Suppose there is a hash function that is unbiased over the covertext channel.
    If $\PRC$ is any PRC, then there exists a stateless, public-key universal steganography scheme with the same rate and robustness as $\PRC$.
\end{theorem*}

Finally, we show that this result can be extended to the setting where an unbiased hash function on the covertext channel is not known, with some loss in robustness.

\begin{theorem*} [\Cref{theorem:stego-weaker-assumption}]
    Suppose there is a hash function that has constant min-entropy over the covertext channel.
    Then under \Cref{assumption:combined}, for any $p \in (0,1/2)$, there exists a constant-rate public-key stateless steganography scheme that is robust to a $p$ rate of random substitutions.
\end{theorem*}

\subsection{Related work: short summary}
\label{subsec:related-work-summary}

We briefly outline some of the related work here. See \Cref{sec:full-related-work} for a more complete discussion.

\begin{itemize}
    \item Code-based cryptography: Our work bears some similarity to the field of code-based cryptography. However, code-based cryptography is generally focused on building existing primitives from new assumptions --- whereas PRCs are a new primitive that we base on existing assumptions.
    \item Trapdoor pseudorandom generators: Our zero-bit PRCs can be equivalently viewed as \emph{robust} trapdoor (or equivalently, backdoored) pseudorandom generators \cite{trapdoor,dodis2015formal}. That is, we require the additional property that the trapdoor (or secret key) can be used to detect even \emph{corrupted} pseudorandom strings.
    \item Watermarking for language models: We build watermarking schemes for language models satisfying the strongest quality guarantee, undetectability. Undetectability was defined by \cite{CGZ23}, where undetectable watermarks for language models were also constructed. In that work and in \cite{scott,KGW+23}, it is essential for watermark detection that many sufficiently-long contiguous substrings of the response remain \emph{unchanged}. Therefore, these watermarks are easily removable by simple attacks (see the ``robustness of our watermark'' paragraph of \Cref{subsec:techo-water}). The watermarks of \cite{KTHL23} are more robust --- their robustness is more comparable to ours --- but they sacrifice undetectability. Instead, their watermarks satisfy the weaker property of distortion-freeness, which is the single-query version of undetectability. \cite{ZALW23} obtain even stronger robustness, at the cost of even further reduced quality.
    \item Impossibility of strong watermarks: \cite{zhang2023watermarks} explore the possibility of watermarking for language models in the presence of motivated adversaries. They argue that sampling a \emph{random} response is easier when one is provided \emph{any} response. Since a random response cannot be watermarked (or else there would be a high false-positive rate), they use this to argue that any watermarked language model necessarily provides some assistance in generating un-watermarked text.
    \item Steganography: Steganography is the study of concealing secret messages in innocent-looking content. Whereas encryption is about hiding the message, steganography is about hiding the \emph{existence} of the message. Ever since steganography was formalized by \cite{HLvA02}, \emph{robust} steganography schemes (that don't require a shared state) have remained elusive. We resolve this problem using PRCs.
\end{itemize}

\subsection{Organization}
\ifeprint
In \Cref{sec:techo}, we give a technical overview of the paper, which is self-contained and sufficient to understand the high-level ideas and results of each section.
In \Cref{sec:prelims}, we state relevant preliminaries and notation.
In \Cref{sec:prc-basics}, we formally define PRCs and provide a heuristic construction. 
In \Cref{sec:ldpc-prcs}, we build PRCs from LDPC codes and prove their pseudorandomness from standard cryptographic assumptions.
We then show in \Cref{sec:improved-prcs} how to improve the rate and robustness of any PRC, resulting in our constant-rate multi-bit PRCs and PRCs for deletion channels.
In \Cref{sec:water}, we present our language model watermarking schemes from PRCs, including both our standard watermarks and our watermarks with unforgeable public attribution.
Finally, in \Cref{sec:stego}, we show how PRCs can be used to construct robust universal steganography schemes.

\Cref{sec:full-related-work} gives a more comprehensive overview of related work.

\fi

\subsection{Differences from a previous version}
This version of the paper contains two significant technical updates relative to the previous version, as well as a few editorial updates. The two significant updates address subtle technical issues, but do not substantially change any of the ideas or messages of the paper.

The first significant update is to the choice of parameters for which we invoke the planted XOR assumption (\Cref{assumption:planted-xor}). The previous version of the paper invoked the assumption with $m$, the number of samples, set to $\Omega(n)$. However, as pointed out to us by an anonymous CRYPTO 2024 reviewer, Theorem 4.26 of \cite{ASSVV23} (which is itself an updated version of a paper that we had cited in the old version of this paper) shows this assumption to be false. Fortunately we had only set $m = \Omega(n)$ because it made notation more convenient, and by instead setting $m = n^{1-\Omega(1)}$, we avoid the attack without sacrificing any of our results. We have updated all of the relevant theorem statements and proofs in this version. We emphasize that this issue did not affect our other main proof that our LDPC-based PRCs are pseudorandom, because that proof only relies on LPN.

The other significant update regards the completeness and robustness of our watermarking scheme in \Cref{sec:water}. To sample a long response, we repeatedly sample fresh PRC codewords to bias the text. However, in order to argue that the resulting text contains the watermark, we need to ensure that the channel applied to the text is $p$-bounded --- which, according to our definition, means that the error channel is not allowed to depend on the PRC key. But if the response contains multiple PRC blocks, then the error channel on one block could depend on codewords embedded in the previous blocks, violating this condition. In the prior version we had missed this issue, making our completeness and robustness claims incorrect. 
In this updated version we apply a one-time pad to each PRC codeword before embedding it, to ensure that the error channel cannot depend on the PRC key --- thus resolving the issue.

\ifeprint
\paragraph{Acknowledgements.}
We thank Yael Kalai, Venkat Guruswami, Rainer B\"ohme, Or Zamir, Shyamal Patel, and Shivam Nadimpalli for helpful research conversations. 
We additionally thank Vinod Vaikuntanathan for pointing out the LDPC-based PRC variant that we use to present \Cref{lemma:low-rate-random-ldpc} in the Technical Overview.
We thank Mihalis Yannakakis and Fermi Ma for helpful suggestions about the write-up. We thank Omar Alrabiah for help in proving \Cref{lemma:low-rate-random-ldpc}, as well as general assistance with coding theory.
We thank an anonymous CRYPTO 2024 reviewer for helpful feedback, especially regarding our use of the planted XOR assumption (\Cref{assumption:planted-xor}).
\else
\fi
\newpage

\section{Technical overview} \label{sec:techo}
\subsection{Pseudorandom code basics} \label{subsec:techo-prc-basics}
Pseudorandom codes (PRCs) can be viewed as a combination of two related primitives:
\begin{itemize}
    \item Pseudorandom encryption, where ciphertexts are indistinguishable from random under a chosen plaintext attack \cite{rogaway2003ocb}. Secret-key pseudorandom encryption is easy to build using a pseudorandom function $F_\sk$ --- just encrypt $\m$ by sampling a random $r$ and outputting $(r, \m \oplus F_\sk(r))$. Public-key pseudorandom encryption is also known from standard assumptions \cite{vAH04}. However, none of these constructions have any nontrivial robustness.
    \item Robust encryption, where encryptions of messages are robust to errors. Robust encryption is easy to build by applying an error-correcting code to ciphertexts from any standard encryption scheme. Even if that encryption scheme is pseudorandom, the use of the error-correcting code will in general render the robust encryption scheme not pseudorandom.
\end{itemize}
A PRC is required to simultaneously satisfy \emph{both} pseudorandomness and robustness --- properties that are in direct tension with each other. Using the secret key, one should be able to discern the redundancy and structure that give ciphertexts their robustness. Without the secret key, ciphertexts must appear completely unstructured.

We define secret-key PRCs below. For public-key PRCs (\Cref{def:pkPRC}), we further require that the encoding algorithm can be made public without sacrificing pseudorandomness.

\begin{definition*}[\Cref{def:skPRC}]
    Let $\Sigma$ be an alphabet and $\calE : \Sigma^* \to \Sigma^*$ be a channel. A \emph{secret-key PRC} with robustness to $\calE$ is described by algorithms $\Encode_{\sk} : \Sigma^k \to \Sigma^n$ and $\Decode_{\sk} : \Sigma^* \to \Sigma^k \cup \{\bot\}$, parameterized by a secret key $\sk$, satisfying the following criteria for every security parameter $\secpar$:
    \begin{itemize}
        \item (Error correction, or robustness) For any message $\m \in \Sigma^k$,
        \[
            \Pr_{\sk}[\Decode_{\sk}(\calE(x)) = \m : x \from \Encode_{\sk}(\m)] \geq 1-\negl.
        \]
        \item (Soundness) For any fixed $c \in \Sigma^*$,
        \[
            \Pr_{\sk}[\Decode_{\sk}(c) = \bot] \geq 1 - \negl.
        \]
        \item (Pseudorandomness) For any polynomial-time adversary $\adv$,
        \[
            \abs{\Pr_{\sk}[\adv^{\Encode_{\sk}}(\secparam) = 1] - \Pr_{\calU}[\adv^{\calU}(\secparam) = 1]} \leq \negl,
        \]
        where $\adv^{\calU}$ means that the adversary has access to an oracle that, on any (even previously queried) input, responds with a freshly drawn uniform value in $\Sigma^n$.
    \end{itemize}
\end{definition*}

If the scheme can only encode a singular message (i.e. $k = 0$), then we call it a \emph{zero-bit} PRC. Soundness is a technical condition that we include only to ensure that zero-bit PRCs are non-trivial.

For a sufficiently weak channel $\calE$, it is not hard to construct a secret-key PRC with robustness to $\calE$ where pseudorandomness rests on very mild assumptions. For instance, if $F_{\sk} : \{0,1\}^\ell \to \{0,1\}$ is a pseudorandom function, we can build a zero-bit secret-key PRC with the following encoding algorithm:
\begin{enumerate}
    \item[] $\Encode_\sk(1)$:
    \item Randomly sample $x_1, \dots, x_\ell \from \{0,1\}$.
    \item For $i = \ell+1, \dots, n$, let $x_i = F_\sk(x_{i-\ell}, \dots, x_{i-1})$.
    \item Output $x_1, \dots, x_n$.
\end{enumerate}
The decoding algorithm $\Decode_\sk$ simply checks whether much more than a $1/2$ fraction of conditions $x_i = F_\sk(x_{i-\ell}, \dots, x_{i-1})$ are satisfied. Pseudorandomness follows by taking $\ell$ to be the security parameter. It is immediate that this PRC is robust to any length-preserving channel that introduces at most a $p$ fraction of errors, for $p << 1/\ell$. We call any such channel $p$-\emph{bounded}.

This secret-key PRC construction can be seen as implicit in prior hash-based watermarking schemes \cite{scott,KGW+23,CGZ23}, where essentially the same level of robustness to a $1/\ell$ error rate is obtained. Unfortunately, it has an inherent trade-off between pseudorandomness and robustness: After roughly $2^{\ell/2}$ samples from $\Encode_\sk(1)$, there will be repeated prefixes and therefore correlations between the samples. In particular, if we want this PRC to be robust to a constant rate of errors, we have to set $\ell = O(1)$, in which case even a constant number of queries are enough to observe correlations. We therefore turn to alternative constructions.

\subsection{Pseudorandom LDPC codes} \label{subsec:techo-ldpc-prcs}
Fortunately, one of the most prominent assumptions in theoretical cryptography is a statement about codes: Learning Parity with Noise (LPN). The LPN assumption states that noisy samples from the codespace of a random linear code are pseudorandom, even to an adversary who knows a generator matrix for the code.
In more detail, let $n, g$ be integers and $G \from \F_2^{n \times g}$ be a random matrix. The LPN assumption (with noise rate $\eta$ and secrets of size $g$) states that $(G, Gs \oplus e) \approx (G, u)$,\footnote{Throughout this work we use $\approx$ to refer to computational indistinguishability.} where $s \from \F_2^g$, $e \from \Ber(n, \eta)$, and $u \from \F_2^n$.

The LPN assumption suggests using noisy codewords from a random linear code as a PRC. That is, let $G \from \F_2^{n \times g}$ be the secret key and $\Encode_G$ be the following zero-bit secret-key PRC encoder:
\begin{itemize}
    \item $\Encode_G(1)$: Sample $s \from \F_2^g$ and $e \from \Ber(n,\eta)$. Output $Gs \oplus e$.
\end{itemize}
The LPN assumption immediately implies that an arbitrary polynomial number of samples from $\Encode_G(1)$ are pseudorandom. However, recall that the LPN assumption states that these samples are pseudorandom \emph{even given} $G$ --- which means precisely that there does not exist an efficient zero-bit decoder $\Decode_G(x)$!

While this random linear code construction does not work, it naturally suggests a strategy that does.
If we find a sampling procedure that produces a random (or even pseudorandom) generator matrix $G$ \emph{together with a trapdoor for efficient decoding}, then we have a public-key PRC where the generator matrix is the public encoding key and the trapdoor is the secret decoding key.
By the LPN assumption, $\Encode_G(1)$ produces pseudorandom vectors even to an adversary who knows $G$, so the construction will satisfy pseudorandomness.

It turns out that low-weight parity checks can serve as such a trapdoor. That is, instead of sampling $G$ uniformly at random, we first sample a ``parity-check matrix'' $P \in \F_2^{r \times n}$ with sparse rows (i.e., ``low density''), and then sample $G \in \F_2^{n \times g}$ subject to $PG = 0$. 
For appropriate choice of $n,g,t,r$, we will show that the resulting marginal distribution on $G$ is random or pseudorandom. The low-density parity-check matrix $P$ will allow for efficient detection of near-codewords.

Codes defined by Low-Density Parity-Check matrices are called LDPC codes. For $n, g, t, r \in \N$, we define an $(n,g,t,r)$ random LDPC code by the following distribution over parity-check and generator matrices:
\begin{enumerate}
    \item[] $\LDPC[n,g,t,r]$:
    \item Sample a random matrix $P \in \F_2^{r \times n}$ subject to every row of $P$ being $t$-sparse.
    \item Sample a random matrix $G \in \F_2^{n \times g}$ subject to $PG = 0$.
    \item Output $(P,G)$.
\end{enumerate}
Zero-bit decoding works by counting the number of satisfied parity checks.
For any fixed $x \in \F_2^n$, with high probability over $P \in \F_2^{r \times n}$ we expect that the number of unsatisfied parity checks, $\wt(Px)$, is roughly $r/2$.
But if $x$ is close to $\img G \subseteq \ker P$ in Hamming distance, then as long as the error and the sparsity $t$ of the parity checks are not too high, we expect $\wt(Px)$ to be significantly smaller than $r/2$.

Therefore our zero-bit pseudorandom LDPC code uses the following zero-bit decoding algorithm:
\begin{itemize}
    \item $\Decode_P(x)$: If $\wt(Px) < (1/2-r^{-1/4}) \cdot r$, output 1; otherwise output $\bot$.\footnote{Throughout this work, $\wt$ will refer to Hamming weight.}
\end{itemize}
Encoding is exactly the same $\Encode_G$ algorithm as above --- but now that $G$ is sampled together with the trapdoor $P$, we have an efficient decoding algorithm.
Observe that this is a zero-bit scheme, because the decoder only determines whether the input is related to the keys or not.
Using belief propagation, it is possible to push this construction beyond a zero-bit PRC, although this results in lower robustness.
We ultimately construct a constant-rate multi-bit PRC by other means, which we discuss in \Cref{subsec:constant-rate-prcs}.

Let $\LDPCPRC_0$ be the zero-bit public-key PRC defined by $(\Encode_G, \Decode_P)$ for $(P,G) \from \LDPC[n,g,t,r]$.
In a moment we will outline our proofs that $\LDPCPRC_0$ is a public-key PRC with very strong robustness. First, let us see some restrictions on the sparsity parameter $t$ that provide important context for these proofs.

If random noise of rate $1/2-\varepsilon$ is applied to $x \in \img G$, then the probability of each parity check being satisfied for the noisy codeword is $1/2-(2\varepsilon)^t/2$.
So in order for $\Decode_P(x)$ to output 1 with high probability, we need $(2\varepsilon)^t/2 > r^{-1/4}$, i.e., $t < 1 + \log r / 4 \log(1/2\varepsilon) = O(\log r)$. We will always have $r = n^{\Omega(1)}$, so this restriction says that $t = O(\log n)$ for appropriate choice of constant.

On the other hand, if we set $t = O(1)$ then $\Encode_G(1)$ cannot be pseudorandom. This is because it is possible to brute-force search over all $\binom{n}{t}$ possible parity checks of weight $t$, and one can test whether $\Encode_G$ is consistent with a given parity check $s \in \F_2^n$ by simply computing $s \cdot x$ for many samples $x \from \Encode_G(1)$.

Therefore, we will choose $t = \Theta(\log n)$ in order to rely on the weakest possible cryptographic assumption for pseudorandomness, without sacrificing robustness to a constant noise rate.

\begin{remark*}
    The LDPC codes considered in this work differ from the traditional Gallager's LDPC ensemble in two important ways. First, our LDPC codes will have $t = \Theta(\log n)$ sparsity as opposed to constant sparsity. Unfortunately, the usual belief propagation decoder does not work for noise rates beyond $O(\log t / t)$; this is the reason why we only perform the simple zero-bit decoding. The second difference is that we use independent parity checks, which results in an irregular Tanner graph.
\end{remark*}

\begin{remark*}
    There is a well-known public-key encryption scheme, due to Alekhnovich \cite[Cryptosystem 1]{Alekhnovich03}, based on a low-noise variant of LPN. This scheme is similar to ours, but the decoder cannot tolerate any constant rate of errors.
\end{remark*}

\paragraph{Pseudorandom generator matrix (\Cref{lemma:xor-hybrid}).}
For appropriate choices of parameters, it turns out that the generator matrix of $\LDPC[n,g,t,r]$ is pseudorandom under the planted $t$-XOR assumption. The planted $t$-XOR problem (and its generalization, the planted $t$-SUM problem) is a natural and well-studied problem --- see e.g. \cite{ASSVV23} for a more detailed discussion. Formally, the $(n, m, t)$ \emph{planted XOR problem} states that it is computationally hard to distinguish between
\begin{itemize}
    \item[$\calD_0^m$:] a random $m$-dimensional linear subspace $V \subseteq \F_2^n$, and
    \item[$\calD_1^m$:] a random $m$-dimensional linear subspace $V_s \subseteq \F_2^n$ satisfying a random planted $t$-XOR relation $s$ (i.e., $s$ is a random $t$-sparse vector and $s \cdot v = 0$ for all $v \in V_s$).
\end{itemize}
Throughout this overview, we consider linear subspaces to be described by a random basis. Recalling the definition of $(P,G) \from \LDPC[n,g,t,r]$, if $r = 1$ the $(n, g, t)$ planted XOR assumption immediately implies that $G$ is pseudorandom (by identifying $V$ with $\img G$).
But for the more interesting case that $r > 1$, we require a stronger version of the planted XOR assumption with many planted relations. That is, we need to assume that the following distribution is indistinguishable from $\calD_0^m$:
\begin{itemize}
    \item[$\calD_r^m$:] a random $m$-dimensional linear subspace $V_{s_1, \dots, s_r} \subseteq \F_2^n$ satisfying $r$ random planted $t$-XOR relations $s_1, \dots, s_r$ (i.e., $s_1,\dots,s_r$ are random $t$-sparse vectors and $s_1 \cdot v = \cdots = s_r \cdot v = 0$ for all $v \in V_{s_1, \dots, s_r}$).
\end{itemize}
We are not aware of any prior work on this assumption that $\calD^m_0 \approx \calD^m_r$, so it is not immediately clear how reliable it is. Fortunately, it is \emph{implied} by the $(n,m+r,t)$ planted XOR assumption. 

We prove this with a hybrid argument.
Suppose that an efficient adversary $\adv$ distinguishes between $\calD_0^m$ and $\calD_r^m$ with advantage $\varepsilon > 0$.
By a telescoping argument, $\adv$ must distinguish between $\calD_i^m$ and $\calD_{i+1}^m$ with advantage $\varepsilon/r$, for some $i \in \{0, \dots, r-1\}$.
For each $i$, the following efficient reduction $\Red_i$ satisfies $\Red_i(\calD_0^{m+r}) \equiv \calD_i^m$ and $\Red_i(\calD_1^{m+r}) \equiv \calD_{i+1}^m$, which implies that $\varepsilon/r$ (and therefore $\varepsilon$) is negligible under the $(n,m+r,t)$ planted XOR assumption.

\begin{enumerate}
    \item[] $\Red_i(W)$:
    \item Sample $i$ random $t$-sparse vectors $s_1, \dots, s_i \in \F_2^n$ and let $S = \{v \in \F_2^n : v \cdot s_j = 0\ \forall j \in [i]\}$.
    \item Let $U = W \cap S$. Notice that $\dim U \ge \dim W - i$. Since $\dim W = m+r$ and $i < r$, this is at least $m$.
    \item Output a random $m$-dimensional subspace of $U$.
\end{enumerate}
It remains to see why $\Red_i(\calD_0^{m+r}) \equiv \calD_i^m$ and $\Red_i(\calD_1^{m+r}) \equiv \calD_{i+1}^m$.
In fact both of these statements are true even for \emph{fixed} planted relations.

\begin{itemize}
    \item $\Red_i(\calD_0^{m+r}) \equiv \calD_i^m$. Fix $s_1, \dots, s_i$ sampled in $\Red_i$ and let $S = \{v \in \F_2^n : v \cdot s_j = 0\ \forall j \in [i]\}$. Since $W$ is a random subspace of $\F_2^n$, conditioned on any $d = \dim(W \cap S)$, $U = W \cap S$ is a random $d$-dimensional subspace of $S$. Therefore the output of $\Red_i(\calD_0^{m+r})$ is a random $m$-dimensional subspace of $S$.
    \item $\Red_i(\calD_1^{m+r}) \equiv \calD_{i+1}^m$. Fix $s_1, \dots, s_i$ sampled in $\Red_i$ and let $S = \{v \in \F_2^n : v \cdot s_j = 0\ \forall j \in [i]\}$, as above. Suppose that $W \from \calD_1^{m+r}$ is sampled with the planted relation $s$. Fix $s$ and let $S' = \{v \in S : v \cdot s = 0\}$. Again, conditioned on any $d = \dim(W \cap S)$, $U = W \cap S = W \cap S'$ is a random $d$-dimensional subspace of $S'$. Therefore the output of $\Red_i(\calD_1^{m+r})$ is a random $m$-dimensional subspace of $S'$.
\end{itemize}

\paragraph{Narrow, statistically random generator matrix (\Cref{lemma:low-rate-random-ldpc}).}
Since the planted XOR assumption is not a standard cryptographic assumption, we show in \Cref{lemma:low-rate-random-ldpc} that the generator matrix of $\LDPC[n, c \log^2 n, \log n, 0.99 n]$ is \emph{statistically} random for some $c > 0$. This removes the need for the planted XOR assumption, but it comes at the cost of requiring a stronger version of the LPN assumption: When we invoke LPN to see that samples $(G, Gs \oplus e)$ are pseudorandom, the secrets $s$ are now only of size $c \log^2 n$. Therefore, for this PRC we will rely on a subexponential version of the LPN assumption which states that LPN is $2^{O(\sqrt{n})}$-hard.

For the purposes of this technical overview, we will show that the generator matrix of a closely related code is random. The proof for this version is significantly simpler, but the distribution is less natural and has worse error-correcting properties. The modified distribution on $(P,G)$ is defined as follows:
\begin{enumerate}
    \item Sample a uniformly random $G_0 \from \F_2^{0.01 n \times g}$.
    \item For $i \in [0.99n]$:
    \begin{enumerate}
        \item Sample a random $(t-1)$-sparse $s_i \in \F_2^{0.01 n}$.
        \item Let $G_i$ be the matrix $G_{i-1}$ with the extra row $s_i^T G_0$ appended to the bottom, 
        \[
            G_i = \begin{bmatrix}
                G_{i-1} \\
                s_i^T G_0
            \end{bmatrix}.
        \]
        \item Let $s_i' = [s_i^T, 0^{i-1}, 1, 0^{0.99n-i}]$. 
    \end{enumerate}
    \item Let $P$ be the matrix whose rows are $s_1', \dots, s_{0.99n}'$ and $G = G_{0.99n}$. Output $(P, G)$.
\end{enumerate}
First observe that $PG = 0$, because $s_i' G = [s_i^T, 0^{i-1}, 1, 0^{0.99n-i}] G = s_i^T G_0 \oplus (s_i^T G_0) = 0$ for every $i \in [0.99n]$.

The leftover hash lemma immediately implies that $G$ is statistically random. Recall that the leftover hash lemma states that if $A \from \F_2^{g \times \ell}$ is a uniformly random matrix and $s \in \F_2^\ell$ has min-entropy $\mu$, then $(A, As)$ is $2^{-(\mu-g)/2}$-close to uniform in statistical distance. In our case, we use $A = G_0^T$ and $s = s_i$ to see that $s_i^T G_0$ is $2^{-(\log \binom{0.01n}{t} - g)/2}$-close to uniform in statistical distance for each $i \in [n-g]$.
If $t = \Omega(\log n)$, then there is a choice of $g = \Omega(\log^2 n)$ such that $2^{-(\log \binom{0.01n}{t} - g)/2} = 2^{-\Omega(\log^2 n)} = \negl[n]$, completing the proof.

\paragraph{$\LDPCPRC_0$ is robust to any $p$-bounded channel (\Cref{lemma:zero-bit-ldpc-decoding}).} Recall that we say that any length-preserving channel that introduces at most a $p$ fraction of bit-flip errors is $p$-bounded. To prove robustness, we need to show two things:
\begin{enumerate}
    \item any fixed $x \in \F_2^n$ decodes to $\bot$ with high probability, and
    \item for any $p$-bounded channel $\calE$, samples from $\calE(\Encode_G(1))$ decode to 1 with high probability.
\end{enumerate}
Unfortunately, (1) does not quite hold for the scheme presented above. The issue is that, while \emph{most} fixed strings will decode to $\bot$, a small fraction of strings will decode to 1 regardless of $P$. For instance, 0 will always decode to 1 because $\wt(P \cdot 0) = 0$ for any $P$.
Therefore we modify our scheme slightly by using a one-time pad $z \in \F_2^n$, included in the public key. The modified $\Encode_G(1)$ outputs $Gs \oplus e \oplus z$, and $\Decode_P(x)$ computes $\wt(Px \oplus Pz)$ instead of $\wt(Px)$; this is the actual scheme we describe in \Cref{sec:ldpc-prcs}.

Now, for any fixed $x \in \F_2^n$, $\wt(Px \oplus Pz)$ is distributed identically to $\wt(Pz)$ because $z$ is uniform. In \Cref{subsec:subexp-lpn} we will see that $P$ is full rank with high probability, so $Pz$ is uniformly random. By a Chernoff bound, $\wt(Pz) \ge (1/2 - r^{-1/4}) \cdot r$ with high probability and therefore $\Decode_P(x)$ outputs $\bot$. This concludes the proof of (1), so we turn to (2).

Suppose that we sample $Gs \oplus e \oplus z \from \Encode_G(1)$ and apply some $p$-bounded channel $\calE$. The one-time pad effectively converts $\calE$ to a fixed error channel, independent of $P, G, s, e$: Suppose that $\calE(x) = x \oplus e(x)$, where $e(x)$ is a random variable depending on $x$. Since $\calE$ is $p$-bounded, $\wt(e(x)) \le p \cdot n$. Letting $y = Gs \oplus e \oplus z$, we have
\begin{align*}
    \calE(Gs \oplus e \oplus z) \oplus z &= (Gs \oplus e \oplus z) \oplus e(Gs \oplus e \oplus z) \oplus z \\
    &= Gs \oplus e \oplus e(y)
\end{align*}
where $y$ is uniformly random in $\F_2^n$, independent of $P, G, s, e$ because of $z$. Now it only remains to see that $\wt(P(Gs \oplus e \oplus e(y))) = \wt(P(e \oplus e(y))) < (1/2-r^{-1/4}) \cdot r$ with high probability. Since $e$ and $e(y)$ are independent errors, each of weight $(1/2-\Omega(1)) \cdot n$, the combined error $e \oplus e(y)$ also has weight $(1/2-\Omega(1)) \cdot n$. Therefore, if the row sparsity $t$ of $P$ is $c \log n$ for sufficiently small constant $c$, then we will have $\wt(P(e \oplus e(y))) < (1/2-r^{-1/4}) \cdot r$ with high probability, completing the proof of (2).

\subsection{Pseudorandom codes for the deletion channel} \label{subsec:techo-deletion-prcs}
So far, we have only considered PRCs for substitution channels. For our applications to watermarking and steganography, it will be useful to have PRCs for the noisy deletion channel as well.
The noisy deletion channel randomly introduces both deletions and substitutions.

Unfortunately, existing error-correcting codes for the deletion channel introduce a large amount of structure into codewords that precludes pseudorandomness.
For instance, the popular techniques of synchronization symbols or concatenation with constant-sized inner codes are immediately seen to be incompatible with pseudorandomness.
Even further limiting the techniques available to us, we want our PRCs for the noisy deletion channel to have a binary alphabet in order to be useful for watermarking.

We therefore turn to alternative techniques.
Surprisingly, we find that the repetition code --- perhaps the simplest and most-structured error-correcting code --- is a useful starting point.

For odd integer $T$, the rate-$(1/T)$ repetition code works by repeating each bit of the message $T$ times.
That is, for any message $\m = (\m_1 || \cdots || \m_k) \in \{0,1\}^k$, the encoder $\RepEncode_T$ is defined by $\RepEncode_T(\m) = (\m_1)^T || \cdots || (\m_k)^T$, where $(\m_i)^T$ denotes bit $\m_i$ repeated $T$ times.
For example, the rate-$(1/3)$ repetition code encodes $010$ as $\RepEncode_T(010) = 000111000$.

Now suppose that the encoding $(\m_1)^T || \cdots || (\m_k)^T$ is subjected to the noisy deletion channel, resulting in a string $z$.
A natural algorithm for decoding $z$ is to partition $z$ into $k$ equal-length blocks $z_1, \dots, z_k$, and compute the majority of each block:

\indent $\MajDecode_k(z)$:
\begin{enumerate}
    \item Partition $z$ into $k$ equal-length blocks $z = (z_1 || \cdots || z_k)$.
    \item Output $(\Maj(z_1) || \cdots || \Maj(z_k))$.
\end{enumerate}

As long as the deletions are sufficiently balanced across the different blocks, the $z_i$ will align well with the original blocks $(\m_i)^T$.
Provided further that there are not too many substitutions in any block, we should have $\MajDecode_k(z) = \m$.
The issue is that $\RepEncode_T(\m)$ is not pseudorandom even for random $\m$, because a random string is (extremely) unlikely to consist of $T$ repeated bits.

On the other hand, a random string typically does have $\Theta(\sqrt{T})$ bias towards 0 or 1.\footnote{That is, a random string has $\Theta(\sqrt{T})$ more 0's than 1's, or 1's than 0's. This can be seen as a consequence of the fact that a one-dimensional simple random walk of length $T$ will usually terminate $\Theta(\sqrt{T})$ away from the origin.}
So if we change or delete a small $O(\sqrt{T})$ number of bits of a random string, we expect the majority to stay the same.
This observation brings us to the following encoder $\MajEncode_T$, which encodes each bit in the majority of a random string. We refer to the code defined by $(\MajEncode_T, \MajDecode_k)$ as the \emph{majority code}.

\indent $\MajEncode_T(\m)$:
\begin{enumerate}
    \item For $i \in [k]$, let $z_i$ be a random sample from $\{0,1\}^T$ conditioned on $\Maj(z_i) = \m_i$.
    \item Output $(z_1 || \cdots || z_k)$.
\end{enumerate}

Now if $\m$ is random, then $z = \MajEncode_T(\m)$ is random as well.
Furthermore, if we subject $z$ to the noisy deletion channel to obtain $z'$, then the bits of $\m' = \MajDecode_k(z')$ are positively correlated with the bits of $\m$.
This is because the deletions are at random locations, and are therefore (roughly) evenly-distributed across the different blocks $z_i$ --- meaning that $\MajDecode_k$ will mostly use the correct locations to decode each bit.
Since the bit-flip errors are random, they merely dilute the $\Theta(\sqrt{T})$ biases.
If $T >> k$ and the rates of deletions and bit-flip errors are constants below 1 and $1/2$ respectively, then we show in \Cref{lemma:deletion-channel} that $\Pr[\m_i = \m_i']$ is a constant greater than $1/2$.
Therefore, the majority code has the effect of converting the constant-rate noisy deletion channel into some $p$-bounded channel.

Of course, the majority code is not itself a PRC.
The first problem is that codewords for the majority code are only random if the message is random, whereas a PRC needs to allow encoding of any particular message.
The second problem is that, even if the message is random, the majority code recovers a string that is only \emph{correlated} with it.

But these are exactly the problems solved by PRCs for bounded-weight error channels!
That is, if $\PRC$ is any PRC with robustness to every $p$-bounded channel (e.g. the PRCs from \Cref{subsec:techo-ldpc-prcs}), then the combined code $\PRCdel = (\MajEncode \circ \PRC.\Encode, \PRC.\Decode \circ \MajDecode)$ is a PRC with robustness to some constant-rate noisy deletion channel.\footnote{As $p$ approaches $1/2$, the combined PRC tolerates a noisy deletion channel with rates of deletions and bit-flip errors approaching 1 and $1/2$.}
Pseudorandomness follows from the pseudorandomness of $\PRC.\Encode$: Since $\PRC.\Encode(\m)$ is pseudorandom for any message $\m$, $\MajEncode(\PRC.\Encode(\m))$ is as well.
Robustness follows from the fact that the majority code has the effect of converting the constant-rate noisy deletion channel into some $p$-bounded channel, which is handled by $\PRC.\Decode$.

\subsection{Constant-rate pseudorandom codes}
So far we have only considered zero-bit PRCs, but for many applications it will be useful to have PRCs that can encode longer messages.
There is a simple construction of a multi-bit PRC directly from any zero-bit PRC: Encode each bit of the message with either a zero-bit PRC codeword, or a uniformly random string.
That is, if $\PRC$ is a zero-bit PRC, we encode a message $\m \in \{0,1\}^k$ as $(x_1 || \cdots || x_k)$ where for each $i \in [k]$, $x_i \from \{0,1\}^n$ if $\m_i = 0$ and $x_i \from \PRC.\Encode(1)$ if $\m_i = 1$.

Unfortunately this scheme has a very low rate.
If the zero-bit PRC has block length $n$, then the resulting multi-bit PRC has rate $1/n$.
However, we show in \Cref{subsec:constant-rate-prcs} that one can use any such low-rate PRC to make any error-correcting code pseudorandom, with no asymptotic loss in rate.

The idea is to encode a seed for a one-time pad in the simple low-rate PRC just described, and then use the one-time pad to hide an error-correcting encoding of the message.
More formally, let $(\Enc, \Dec)$ be any (standard) error-correcting code and $\PRC$ be a low-rate PRC.
We do not require $(\Enc, \Dec)$ to have any cryptographic properties.
We encode a message $\m$ as\footnote{In order to obtain stronger robustness guarantee, we actually randomly permute the symbols of this encoding. For the purposes of this technical overview we omit this detail.}
\[
    \PRC.\Encode(r), \Enc(\m) \oplus \PRG(r),
\]
where $\PRG$ is any pseudorandom generator and $r \from \{0,1\}^k$ is a fresh uniformly random string.

By pseudorandomness of $\PRC$, $\PRC.\Encode(r)$ is indistinguishable from a uniformly random string --- even for a fixed choice of $r$.
By security of $\PRG$, the encoding is therefore indistinguishable from a totally random string.

Decoding works as long as $\PRC.\Encode(r)$ is not too corrupted for $\PRC$ to recover $r$, and $\Enc(\m) \oplus \PRG(r)$ is not too corrupted for $\Dec$ to recover $\m$.

\subsection{Watermarks for language models} \label{subsec:techo-water}
In this work, a generative language model is formally described by an algorithm $\Model$ that takes as input a prompt $\prompt$ and a sequence of tokens output thus far $\t_1, \ldots, \t_{i-1}$, and produces a distribution over the next token.
A full response is generated by iteratively sampling from these distributions, at each step providing $\Model$ with the partial response, and terminating once a special ``done'' token is sampled.
For simplicity, we assume here that tokens are binary, which allows us to specify the distribution $\p_i$ over the next token as a Bernoulli distribution $\Ber(\hat{\p}_i)$ where $\hat{\p}_i := \E[\p_i] \in [0,1]$.
We also assume for the purposes of this technical overview that the model always generates a response of length $n$.
In \Cref{sec:water} we explain why neither of these assumptions is important for our results.

As defined in \cite{CGZ23}, a watermarking scheme for a language model consists of algorithms $\Wat$ and $\Detect$, where $\Wat$ is the watermarked model and $\Detect$ is an algorithm used to detect the presence of the watermark.
In this work we are interested in watermarks that are \emph{undetectable}, \emph{sound}, and \emph{robust}, loosely defined as follows.
\begin{itemize}
    \item \emph{Undetectability}: Any polynomial number of responses from the watermarked model are computationally indistinguishable from those of the original model.
    \item \emph{Soundness}: Text generated independently of the watermarked model is not falsely detected.
    \item \emph{Robustness}: Sufficiently high-entropy text output by the model is detected as watermarked, even if it is altered.
\end{itemize}

We show that the watermarking strategy from \Cref{sec:intro}, which replaces some of the model's randomness with PRC codewords, yields a scheme that satisfies all of the above properties.

\paragraph{Defining $\Generate$ for language models.}
Recall that the approach from \Cref{sec:intro} requires an algorithm $\Generate$ that takes as input a prompt and a random seed $x \in \{0,1\}^n$, and samples a response $\t \in \{0,1\}^n$ such that
\begin{enumerate}[label={(\arabic*)}]
    \item\label{item:generate-correctness} if $x$ is uniformly random, then $\t$ is distributed identically to a response from $\Model$, and
    \item\label{item:generate-bias} each bit of $\t$ is correlated with the corresponding bit of $x$.
\end{enumerate}
We define $\Generate(\prompt, x)$ to sample the $i^{\text{th}}$ bit $\t_i$ of the response as follows. It first computes $\p_i$ by querying $\Model$ with $\prompt$ and the response output thus far, then:
\begin{itemize}
    \item If $\hat{\p}_i \le 1/2$, sample $\t_i \from \Ber(2 x_i \hat{\p}_i)$.
    \item If $\hat{\p}_i > 1/2$, sample $\t_i \from \Ber(1-2(1-x_i)(1-\hat{\p}_i))$.
\end{itemize}
For any $\hat{\p}_i \in [0,1]$, one can easily see that $\t_i$ is distributed as $\Ber(\hat{\p}_i)$ since $x_i$ is a uniformly random bit.
This means that $\Generate$ satisfies Condition \ref{item:generate-correctness} above.

For Condition \ref{item:generate-bias}, the bias toward the seed is stronger the closer $\hat{\p}_i$ is to $1/2$.
It is strongest when $\hat{\p}_i = 1/2$ exactly, in which case $\t_i$ is sampled from $\Ber(x_i)$ and is therefore equal to $x_i$.
At the other extreme, if $\hat{\p}_i = 0$ or $1$, there is no bias.
In general the response is a noisy version of the seed, where the amount of noise on the $i^{\text{th}}$ token decays as the binary entropy of $\Ber(\hat{\p}_i)$ grows.

\paragraph{Replacing seeds with PRC codewords.}
We use PRC samples $x \from \PRC.\Encode(1)$, instead of random samples $x \from \{0,1\}^n$, as the seeds in $\Generate$.
That is, if $\PRC$ is a zero-bit PRC, we let our watermarking scheme $\calW[\PRC]$ be defined by
\begin{itemize}[leftmargin=3cm]
    \item[$\Wat(\prompt)$:] Sample $x \from \PRC.\Encode(1)$ and output a sample from $\Generate(\prompt, x)$.
    \item[$\Detect(\t)$:] Compute $\PRC.\Decode(\t)$ and output the result.
\end{itemize}
By Condition \ref{item:generate-correctness} and the pseudorandomness property of $\PRC$, the responses from $\Wat$ are computationally indistinguishable from those of the original model.
By Condition \ref{item:generate-bias} and the robustness property of $\PRC$, the watermark will be detectable as long as the PRC is sufficiently powerful.

\begin{remark*}
    Depending on the kind of robustness of the PRC, substituting the entire seed with a single PRC sample results in a watermark that may or may not be detectable from just a subsequence.
    This is easily fixed by using $x = (x_1 || \cdots || x_m)$, where $x_i \from \PRC.\Encode(1)$ are independent PRC samples that are much shorter than the generated content.
    As long as the text contains \emph{at least one} subsequence corresponding to a PRC sample $x_i$, the watermark will be detected.
\end{remark*}

\paragraph{PRC error correction and watermark robustness.}
To understand how error correction of $\PRC$ translates to robustness of $\calW[\PRC]$, it is helpful to think of $\Generate$'s sampling process as a noisy \emph{embedding channel} applied to the seed.
That is, for a seed $x \in \{0,1\}^n$, let $\Eemb(x) = \Generate(\prompt, x)$ be the ``embedding channel'' describing the noise in $x \mapsto \t$.
For detection, it is sufficient for $\PRC$ to correct against the channel $\Eemb$, since watermarked responses are exactly samples from $\Eemb(x)$ for $x \gets \PRC.\Encode(1)$. 

Robustness of $\calW[\PRC]$ is determined by $\PRC$'s ability to correct from additional errors on top of $\Eemb$.
Let $\Eadv$ be a channel modeling the changes an adversary introduces to a watermarked response, so the overall error applied to $\t$ follows $\Eadv \circ \Eemb$.
If $\PRC$ is robust to $\Eadv \circ \Eemb$, the watermark is robust to this adversary's modifications.

In \Cref{subsec:water-simple}, we show that as long as the text has non-zero entropy, $\Eemb$ introduces errors at a rate of less than $1/2$.
Therefore, using any PRC with robustness to every $p$-bounded channel, we immediately obtain watermarks that are robust to a constant rate of random substitutions.

\paragraph{Robustness of our watermark.}
Hashing-based watermarking schemes --- including all existing undetectable schemes --- are removable by the simple ``emoji attack'' \cite{scott,KGW+23}.
In this attack, an adversary asks the model to respond to its prompt and insert an emoji between every word of its response.
The adversary then deletes the emojis from the response.
This attack removes any watermark that relies on the detector seeing contiguous sequences of watermarked text.

It turns out that hashing-based schemes \cite{scott,KGW+23,CGZ23} can easily be made robust to this particular attack.\footnote{For instance, we could choose to only hash tokens whose index has the same parity as the token being sampled.}
However, if the adversary instead instructs the model to insert the emojis \emph{randomly}, then we do not know how to make any hashing-based scheme robust.
By constructing PRCs with robustness to random deletion channels, we give the first undetectable watermarking scheme that can resist this kind of attack.

In order to show this, we require a stronger assumption on the response: That $\Eemb$ behaves as the binary symmetric channel $\BSC_q$ for some $q \in (0,1/2)$.
The binary symmetric channel $\BSC_q$ is the channel that flips each bit of its input independently with probability $q$, so $\BSC_q(x) = x \oplus \Ber(q)$ for $x \in \{0,1\}$.
Essentially, this is equivalent to the assumption that the response has high entropy \emph{and does not repeat words too often}.

Under this assumption, if $\Eadv = \BDC_p$ is the binary deletion channel for some $p \in (0,1)$, then we just need a PRC with robustness to $\Eadv \circ \Eemb = \BDC_p \circ \BSC_q$.
The binary deletion channel $\BDC_p$ is the channel that deletes each bit of its input independently with probability $p$.
Indeed, we saw in \Cref{subsec:techo-deletion-prcs} that there exist PRCs with robustness to $\BDC_p \circ \BSC_q$ for any $p \in (0,1), q \in (0,1/2)$.

\subsection{Watermarks with public attribution}
\label{subsec:techo-attribution}

For the ``standard'' notions of watermarks considered so far, the goal is to determine whether a given text is a possibly-corrupted version of an output generated by a model. 
Standard watermarks are well-suited for applications such as detecting plagiarism, where one wishes to know if a model was used at all to produce a text, even if that text has been altered by the user.

A different use of watermarks is in \emph{attributing} content to an LLM that generated it. 
For example, if harmful content generated by an LLM is found on social media, it would be useful to trace this content back to the model using the watermark.
Ideally, anyone holding a \emph{public} detection key should be able to trace the content.
On the other hand, it should only be possible to embed the watermark by using a \emph{secret} embedding key, in order to avoid falsely attributing text to any model.\footnote{Note that the roles of the public and secret keys are reversed here. For PRCs, a secret key is necessary for decoding, but anyone can encode with knowledge of a public key. Public-key PRCs are useful for public-key steganography.}
In other words, to an attacker who does not know the secret key, the watermark should be \emph{unforgeable}.

In addition to unforgeability, watermarks with public attribution have subtly different detection properties than standard watermarks.
Robustness of standard watermarks means that an LLM-generated text will be detected even if small modifications are made.
If robust watermarks were used for attribution, an attacker could use a model to generate a benign watermark text, then change a few words to make it offensive.
By robustness, the watermark would still be present in this now-offensive content.
So whereas robustness is a useful feature for standard watermarking, in the context of attribution it is actually an issue.

We therefore define a watermark with public attribution to have a separate detection algorithm called $\ForgeDetect$, which is intentionally designed to \emph{not} be robust.
$\ForgeDetect$, given a text $x$ and a public detection key, indicates whether the model output verbatim a significant part of that text, and outputs that portion of the text if so.

In order to preserve the benefits of robust watermarking for applications like detecting plagiarism, our publicly attributable watermarks also retain a $\Detect$ algorithm (in addition to the $\ForgeDetect$ algorithm) with the robustness of our standard watermarking schemes.
One can choose at detection time whether one wants to use $\Detect$ for standard detection, or $\ForgeDetect$ for attribution.

Our watermarking scheme with public attribution, $\Watt[\PRC]$, is a natural extension  of our regular watermarking scheme $\calW[\PRC]$.
Recall that $\calW[\PRC]$ embeds a codeword of a zero-bit PRC into the model's response; the detector checks whether the given text is close to a codeword.
Of course, if we use a PRC that encodes an arbitrary message (rather than only `1' as in a zero-bit PRC), then $\calW[\PRC]$ will embed arbitrary messages in the text.
$\Watt[\PRC]$ does exactly this, where the message that it encodes is a \emph{signature on the response output thus far}.
$\ForgeDetect$ decodes the given text to obtain this signature, and checks using the public detection key that it is a valid signature of a portion of the response.
If so, this signed portion must have been generated by the model.

Concurrent work \cite{public} also constructs a watermark with public detection, although this scheme is designed to have mild robustness (comparable to that of \cite{CGZ23}) and therefore is not appropriate for attribution as-is. 
Their scheme can easily be modified to satisfy our definition of unforgeable public attribution, but it would then lose all robustness guarantees for standard watermarking.
Our scheme simultaneously functions as a highly robust standard watermark via $\Detect$, while also satisfying unforgeable public attribution via $\ForgeDetect$.

\subsection{Robust steganography} \label{subsec:techo-stego}
In steganography, the goal is to send a hidden message such that an observer cannot tell that a message is being sent at all. 
In the classic presentation, a prisoner wishes to secretly communicate with an outside party even though the warden is filtering their letters.
If the warden detects any unusual language then the communication channel will be shut down, so the prisoner cannot simply encrypt the message: The warden should not only be unable to learn anything about the message, but should be unable to even detect that secret communication is occurring at all.

Steganography was formalized in \cite{HLvA02}. In this presentation, there is some underlying \emph{steganographic channel},\footnote{In the steganography literature this is usually just called a ``channel''; we call it a ``steganographic channel'' to differentiate it from the coding-theoretic channels we use in the context of robustness.} a distribution with which the sender wishes to conceal a message. 
The sender is given sample access to this steganographic channel and sends a \emph{stegotext} to the receiver.
\emph{Steganographic secrecy} requires that the distribution of stegotexts is indistinguishable from the steganographic channel, except to the receiver who can recover the message with a secret key.

\cite{HLvA02} proves the security of a steganography scheme of \cite{And98} that can be constructed using any encryption scheme $(\Encode,\Decode)$ with pseudorandom ciphertexts.
The key idea behind this scheme is to embed each bit $x_i$ of a pseudorandom encryption of the message by drawing a sample $d_i$ from the steganographic channel such that $f(d_i) = x_i$ for some hash function $f$:

\begin{enumerate}
    \item[] $\Steg.\Encode(\sk,\m)$ \cite{And98,HLvA02}:
    \item Let $x = x_1 || \ldots || x_n \gets \Encode(\sk,\m)$
    \item For $i \in [n]$, sample a random $d_i$ from the channel conditioned on $f(d_i) = x_i$
    \item Output $d = d_1 || \ldots || d_n$
\end{enumerate}

The decoder $\Steg.\Decode$ simply outputs $\Decode(\sk, f(d_1) || \ldots || f(d_n)) = \Decode(\sk, x) = \m$.

If $x_i$ is uniform over $\{0,1\}$, and $f$ is perfectly unbiased for the channel, then $d_i$ is sampled exactly from the channel distribution.
Therefore, by pseudorandomness of the ciphertext, an observer cannot distinguish stegotexts from samples from the steganographic channel.
The receiver, which knows the decoding key for the encryption scheme, can evaluate $f$ on each block of the stegotext to obtain the ciphertext, then decrypt to recover the message.

However, this scheme is very brittle. If the stegotext is corrupted with any errors at all --- even ones resulting from any small bias in the hash function $f$ --- the message cannot be recovered.
A natural attempt at achieving robustness is for the sender to apply an error-correcting code $(\Enc, \Dec)$ to the ciphertext before embedding it. 
But this loses pseudorandomness and therefore steganographic secrecy!
Consequently, the robust steganography schemes of prior work rely on stronger assumptions like the ability of the sender and receiver to share state. A shared state allows the sender to generate a fresh one-time pad $r$ for each message it sends, making the task much easier: The sender embeds $x = \Enc(\m) \oplus r$ by choosing $d$ such that $f(d_i) = x_i$, and the receiver computes $\Dec(\tilde{x} \oplus r)$, where $(\Enc,\Dec)$ is any error-correcting code.
However, if the sender and receiver become unsynchronized (as is likely in practice), the receiver can no longer decode the message.

We observe that PRCs are exactly the primitive needed for robust stateless steganography: Using a PRC as the pseudorandom encryption scheme in the above $(\Steg.\Encode, \Steg.\Decode)$ construction \emph{immediately} gives us a steganography scheme with the same robustness as the PRC.
If we use a public-key PRC, the resulting steganography scheme is also public-key.
Furthermore, the robustness of the PRC allows us to relax the assumption that $f$ is perfectly unbiased on the steganographic channel.

Our main result of \Cref{sec:stego} is the first stateless steganography scheme with nontrivial robustness to errors. 
In particular, using our LDPC-based PRCs, we construct stateless steganography schemes that are robust to $p$-bounded channels for any constant $p$, or any constant-rate random deletion channel.

\newpage

\section{Preliminaries} \label{sec:prelims}
Let $\N := \{1, 2, \dots\}$ denote the set of positive integers. We will write $[q] := \{1, \dots, q\}$. For a set $X$, we define $X^* := \{(x_1, \dots, x_k) \mid x_1, \dots, x_k \in X \wedge k \in \Z_{\ge 0}\}$ to be the set of all strings with alphabet $X$. For a binary string $s \in X^*$, we let $s_i$ denote the $i^{\text{th}}$ symbol of $s$ and $\len s$ denote the length of $s$. For a string $s \in X^*$ and positive integers $a \leq b \le \len s$, let $s[a:b]$ denote the substring $(s_a, \ldots, s_b)$.

For a finite set $X$, we will use the notation $x \from X$ to denote a uniformly random sample $x$ from $X$. If $X$ is a set of $n$-dimensional column vectors, we will write $X^m$ to refer to the set of $n \times m$ matrices whose columns take values in $X$. Unless otherwise specified, vectors are assumed to be column vectors. For matrices $M \in \F_2^{m \times n}$, let $\ker M$ and $\img M$ denote the kernel and image of $M$ over $\F_2$, respectively.

We use $\log(x)$ to denote the logarithm base 2 of $x$, and $\ln(x)$ to denote the natural logarithm of $x$.

Let $\Ber(p)$ be the Bernoulli distribution on $\{0,1\}$ with expectation $p$. Let $\Ber(n, p)$ be the distribution on $n$-bit strings where each bit is an i.i.d sample from $\Ber(p)$.

We let $\circ$ denote the composition of functions, algorithms, or channels; that is, $f \circ g$ denotes (the function/algorithm/channel) obtained by applying $g$, then $f$.

Let $\secpar$ denote the security parameter. 
A function $f$ of $\secpar$ is \emph{negligible} if $f(\secpar) = O(\frac{1}{\poly})$ for every polynomial $\poly[\cdot]$.
We write $f(\secpar) \leq \negl$ to mean that $f$ is negligible.
We let $\approx$ denote computational indistinguishability and $\equiv$ denote statistical indistinguishability.

\begin{lemma}[Azuma's inequality] \label{theorem:azuma}
    Let $Z_0, \dots, Z_n$ be a martingale with respect to $X_0, \dots, X_n$. If the differences $\abs{Z_i - Z_{i-1}}$ are all bounded by $C$ with probability $1-\varepsilon$, then for all $t > 0$,
    \[
        \Pr[\abs{Z_n - Z_0} > t] \le \exp{\left(\frac{-t^2}{2n C^2}\right)} + \varepsilon.
    \]
\end{lemma}

\begin{lemma}[Chernoff bounds] \label{theorem:chernoff}
    Let $X_1, \ldots, X_n \in [0,1]$ be independent random variables. Let $\mu = \E\left[\sum_{i=1}^n X_i \right]$. Then for any $\delta \in (0,1)$:
    \begin{align*}
        \Pr \left[\sum_{i=1}^n X_i \geq (1+\delta)\mu \right] &\leq \exp \left(-\frac{\mu \delta^2}{3} \right) \text{ and}\\
        \Pr \left[\sum_{i=1}^n X_i \leq (1-\delta)\mu \right] &\leq \exp \left(-\frac{\mu \delta^2}{2} \right).
    \end{align*}
\end{lemma}

Let $\Hyp(N, K, n)$ denote the hypergeometric distribution with a population of size $N$, with $K$ success elements, and $n$ draws.
That is, a random variable $X \sim \Hyp(N, K, n)$ is the number of success elements contained in $n$ uniform draws from the population without replacement.

\begin{lemma}[Hypergeometric tail bounds \cite{hoeffding1994probability}] \label{theorem:hyper}
Let $X \sim \Hyp(N,K,n)$ and $p = K/N$. Then for any $0 < t < K/N$,
\begin{align*}
    \Pr \left[X \leq (p-t)n \right] &\leq e^{-2t^2 n}, \text{ and}\\
    \Pr \left[X \geq (p+t)n \right] &\leq e^{-2t^2 n}.
\end{align*}
\end{lemma}

\subsection{Cryptography preliminaries}
\paragraph{Pseudorandom function (PRF).} Let $\mathcal{F} = \{F_{\sk} : \{0,1\}^{\ell_1(\secpar)} \to \{0, 1\}^{\ell_2(\secpar)} \ | \ \sk \in \{0,1\}^\secpar\}$ be a family of functions. 
$\mathcal{F}$ is a PRF if $F_\sk$ is efficiently computable and for all polynomial-time distinguishers $D$,
\[\left| \Pr_{\sk \gets \{0,1\}^\secpar}\left[D^{F_\sk(\cdot)}(1^\secpar) = 1\right] - \Pr_{f} \left[D^{f(\cdot)}(1^\secpar) = 1 \right]\right| \leq \negl.\]
where $f$ denotes a random function from $\{0,1\}^{\ell_1(\secpar)}$ to $\{0,1\}^{\ell_2(\secpar)}$.

\paragraph{Digital signature scheme.}
We use the definition of a digital signature from \cite{katzlindell}, with small modifications. 
A \emph{digital signature scheme} is defined over a message space $\calM$ and consists of polynomial-time algorithms $(\Gen, \Sign, \Vrfy)$ such that:
\begin{itemize}[leftmargin=1.5cm]
    \item[$\Gen:$] takes as input a security parameter $1^\secpar$ and outputs a public-private key pair $(\pk, \sk)$.
    \item[$\Sign:$] takes as input a private key $\sk$ and a message $\m \in \calM$. It outputs a signature $\sigma$, which we write as $\sigma \gets \Sign_\sk(\m)$.
    \item[$\Vrfy:$] takes as input a public key $\pk$, a message $\m$, and a signature $\sigma$. It outputs 1 if the signature is valid and 0 otherwise. We write $b := \Vrfy_\pk(\m, \sigma)$.
\end{itemize}
It is required that except with negligible over $(\pk, \sk)$ output by $\Gen(1^\secpar)$, it holds that $\Vrfy_\pk(\m, \Sign_\sk(\m)) = 1$ for every $\m \in \calM$.

\begin{figure}[t]
    \centering
    \hfpages{
        0.5
    }{
        \underline{$\SigForge_{\adv, \Pi}(\secpar)$}\\
        $(\pk, \sk) \gets \Gen(1^\secpar)$\\
        $(\m, \sigma) \gets \adv^{\Sign_{\sk}(\cdot)}(\pk)$\\
        Let $\calQ$ denote the set of queries made by $\adv$ to $\Sign_\sk(\cdot)$.\\
        \texttt{return} $\Vrfy_\pk(\m, \sigma) \land \m \notin \calQ$
    }
    \caption{The signature forgery experiment $\SigForge_{\adv, \Pi}(\secpar)$\label{fig:sigforge}}
\end{figure}

\begin{definition*}[Security of a digital signature scheme.]
    A digital signature scheme $\Pi = (\Gen, \Sign, \Vrfy)$ is \emph{existentially unforgeable under an adaptive chosen-message attack}, or just \emph{secure}, if for all polynomial-time adversaries $\adv$, 
    $$\Pr[\SigForge_{\adv,\Pi}(\secpar) = 1] \leq \negl$$
    where the experiment $\SigForge_{\adv,\Pi}(\secpar)$ is defined in \Cref{fig:sigforge}.
\end{definition*}

\subsection{Coding theory preliminaries}
A \emph{channel} is a randomized map $\calE : \Sigma^* \to \Sigma^*$. That is, for $x \in \Sigma^*$, $\calE(x)$ is a random sample from $\Sigma^*$. We use channels to model errors introduced by the environment, or by an adversary attempting to e.g. remove a watermark. Two of the most important channels we consider are the \emph{binary symmetric channel} $\BSC$ and the \emph{binary deletion channel} $\BDC$:
\begin{itemize}
    \item $\BSC_p : \{0,1\}^* \to \{0,1\}^*$ is the binary symmetric channel with error rate $p \in (0,1)$. That is, $\BSC_p(x) = x \oplus e$ where $e \from \Ber(\len x, p)$.
    \item $\BDC_q : \{0,1\}^* \to \{0,1\}^*$ is the binary deletion channel with deletion rate $q \in (0,1)$. That is, $\BDC_p(x)$ randomly deletes each bit $x_i$ independently with probability $q$.
\end{itemize}

For the purposes of this work, an error-correcting code with robustness to the channel $\calE$ is a pair of algorithms $(\enc,\dec)$ where $\enc, \dec : \Sigma^* \to \Sigma^*$. Error-correction (or robustness) for the channel $\calE$ says that if $x \from \enc(\m)$, then $\dec(\calE(x)) = \m$.
The \emph{block length} of an error correcting code is the number of symbols in a codeword required to encode a message of a particular length.
We therefore write the block length $n = n(k)$ as a function of the message length $k$.
The \emph{rate} of a code is the function $k \mapsto k / n(k)$, which may or may not depend on the message length $k$.

\newpage

\section{Pseudorandom code basics} \label{sec:prc-basics}
\subsection{Definitions} \label{subsec:prc-defs}
In this section we define pseudorandom codes (PRCs)\footnote{An unrelated notion of pseudorandom codes was defined by \cite{kolesnikov2016efficient}. Their notion does not require efficient decoding or pseudorandomness of codewords; rather, they require only that codewords are far apart.} and related terminology. 
A PRC can be viewed as a family\footnote{We remark that if a PRC were a fixed code rather than a family of codes, pseudorandomness against non-uniform adversaries would be impossible since the adversary could have the decoding key hard-coded.} of error-correcting codes indexed by encoding and decoding keys.
For secret-key PRCs, the encoding and decoding keys are identical; for public-key PRCs, the encoding key is public and the decoding key is secret.

Formally, a PRC is specified by three algorithms. A key generation function samples the keys.
An encoding function takes as input the encoding key and a message, and it outputs a codeword.
A decoding function takes as input the decoding key and a perturbed codeword, and it outputs the message or $\bot$.

Our error correction guarantee is defined in terms of a channel.
A PRC is robust to a channel $\calE$ if $\Decode$ can recover $\m$ from $\calE(\Encode(\m))$ with overwhelming probability.
In addition to requiring that we can recover messages from noisy codewords, we also require that $\Decode$ outputs $\bot$ given a string that is unrelated to the code.
That is, for any string $c$, $\Decode(c)$ outputs $\bot$ with overwhelming probability over the choice of the decoding key.
This property is important for applications such as watermarking, where we want the ability to distinguish codewords from uniformly random strings.

Pseudorandomness of the PRC ensures that an adversary without knowledge of the secret key cannot distinguish between an oracle for the $\Encode$ algorithm of the scheme, or an oracle that outputs independently drawn uniform strings.
If a PRC is viewed as an encryption scheme, pseudorandomness is equivalent to indistinguishability from random bits against a chosen plaintext attack (IND-\$CPA security) \cite{rogaway2003ocb}.
For public-key PRCs, pseudorandomness holds even against adversaries holding the encryption key.

\begin{definition}[Secret-key PRC] \label{def:skPRC}
    Let $\Sigma$ be a fixed alphabet. A \emph{secret-key pseudorandom error-correcting code} (abbreviated as secret-key PRC) with robustness to a channel $\calE : \Sigma^* \to \Sigma^*$ is a triple of polynomial-time randomized algorithms $(\KeyGen, \Encode, \Decode)$ satisfying
    \begin{itemize}
        \item (Syntax) There exist functions $\ell, n, k : \N \to \N$ such that for all $\secpar \in \N$, $\KeyGen(\secparam) \in \{0,1\}^{\ell(\secpar)}$, $\Encode : \{\secparam\} \times \{0,1\}^{\ell(\secpar)} \times \Sigma^{k(\secpar)} \to \Sigma^{n(\secpar)}$, and $\Decode : \{\secparam\} \times \{0,1\}^{\ell(\secpar)} \times \Sigma^* \to \Sigma^{k(\secpar)} \cup \{\bot\}$.
        \item (Error correction, or robustness) For any $\secpar \in \N$ and any message $\m \in \Sigma^{k(\secpar)}$,
        \[
            \Pr_{\sk \from \KeyGen(\secparam)}[\Decode(\secparam,\sk,\calE(x)) = \m : x \from \Encode(\secparam,\sk,\m)] \geq 1-\negl.
        \]
        \item (Soundness) For any fixed $c \in \Sigma^*$,
        \[
            \Pr_{\sk \from \KeyGen(\secparam)}[\Decode(\secparam,\sk,c) = \bot] \geq 1 - \negl.
        \]
        \item (Pseudorandomness) For any polynomial-time adversary $\adv$,
        \[
            \abs{\Pr_{\sk \from \KeyGen(\secparam)}[\adv^{\Encode(\secparam,\sk, \cdot)}(\secparam) = 1] - \Pr_{\calU}[\adv^{\calU} (\secparam) = 1]} \leq \negl,
        \]
        where $\adv^{\calU}$ means that the adversary has access to an oracle that, on any (even previously queried) input, responds with a freshly drawn uniform value in $\Sigma^{n(\secpar)}$. 
    \end{itemize}
\end{definition}

\begin{definition}[Public-key PRC] \label{def:pkPRC}
    Let $\Sigma$ be a fixed alphabet. A \emph{public-key pseudorandom error-correcting code} (abbreviated as public-key PRC) with robustness to a channel $\calE : \Sigma^* \to \Sigma^*$ is a triple of polynomial-time randomized algorithms $(\KeyGen, \Encode, \Decode)$ satisfying
    \begin{itemize}
        \item (Syntax) There exist functions $\ell_\dec, \ell_\enc, n, k : \N \to \N$ such that for all $\secpar \in \N$, $\KeyGen(\secparam) \in \{0,1\}^{\ell_\dec(\secpar)} \times \{0,1\}^{\ell_\enc(\secpar)}$, $\Encode : \{\secparam\} \times \{0,1\}^{\ell_\enc(\secpar)} \times \Sigma^{k(\secpar)} \to \Sigma^{n(\secpar)}$, and $\Decode : \{\secparam\} \times \{0,1\}^{\ell_\dec(\secpar)} \times \Sigma^* \to \Sigma^{k(\secpar)} \cup \{\bot\}$.
        \item (Error correction, or robustness) For any $\secpar \in \N$ and any message $\m \in \Sigma^{k(\secpar)}$,
        \begin{align*}
            \Pr_{(\sk,\pk) \from \KeyGen(\secparam)}[\Decode(\secparam,\sk,\calE(x)) = \m : x \from \Encode(\secparam,\pk,\m)] \geq 1-\negl.
        \end{align*}
        \item (Soundness) For any fixed $c \in \Sigma^*$,
        \[
            \Pr_{(\sk,\pk) \from \KeyGen(\secparam)}[\Decode(\secparam,\sk,c) = \bot] \geq 1 - \negl.
        \]
        \item (Pseudorandomness) For any polynomial-time adversary $\adv$,
        \ifeprint
        \[
            \abs{\Pr_{(\sk,\pk) \from \KeyGen(\secparam)}[\adv^{\Encode(\secparam,\pk, \cdot)}(\secparam,\pk) = 1] - \Pr_{\substack{(\sk,\pk) \from \KeyGen(\secparam) \\ \calU}}[\adv^{\calU}(\secparam,\pk) = 1]} \leq \negl,
        \]
        \fi
        where $\adv^{\calU}$ means that the adversary has access to an oracle that, on any (even previously queried) input, responds with a freshly drawn uniform value in $\Sigma^{n(\secpar)}$.
    \end{itemize}
\end{definition}

The \emph{block length} of a (secret-key or public-key) PRC is $n(\secpar)$ and the \emph{message length} is $k(\secpar)$. The \emph{rate} is the function $\secpar \mapsto k(\secpar) / n(\secpar)$. We often drop the dependence on $\secpar$ when it is clear from context.

For both secret-key and public-key PRCs, if there is only one possible message (i.e. $k(\secpar) = 0$), then we say that the scheme is a \emph{zero-bit} PRC.

\begin{remark*}
    We present our constructions of PRCs for messages of \emph{fixed} lengths. That is, for a given set of keys, the construction will only work for messages of a fixed length. However, we note that this can be easily remedied using a pseudorandom function (PRF) to select new keys for every message length.
    We do not include the PRF in our constructions in order to simplify presentation.
\end{remark*}

\ifeprint
\subsection{Heuristic construction from permuted codes} \label{subsec:permuted-codes}
In this section we describe a natural \emph{heuristic} transformation for building a candidate secret-key PRC from any error-correcting code. In \Cref{sec:ldpc-prcs} we will give provably secure constructions of PRCs from standard (subexponential) cryptographic assumptions.

Our heuristic construction can be applied to any binary error-correcting code, such as the polar code \cite{Ari09}. The secret key will be a random permutation of the indices of the codewords. To encode a message, we encode it using the error-correcting code, then apply the secret permutation, and finally add a small amount of random Bernoulli noise.

There are two drawbacks of this generic permuted code construction relative to our pseudorandom LDPC codes:
\begin{itemize}
    \item In general the pseudorandomness of a permuted code is based on non-standard, ad-hoc conjectures. For certain codes, such as Gallager's LDPC ensemble \cite{Gal62}, the permuted code construction is \emph{not} pseudorandom.
    \item Permuted codes are inherently secret-key PRCs, whereas our pseudorandom LDPC codes are public-key.
\end{itemize}

For any error-correcting code $(\enc, \dec)$ and parameter $\eta > 0$, we formally define the corresponding permuted code $\PermPRC_\eta[\enc,\dec]$ as follows. Let $n = n(\cdot)$ be the block length for $(\enc, \dec)$, as a function of the message length. (Recall that the block length is the number of codeword symbols needed to encode a given message.)
$\PermPRC_\eta[\enc,\dec]$ will encode messages of length $k$ into codewords of length $n = n(k+\secpar)$, where $\secpar$ is a security parameter.

\begin{construction}[{$\PermPRC_\eta[\enc,\dec]$}]
\begin{itemize}
    \item[]
    \item Let $\secpar$ be a security parameter, let $k$ be the length of messages we wish to encode, and let $n = n(k+\secpar)$.
    \item $\KeyGen(1^\secpar)$: Sample a random permutation $\pi : [n] \to [n]$ and output $\sk = \pi$.
    \item $\Encode(\sk, \m)$ for $\m \in \Sigma^{k}$:
    \begin{enumerate}
        \item Sample a random string $r \from \Sigma^\secpar$ and a noise vector $e \from \Ber(n,\eta)$.
        \item Compute $c = \enc(r || \m) \oplus e$.
        \item Output $(c_{\pi(1)} || \cdots || c_{\pi(n)})$.
    \end{enumerate}
    \item $\Decode(\sk, x)$ for $x \in \Sigma^n$:
    \begin{enumerate}
        \item Compute $c = (x_{\pi^{-1}(1)} || \cdots || x_{\pi^{-1}(n)})$.
        \item Compute $\m' = \dec(c)$.
        \item Output the last $k$ symbols of $\m'$. 
    \end{enumerate}
\end{itemize}
\end{construction}
A permuted code has the same robustness to substitutions, and nearly the same rate, as $(\enc,\dec)$.

A natural choice of error-correcting codes to use in this permuted code construction is polar codes \cite{Ari09}. Since polar codes have linear rate and tolerate a constant rate of adversarial errors, permuted polar codes are a candidate linear-rate secret-key PRC with robustness to a constant rate of adversarial errors. We do not know whether such codes satisfy pseudorandomness.
\else
\fi

\newpage

\section{Constructing pseudorandom codes from cryptographic assumptions} \label{sec:ldpc-prcs}
\ifeprint
 In this section, we introduce a public-key PRC based on LDPC codes. We present the zero-bit version of our scheme, $\LDPCPRC_0$, in \Cref{subsec:ldpc-prc-construction}; we will see in \Cref{sec:improved-prcs} that this immediately implies a many-bit scheme with essentially the same robustness. In \Cref{subsec:codeword-detection} we show that $\LDPCPRC_0$ is robust to every error channel of bounded weight, as long as the parity checks have sufficiently low weight. Finally, we prove pseudorandomness of $\LDPCPRC_0$ in two different parameter regimes under different cryptographic assumptions: In \Cref{subsec:planted-xor}, we prove pseudorandomness under LPN and a certain planted XOR assumption; in \Cref{subsec:subexp-lpn}, we prove pseudorandomness under a subexponential-query variant of LPN.
\fi

\subsection{The zero-bit construction} \label{subsec:ldpc-prc-construction}
Let
\[
    \calS_{t,n} = \{s \in \F_2^n : \wt(s) = t\}
\]
be the set of all $t$-sparse vectors in $\F_2^n$, and
\[
    \calS_{t,r,n} = \{P \in \F_2^{r \times n} : \wt(P_{i,:}) = t\ \forall i \in [r]\}
\]
be the set of all $t$-row-sparse matrices in $\F_2^{r \times n}$.

Our zero-bit pseudorandom LDPC codes are parameterized by a public generator matrix $G \in \F_2^{n \times g}$ and a secret parity-check matrix $P \in \F_2^{r \times n}$. The sampling process for these matrices is described in \Cref{def:random-ldpc}.

\begin{definition}[Random LDPC code, {$\LDPC[n,g,t,r]$}] \label{def:random-ldpc}
    \sloppy
    For $n, g, t, r \in \N$, define the distribution $\LDPC[n,g,t,r]$ over $\F_2^{r \times n} \times \F_2^{n \times g}$ as follows:
    \begin{enumerate}
        \item[] $\LDPC[n,g,t,r]$:
        \item Sample $P \from \calS_{t,r,n}$, i.e. $P \in \F_2^{r \times n}$ is chosen to have i.i.d random $t$-sparse rows.
        \item Sample $G \from (\ker P)^g$, i.e. $G \in \F_2^{n \times g}$ is a random matrix subject to $PG = 0$.
        \item Output $(P,G)$.
    \end{enumerate}
    An $(n,g,t,r)$ \emph{random LDPC code} is a pair of matrices $(P, G) \from \LDPC[n,g,t,r]$.
\end{definition}

The focus of this section will be on the following \emph{zero-bit} PRC. Recall that a zero-bit PRC is one whose message space is just $\{1\}$. We will see in \Cref{subsec:constant-rate-prcs} that a constant-rate PRC can be generically constructed from any zero-bit PRC.
\ifeprint
\begin{construction}[Zero-bit public-key pseudorandom LDPC code, {$\LDPCPRC_0[n,g,t,r,\eta,\zeta]$}] \label{const:zero-bit-ldpc-prc}
\else 
\begin{construction}[Zero-bit public-key LDPC PRC, {$\LDPCPRC_0[n,g,t,r,\eta,\zeta]$}] \label{const:zero-bit-ldpc-prc}
\fi 
    Let $n, g, t, r : \mathbb{N} \to \mathbb{N}$ and $\eta, \zeta : \N \to [0,1/2)$ be efficiently-computable functions of the security parameter. We define $\LDPCPRC_0[n,g,t,r,\eta,\zeta]$ by the following algorithms, where we leave the dependence of $n,g,t,r,\eta,\zeta$ on $\secpar$ implicit:
    \begin{itemize}
        \item $\KeyGen(\secparam)$: Sample $(P,G) \from \LDPC[n,g,t,r]$ and $z \gets \F_2^n$. Output $(\sk = (P,z), \pk = (G,z))$.
        \item $\Encode(\secparam,(G,z))$: Sample $u \from \F_2^g$, $e \from \Ber(n,\eta)$. Output $Gu \oplus z \oplus e$.
        \item $\Decode(\secparam,(P,z),x)$: If $\wt(Px \oplus Pz) < \left(\frac{1}{2} - \zeta\right) \cdot r$, output 1; otherwise output $\bot$.
    \end{itemize}
\end{construction}

For the remainder of this section, we will identify the security parameter $\secpar$ with the dimension of the code $n$. We will therefore write $g, t, r, \eta, \zeta$ as functions of $n$, with the understanding that $n(\secpar) = \secpar$.

\subsection{Codeword detection (zero-bit decoding)} \label{subsec:codeword-detection}
We say that a length-preserving binary channel $\calE : \{0,1\}^* \to \{0,1\}^*$ is \emph{$p$-bounded} if there exists a negligible function $\negl[\cdot]$ such that for all $n \in \N$, $\Pr_{x \from \{0,1\}^n}[\wt(\calE(x) \oplus x) > p n] \leq \negl[n]$.

\begin{lemma} \label{lemma:zero-bit-ldpc-decoding}
    For any $p, \eta \in [0,1/2)$ and $\varepsilon \in (0,1)$, there exits $\delta > 0$ such that the following holds. For any $t \le \delta \log n$, $g > 0$, and $n^{\varepsilon} \le r \le 0.99 n$, $\LDPCPRC_0[n,g,t,r,\eta,r^{-1/4}]$ is robust to every $p$-bounded channel.
\end{lemma}
\begin{proof}
    \Cref{lemma:codeword-nondetection} shows that any fixed string $c \in \F_2^n$ decodes to $\bot$ with probability $1-\negl[n]$. It remains to show that codewords subject to any $p$-bounded channel decode to $1$.
    
    Let $\calE$ be any $p$-bounded channel. We need to show that $\wt(Px \oplus Pz) < \left(\frac{1}{2} - r^{-1/4}\right) r$ with probability $1-\negl[n]$ over $(P,G) \from \LDPC[n,g,t,r], u \from \F_2^g, z \from \F_2^n, e \from \Ber(n,\eta), x \from \calE(Gu \oplus z \oplus e)$.
    
    We will show that there exists a constant $\alpha \in (0,1/2)$ such that $\wt(Gu \oplus z \oplus x) \le \left(\frac{1}{2} - \alpha\right) n$ with probability $1-\negl[n]$. Then we can apply \Cref{lemma:codeword-detection} with $y = Gu \oplus z \oplus x$ to see that $\wt(P y) = \wt(Px \oplus Pz) < \left(\frac{1}{2} - r^{-1/4}\right) r$ with probability $1-\negl[n]$ for appropriate choices of $t,r$.

    Let $c = Gu \oplus z \oplus e$. Then
    \begin{align*}
        y &= Gu \oplus z \oplus x \\
        &= Gu \oplus z \oplus \calE(Gu \oplus z \oplus e) \\
        &= e \oplus c \oplus \calE(c),
    \end{align*}
    where $c$ is uniformly distributed because of $z$. Let $e' = c \oplus \calE(c)$, so that $y = e \oplus e'$. Then since $e \from \Ber(n,\eta)$ is independent from $e'$, we have
    \begin{align*}
        \E[\wt(y)] &= \E[\wt(e \oplus e')] \\
        &= (1-\eta) \E[\wt(e')] + \eta \E[n-\wt(e')] \\
        &\le (1-\eta) p n + \eta (1-p) n + \negl[n]
    \end{align*}
    where the inequality holds because $\calE$ is $p$-bounded, so $\wt(e') \le p n$ with probability $1-\negl[n]$. Conditioned on $\wt(e') \le p n$, a Chernoff bound over the random choice of $e$ implies that
    \[
        \wt(e \oplus e') \le \left[\frac{(1-\eta) p + \eta (1-p)}{2} + \frac{1}{4}\right] \cdot n
    \]
    with probability $1-\negl[n]$. Therefore we apply \Cref{lemma:codeword-detection} with $\alpha = \left[\frac{(1-\eta) p + \eta (1-p)}{2} + \frac{1}{4}\right]$.
\end{proof}

\begin{lemma} \label{lemma:codeword-detection}
    Let $\alpha \in (0,1/2)$ be some constant and suppose that $r = n^{\Omega(1)}$ and $t \le \frac{1}{5} \log_{1/2\alpha} r$. If $y \in \F_2^n$ satisfies $\wt(y) \le \left(\frac{1}{2} - \alpha\right) n$, then
    \[
        \Pr_{P \from \calS_{t,r,n}} \left[\wt(Py) < \left(\frac{1}{2} - r^{-1/4}\right) r \right] \geq 1-\negl[n].
    \]
\end{lemma}
\begin{proof}
Let $P_i$ denote the $i^{\text{th}}$ row of $P$. We will show that
\begin{equation} \label{eq:row-marginals}
    \Pr_{P_i \from \calS_{t,1,n}}[P_i y = 0] \geq \frac{1}{2} + \frac{2}{r^{1/4}},
\end{equation}
and the independence of the rows of $P$ will imply the lemma by a Chernoff bound.

Let $j_1, \dots, j_t \in [n]$, $Y_1, \dots, Y_t \in \{-1,1\}$ be random variables, defined as follows, for $\ell = 1, \dots, t$:
\begin{enumerate}
    \item Sample $j_\ell \from [n] \setminus \{j_1, \dots, j_{\ell-1}\}$.
    \item Let $Y_\ell = (-1)^{x_{j_\ell}}$.
\end{enumerate}
Then the bit $P_i y$ is distributed as $(1 - Y_1 \cdots Y_t)/2$, so
\begin{align*}
    \Pr_{P_i \from \calS_{t,1,n}}[P_i y = 0] &= \Pr[Y_1 \cdots Y_t = 1] \\
    &= \frac{1}{2} + \frac{1}{2}\E[Y_1 \cdots Y_t].
\end{align*}
The remainder of the proof is devoted to showing that
\begin{equation} \label{eq:Ys-product}
    \E[Y_1 \cdots Y_t] \ge (2\alpha)^t - 2t^2/n.
\end{equation}
For large enough $n$, this will imply \Cref{eq:row-marginals} by the assumption that $t \le \frac{1}{5} \log_{1/2\alpha} r$.

Let $\alpha' = 1/2-\wt(y)/n$, and note that $\alpha' \ge \alpha$. To prove \Cref{eq:Ys-product} we first show that for all $m \in [t]$,
\begin{equation} \label{eq:recurrence}
    \abs{\E[Y_1 \cdots Y_m] - 2\alpha' \E[Y_1 \cdots Y_{m-1}]} \le 2t/n. 
\end{equation}
By the tower property and linearity of expectation,
\begin{equation} \label{eq:expanded-recurrence}
    \abs{\E[Y_1 \cdots Y_m] - 2\alpha' \E[Y_1 \cdots Y_{m-1}]} = \abs{\E\left[Y_1 \cdots Y_{m-1} \cdot \big(\E[Y_m \mid Y_{<m}] - 2\alpha'\big)\right]}.
\end{equation}
For all possible assignments of $Y_{<m}$, 
\begin{align*}
    \E[Y_m \mid Y_{<m}] &= \Pr_{j_m}[x_{j_m} = 0 \mid Y_{<m}] - (1 - \Pr_{j_m}[x_{j_m} = 0 \mid Y_{<m}]) \\
    &= 2\Pr_{j_m}[x_{j_m} = 0 \mid Y_{< m}] - 1
\end{align*}
Recall that $j_m$ is chosen to be a uniformly random index from $[n] \setminus \{j_{<  m}\}$.
Since $m \leq t$, $[n] \setminus \{j_{<m}\}$ contains at least $(\frac{1}{2} + \alpha')n - t$ and at most $(\frac{1}{2} + \alpha')n$ indices for which the corresponding bit of $y$ is 0. 
Therefore, $\frac{1}{2} + \alpha' - t/n \leq \Pr_{j_m}[x_{j_m} = 0 \mid Y_{< m}] \leq \frac{1}{2} + \alpha' + t/n$.

We now have that $\abs{\E[Y_m \mid Y_{<m}] - 2\alpha'} \le 2t/n$, and $\abs{Y_1 \cdots Y_{m-1}} = 1$. Plugging these two facts into \Cref{eq:expanded-recurrence} gives us \Cref{eq:recurrence}.

We complete the proof by showing by induction that $\E[Y_1 \cdots Y_m] \ge (2\alpha')^m - 2mt/n$. Observe first that $\E[Y_1] = (\frac{1}{2} + \alpha') - (\frac{1}{2} - \alpha') > 2\alpha' - 2t/n$.
Assume that the inequality holds for some $m < t$; we'll show that it holds for $m + 1$.
We have $\E[Y_1 \cdots Y_{m+1}] \geq 2\alpha' \E[Y_1 \cdots Y_m] - 2t/n$. Therefore, since $2\alpha' \le 1$,
\begin{align*}
    \E[Y_1 \cdots Y_{m+1}] &\geq 2\alpha'\left( (2\alpha')^m - \frac{2mt}{n}\right) -\frac{2t}{n}\\
    &= (2\alpha')^{m+1} - \frac{2(m+1) t}{n}\\
\end{align*}
as desired.
\end{proof}

\begin{lemma} \label{lemma:codeword-nondetection}
    If $n^{\Omega(1)} \le r \le 0.99 n$, then for any fixed $c \in \F_2^n$,
    \[
        \Pr_{\substack{P \from \calS_{t,r,n} \\ z \from \F_2^n}}\left[\wt(Pc \oplus Pz) \ge \left(\frac{1}{2} - r^{-1/4}\right) r\right] \geq 1-\negl[n].
    \]
\end{lemma}
\begin{proof}
    With probability $1-\negl[n]$, $P \from \calS_{t,r,n}$ is full rank by \Cref{lemma:full-rank} (invoking the lemma with $g=0$). For any such $P$, over the random choice of $z \from \F_2^n$, the vector $P(c \oplus z) \in \F_2^r$ is uniformly random. The lemma follows from a Chernoff bound.
\end{proof}

\subsection{Pseudorandomness from the planted XOR assumption and LPN} \label{subsec:planted-xor}
In this subsection we prove that the generator matrix (public key) of $\LDPCPRC_0$ is pseudorandom under the planted XOR assumption. This is \Cref{lemma:xor-hybrid}. The number of columns of the generator matrix and the number of rows of the parity check matrix will depend on the particular planted XOR assumption one is willing to make. Then, since the generator matrix is pseudorandom, pseudorandomness of $\LDPCPRC_0$ (\Cref{theorem:ldpc-prc-xor}) follows directly from the standard LPN assumption.

\paragraph{Cryptographic assumptions.}
The existence of our LDPC-based PRCs relies on either of two assumptions. We state the two assumptions together as \Cref{assumption:combined}.

\begin{assumption} \label{assumption:combined}
    At least one of the following two statements is true:
    \begin{itemize}
        \item There exists a constant $\eta \in (0,1/2)$ such that, for any function $g(n) = \Omega(\log^2 n)$, the $\LPN_{g, \eta}$ assumption (\Cref{assumption:LPN}) holds.
        \item There exist constants $\eta \in (0,1/2)$ and $\varepsilon \in (0,1)$ such that, for any function $t(n) = \Theta(\log n)$, both the $\LPN_{n^\varepsilon, \eta}$ assumption (\Cref{assumption:LPN}) and the $\XOR_{2 n^\varepsilon, t}$ assumption (\Cref{assumption:planted-xor}) hold.
    \end{itemize}
\end{assumption}

We now define the $\LPN_{g,\eta}$ and $\XOR_{m,t}$ assumptions. Let us first recall the LPN assumption. We only state the decisional, constant-noise version, since all of our results pertain to that variant. For $\eta \in (0, 1/2)$ and $g : \N \to \N$, the $\LPN_{g,\eta}$ problem is to distinguish between $(A, As \oplus e)$ and $(A, u)$, where $A \from \F_2^{n \times g(n)}$, $s \from \F_2^{g(n)}, e \from \Ber(n,\eta)$, and $u \from \F_2^n$.

\begin{assumption}[LPN assumption] \label{assumption:LPN}
    For $\eta \in (0, 1/2)$ and $g : \N \to \N$, the $\LPN_{g,\eta}$ assumption states that for every $n \in \N$ and every polynomial-time adversary $\adv$,
    \[
        \abs{\Pr_{\substack{A \from \F_2^{n \times g(n)}\\ s \from \F_2^{g(n)}\\ e \from \Ber(n,\eta)}}[\adv(A, As \oplus e) = 1] - \Pr_{\substack{A \from \F_2^{n \times g(n)}\\ u \from \F_2^n}}[\adv(A, u) = 1]} \leq \negl[n].
    \]
\end{assumption}

By a standard hybrid argument, the $\LPN_{g, \eta}$ assumption implies that any polynomial number of samples of the form $(A,As \oplus e)$ are indistinguishable from uniformly random samples.

The planted XOR assumption is lesser-known than LPN, but there is still precedent \cite{ASSVV23}. This assumption states that a random linear subspace over $\F_2$ is indistinguishable from one containing a planted sparse vector. Formally, it says that the distributions $\calD_0(n,m)$ and $\calD_1(n,m,t)$ are computationally indistinguishable, where the ``null'' distribution $\calD_0$ is the uniform distribution over $\F_2^{n \times m}$, and the ``planted'' distribution $\calD_1$ is defined as follows.

\indent $\calD_1(n,m,t)$:
\begin{enumerate}
    \item Sample $s \from \calS_{t,n}$.
    \item Sample a random matrix $G \in \F_2^{n \times m}$ subject to $s^T G = 0$.
    \item Output $G$.
\end{enumerate}

\begin{assumption}[Planted XOR assumption] \label{assumption:planted-xor}
    For $m, t : \N \to \N$, the $\XOR_{m,t}$ assumption states that for every $n \in \N$ and every polynomial-time adversary $\adv$,
     \[
        \abs{\Pr_{G \from \calD_0(n,m(n))}[\adv(G) = 1] - \Pr_{G \from \calD_1(n,m(n),t(n))}[\adv(G) = 1]} \leq \negl[n].
    \]
\end{assumption}

\begin{remark*}
    A previous version of this paper made use of the planted XOR assumption with $m = \Omega(n)$. However, an anonymous CRYPTO 2024 reviewer pointed out to us that this assumption is false by Theorem 4.26 of \cite{ASSVV23}. All of our results are therefore updated in this version of the paper to use $m = n^{\Omega(1)}$; none of our results significantly change with the new choice of parameters. See \Cref{theorem:ldpc-prc-xor} for the updated theorem statement. Note also that \Cref{theorem:ldpc-prc-lpn} is not affected, as it does not make use of the planted XOR assumption.
\end{remark*}

Note that when $m = n - O(\log n)$, the $\XOR_{m,t}$ assumption is false because there will be only $n^{O(1)}$ vectors $v \in \F_2^n$ such that $v^T G = 0$, and one can therefore brute force search for the planted relation $s$. On the other hand, \Cref{claim:lhl} shows that when (for instance) $t \sim \log n$ and $m \sim \log^2 n$ the $\XOR_{m,t}$ assumption holds against even \emph{unbounded} adversaries. This claim follows from Theorem 4.2 of \cite{ASSVV23}.

\begin{numberedclaim} \label{claim:lhl}
    If $t = O(\log n)$ and $m \le (1-\Omega(1)) t \log n$, then
    \[
        SD(\calD_0(n,m), \calD_1(n,m,t)) = n^{-\Omega(t)}.
    \]
\end{numberedclaim}

For larger values of $m$, $\calD_0(n,m)$ and $\calD_1(n,m,t)$ are no longer statistically close, but the planted XOR assumption says that they remain computationally indistinguishable.

\begin{lemma} \label{lemma:xor-hybrid}
    Let $m,t,r : \N \to \N$ be such that $m(n) + r(n) \le n-\omega(\log n)$. The $\XOR_{m+r,t}$ assumption (\Cref{assumption:planted-xor}) implies that the marginal distribution on $G$ for $(P, G) \from \LDPC[n,m,t,r]$ is pseudorandom.
\end{lemma}
\begin{proof}
    The proof closely mirrors that in the technical overview (\Cref{subsec:techo-prc-basics}), with the main difference being that here we deal with generator matrices instead of the linear subspaces themselves. For $i \in \{0, \dots, r-1\}$ and $m' \in [m+r]$, let $\calD_i(n,m',t)$ be defined as follows.
    \begin{enumerate}
        \item[] $\calD_i(n,m',t)$:
        \item Sample $s_1, \dots, s_i \from \calS_{t,n}$.
        \item Sample a random matrix $G \in \F_2^{n \times m'}$ subject to $s_j^T G = 0$ for all $j \in [i]$.
        \item Output $G$.
    \end{enumerate}
    Observe that $\calD_0(n,m',t) = \calD_0(n,m')$, and $\calD_1(n,m',t)$ is consistent with the definition given earlier.

    For each $i \in \{0, \dots, r-1\}$ and $m' \in [m+r]$, since $m' \le m + r \le n - \omega(\log n)$, the matrix $G \from \calD_i(n,m',t)$ has full rank with probability $1-\negl[n]$. Therefore, $\calD_i(n,m',t)$ is $\negl[n]$-close in statistical distance to the following distribution:
    \begin{enumerate}
        \item[] $\hat{\calD}_i(n,m',t)$:
        \item Sample $s_1, \dots, s_i \from \calS_{t,n}$.
        \item Sample a random \emph{full-rank} matrix $G \in \F_2^{n \times m'}$ subject to $s_j^T G = 0$ for all $j \in [i]$.
        \item Output $G$.
    \end{enumerate}
    Since $\hat{\calD}_0(n,m+r,t)$ (resp. $\hat{\calD}_1(n,m+r,t)$) is $\negl[n]$-close to $\calD_0(n,m+r)$ (resp. $\calD_1(n,m+r,t)$) in statistical distance, the $(n,m+r,t)$ planted XOR assumption implies that $\hat{\calD}_0(n,m+r,t)$, and $\hat{\calD}_1(n,m+r,t)$ are computationally indistinguishable.
    
    Now suppose that an efficient adversary $\adv$ distinguishes between $\hat{\calD}_0(n,m,t)$ and $\hat{\calD}_r(n,m,t)$ with advantage $\varepsilon > 0$. By a telescoping argument, $\adv$ must distinguish between $\hat{\calD}_i(n,m,t)$ and $\hat{\calD}_{i+1}(n,m,t)$ with advantage $\varepsilon/r$, for some $i \in \{0, \dots, r-1\}$. For each $i \in \{0,\dots,r-1\}$, the following efficient reduction $\Red_i$ satisfies $\Red_i(\hat{\calD}_0(n,m+r,t)) \equiv \hat{\calD}_i(n,m,t)$ and $\Red_i(\hat{\calD}_1(n,m+r,t)) \equiv \hat{\calD}_{i+1}(n,m,t)$.
    Therefore, the $(n,m+r,t)$ planted XOR assumption implies that $\varepsilon/r = \negl[n]$, which will complete the proof.
    \begin{enumerate}
        \item[] $\Red_i(W)$:
        \item Sample $i$ random $t$-sparse vectors $s_1, \dots, s_i \in \F_2^n$ and let $S = \{v \in \F_2^n : v \cdot s_j = 0\ \forall j \in [i]\}$.
        \item Let $U = \img(W) \cap S$. Since $i < r$ and $\dim \img W = m+r$, we have $\dim U > m$.
        \item Sample a random full-rank matrix $G \in \F_2^{n \times m}$ such that $\img(G) \subseteq U$.
        \item Output $G$.
    \end{enumerate}

    It remains to see why $\Red_i(\hat{\calD}_0(n,m+r,t)) \equiv \hat{\calD}_i(n,m,t)$ and $\Red_i(\hat{\calD}_1(n,m+r,t)) \equiv \hat{\calD}_{i+1}(n,m,t)$.
    In fact both of these statements are true even for \emph{fixed} planted relations.
    
    \paragraph{Proof that $\Red_i(\hat{\calD}_0(n,m+r,t)) \equiv \hat{\calD}_i(n,m,t)$.}
    Suppose that $W \from \hat{\calD}_0(n,m+r,t)$. Fix $s_1, \dots, s_i$ sampled in $\Red_i(W)$ and let $S = \{v \in \F_2^n : v \cdot s_j = 0\ \forall j \in [i]\}$. For any $d \in \{m+1, \dots, m+r\}$, consider the distribution of the subspace $U = \img(W) \cap S$ conditioned on the event that $\dim U = d$. Before the conditioning $\img W$ was a random subspace of $\F_2^n$, so after the conditioning $U$ is a random $d$-dimensional subspace of $S$. The output of $\Red_i(W)$ is a random full-rank matrix $G \in \F_2^{n \times m}$ such that $\img(G) \subseteq U \subseteq S$. But we have just seen that $U$ is a uniformly random $d > m$ dimensional subspace of $S$, so it follows that $G$ is a random matrix subject to $\img G \subseteq S$ --- i.e., $s_j^T G = 0$ for all $j \in [i]$.
    
    \paragraph{Proof that $\Red_i(\hat{\calD}_1(n,m+r,t)) \equiv \hat{\calD}_{i+1}(n,m,t)$.}
    Suppose that $W \from \hat{\calD}_1(n,m+r,t)$ is sampled with the planted relation $s$. Fix $s$ and $s_1, \dots, s_i$ sampled in $\Red_i$. Let $S = \{v \in \F_2^n : v \cdot s_j = 0\ \forall j \in [i]\}$ and $S' = \{v \in \F_2^n : v \cdot s = 0\}$. Again, for each $d \in \{m+1, \dots, m+r\}$, consider the distribution of $U = \img(W) \cap S$ conditioned on the event that $\dim U = d$. Before the conditioning $\img W$ was a random subspace of $S'$, so after the conditioning $U$ is a random $d$-dimensional subspace of $S \cap S'$. The output of $\Red_i(W)$ is a random full-rank matrix $G \in \F_2^{n \times m}$ such that $\img(G) \subseteq U \subseteq S \cap S'$. But we have just seen that $U$ is a uniformly random $d > m$ dimensional subspace of $S \cap S'$, so it follows that $G$ is a random matrix subject to $\img G \subseteq S \cap S'$ --- i.e., $s^T G = 0$ and $s_j^T G = 0$ for all $j \in [i]$.
\end{proof}

\Cref{claim:lhl} and \Cref{lemma:xor-hybrid} together imply that the generator matrix from $\LDPC[n, \log^2 n, 4 \log n, \log^2 n]$ is statistically uniform. In \Cref{lemma:low-rate-random-ldpc} we improve this to show that the generator matrix from $\LDPC[n, g, 4 \log n, r]$ is statistically uniform for any $r \le 0.99 n$ (and some $g = \Omega(\log^2 n)$); this forms the basis of our PRC construction based only on LPN (\Cref{theorem:ldpc-prc-lpn}).

\begin{theorem} \label{theorem:ldpc-prc-xor}
    \sloppy
    For any constants $p, \eta \in (0,1/2)$ and $\varepsilon \in (0,1)$, there exists a function $t = \Theta(\log n)$ such that $\LDPCPRC_0[n, n^\varepsilon, t, n^\varepsilon, \eta, n^{-\varepsilon/4}]$ is a zero-bit public-key PRC (\Cref{def:pkPRC}) that is robust to every $p$-bounded channel, where pseudorandomness rests on the $\LPN_{n^\varepsilon, \eta}$ assumption (\Cref{assumption:LPN}) and the $\XOR_{2 n^\varepsilon, t}$ assumption (\Cref{assumption:planted-xor}).
\end{theorem}
\begin{proof}
    By \Cref{lemma:zero-bit-ldpc-decoding}, there exists $t = \Theta(\log n)$ such that $\LDPCPRC_0[n, n^\varepsilon, t, n^\varepsilon, \eta, n^{-\varepsilon/4}]$ is robust to every $p$-bounded channel.
    
    By \Cref{lemma:xor-hybrid}, the $\XOR_{2n^\varepsilon, t}$ assumption implies that for $(P,G) \from \LDPC[n,n^\varepsilon, t, n^\varepsilon]$, the marginal distribution on $G$ is pseudorandom. By the $\LPN_{n^\varepsilon, \eta}$ assumption, it follows that $\LDPCPRC_0[n, n^\varepsilon, t, n^\varepsilon, \eta, n^{-\varepsilon/4}]$ is a public-key PRC.
\end{proof}

\subsection{Pseudorandomness from subexponential LPN} \label{subsec:subexp-lpn}
In \Cref{subsec:planted-xor}, we showed that the generator matrix $G \in \F_2^{n \times g}$ of $\LDPCPRC_0$ was pseudorandom under the planted XOR assumption. In \Cref{lemma:low-rate-random-ldpc} of this section, we will show that $G$ is \emph{statistically} random if we set $g$ to be sufficiently small. Therefore, the planted XOR assumption is no longer necessary; however, the number of columns will be only $g = O(\log^2 n)$, so we must rely on a stronger, sub-exponential variant of LPN than in \Cref{subsec:planted-xor}.

\begin{remark*}
    Before we see \Cref{lemma:low-rate-random-ldpc}, note that we use a different proof strategy here than the technical overview. The reason we use this more involved proof here is that the proof outlined in the technical overview results in a different ensemble of parity-check matrices $P$ with a triangular restriction and non-independent rows. By proving \Cref{lemma:low-rate-random-ldpc}, we are able to use the same ensemble as in \Cref{subsec:planted-xor} --- that is, the rows of $P$ are still independent and uniform $t$-sparse vectors.
\end{remark*}

\begin{lemma} \label{lemma:low-rate-random-ldpc}
    If $r \le 0.99 n$ and $\omega(\sqrt{\log n}) \le t \le O(\log n)$, then there is a $g = \Omega(t^2)$ such that the marginal distribution on $G$ for $(P,G) \from \LDPC[n,g,t,r]$ is $\negl[n]$-close to uniform in statistical distance.
\end{lemma}
\begin{proof}
Recall the definitions of $\calS_{t,n}$ and $\calS_{t,r,n}$ from \Cref{subsec:ldpc-prc-construction}. Let $\calD$ be the marginal distribution on $G \in \F_2^{n \times g}$ for $(P,G) \from \LDPC[n,g,t,r]$.

We will show that $\calD$ is $\negl[n]$-close to the uniform distribution in statistical distance. For any $G^* \in \F_2^{n \times g}$,
\begin{align*}
    \Pr_{G \from \calD}[G=G^*] &= \sum_{\substack{P^* \in \calS_{t,r,n} : \\ P^* G^* = 0}} \Pr_{P \from \calS_{t,r,n}}[P=P^*] \cdot \Pr_{G \from (\ker P^*)^g}[G=G^*] \\
    &= \frac{1}{\binom{n}{t}^r} \sum_{\substack{P^* \in \calS_{t,r,n} : \\ P^* G^* = 0}} \frac{1}{\abs{\ker P^*}^g} \\
    &\le \frac{1}{\binom{n}{t}^r} \sum_{\substack{P^* \in \calS_{t,r,n} : \\ P^* G^* = 0}} \frac{1}{2^{(n-r) m}} \\
    &= \frac{\eta(G^*)^r}{\binom{n}{t}^r \cdot 2^{(n-r) m}}
\end{align*}
where $\eta(G^*)$ is the number of collections of $t$ rows of $G^*$ that sum to 0. Since valid rows of $P^*$ correspond to collections of $t$ rows of $G^*$ that sum to 0, $|\{P^* \in \F_2^{r \times n} : \ P^* G^* = 0\}| = \eta(G^*)^r$, which gives us the last equality above. 

Let $\calF \subset \F_2^{r \times n}$ be the set of $t$-row-sparse matrices of full rank,
\[
    \calF := \{P \in \calS_{t,r,n} \mid \rank(P) = r\} = \{P \in \calS_{t,r,n} \mid \dim \ker(P) = n-r\}.
\]
Then we also have
\begin{align*}
    \Pr_{G \from \calD}[G=G^*] &= \frac{1}{\binom{n}{t}^r} \sum_{\substack{P^* \in \calS_{t,r,n} : \\ P^* G^* = 0}} \frac{1}{\abs{\ker P^*}^g} \\
    &\ge \frac{1}{\binom{n}{t}^r} \sum_{\substack{P^* \in \calF : \\ P^* G^* = 0}} \frac{1}{2^{(n-r) g}} \\
    &= \frac{\eta(G^*)^r}{\binom{n}{t}^r \cdot 2^{(n-r) g}} \cdot \Pr_{P^* \from \F_2^{r \times n}}[P^* \text{ has full rank} \mid P^*G^*=0].
\end{align*}
Let $\calP(G^*) = \{P^* \in \calS_{t,r,n} \mid P^* G^* = 0\}$. The above two inequalities give us a point-wise approximation to the density of $\calD$,
\begin{equation} \label{eq:prob-bounds}
    \Pr_{P^* \from \calP(G^*)}[P^* \text{ has full rank}] \cdot \frac{\eta(G^*)^r}{\binom{n}{t}^r \cdot 2^{(n-r) g}} \le \Pr_{G \from \calD}[G=G^*] \le \frac{\eta(G^*)^r}{\binom{n}{t}^r \cdot 2^{(n-r) g}}.
\end{equation}
We complete the proof with the following three facts:
\begin{itemize}
    \item \Cref{lemma:uniform-approx} is a general statistical fact, which implies that if
    \[
        \Pr_{G^* \from \F_2^{n \times g}}\left[\abs{\Pr_{G \from \calD}[G=G^*] - \frac{1}{2^{n g}}} \le \frac{\negl[n]}{2^{n g}}\right] \ge 1-\negl[n],
    \]
    then $\calD$ is $\negl[n]$-close to uniform in statistical distance. Crucially, \Cref{lemma:uniform-approx} and \Cref{eq:prob-bounds} reduce the problem to reasoning about \emph{the uniform distribution} over $G^*$, rather than $\calD$.
    \item \Cref{lemma:eta-concentration} implies that
    \[
        \Pr_{G^* \from \F_2^{n \times g}}\left[\abs{\frac{\eta(G^*)^r}{\binom{n}{t}^r \cdot 2^{(n-r) g}} - \frac{1}{2^{n g}}} \le \frac{\negl[n]}{2^{n g}}\right] \ge 1-\negl[n].
    \]
    The proof is a simple invocation of Chebyshev's inequality.
    \item \Cref{lemma:full-rank} implies that
    \[
        \Pr_{G^* \from \F_2^{n \times g}}\left[\Pr_{P^* \from \calP(G^*)}[P^* \text{ has full rank}] \ge 1-\negl[n]\right] \ge 1-\negl[n].
    \]
    
\end{itemize}
Using \Cref{lemma:eta-concentration,lemma:full-rank} with \Cref{eq:prob-bounds}, the condition of \Cref{lemma:uniform-approx} is satisfied. This completes the proof of the theorem.
\end{proof}

\begin{lemma} \label{lemma:uniform-approx}
    Let $p$ be a probability distribution on a finite set $\calX$. If
    \[
        \Pr_{x \from \calX}\left[\abs{p(x) - \frac{1}{\abs{\calX}}} > \frac{\alpha}{\abs{\calX}}\right] \le \beta,
    \]
    then the statistical distance between $p$ and the uniform distribution is at most $\alpha + \beta$.
\end{lemma}
\begin{proof}
    We just compute
    \begin{align*}
        & \sum_{\substack{x \in \calX : \\ p(x) < \frac{1}{\abs{X}}}} \left(\frac{1}{\abs{X}} - p(x)\right) \\
        &= \sum_{\substack{x \in \calX : \\ \frac{1-\alpha}{\abs{X}} \le p(x) < \frac{1}{\abs{X}}}} \left(\frac{1}{\abs{X}} - p(x)\right) + \sum_{\substack{p(x) < \frac{1-\alpha}{\abs{X}}}} \left(\frac{1}{\abs{X}} - p(x)\right) \\
        &\le \alpha + \sum_{\substack{p(x) < \frac{1-\alpha}{\abs{X}}}} \frac{1}{\abs{X}} \\
        &\le \alpha + \beta. 
    \end{align*} \qedhere
\end{proof}

\begin{lemma} \label{lemma:eta-concentration}
    For any $\varepsilon > 0$,
    \[
        \Pr_{G^* \from \F_2^{n \times g}}\left[\abs{\eta(G^*) - \binom{n}{t} \cdot 2^{-g}} > \varepsilon \cdot \binom{n}{t} \cdot 2^{-g}\right] \le \frac{2^g}{\binom{n}{t} \cdot \varepsilon^2}.
    \]
\end{lemma}
\begin{proof}
    By linearity of expectation,
    \[
        \E_{G^* \from \F_2^{n \times g}} \eta(G^*) = \sum_{w \in \calS_{t,n}} \Pr_{G^* \from \F_2^{n \times g}}[w^T G^* = 0] = \binom{n}{t} \cdot 2^{-g}.
    \]
    Furthermore, $\{\mathbbm{1}[w^T G^* = 0]\}_{w \in \calS_{t,n}}$ are pairwise independent random variables over $G^* \from \F_2^{n \times g}$, so the lemma follows from Chebyshev's inequality.
\end{proof}

\begin{lemma} \label{lemma:full-rank}
    If $r \le 0.99 n$ and $\omega(\sqrt{\log n}) \le t \le O(\log n)$, then there is a $\tilde{g} = \Omega(t^2)$ such that for all $g \le \tilde{g}$,
    \[
        \Pr_{G^* \from \F_2^{n \times g}}\left[\Pr_{P^* \from \calP(G^*)}[P^* \text{\emph{ has full rank}}] \ge 1-\negl[n]\right] \ge 1-\negl[n].
    \]
\end{lemma}
\begin{proof}
We say that a subset of rows $S \subseteq [r]$ of $P^*$ forms a ``simple dependency'' if those rows sum to $0$, and no subset of them sums to $0$. For a $1-\negl[n]$ fraction of $G^* \from \F_2^{n \times g}$, we will show that over $P^* \from \calP(G^*)$ the expected number of simple dependencies in $P^*$ is $\negl[n]$.

Letting $W_i$ denote the $i$th row of the matrix $W$, we define
\[
    \SD_\ell = \left\{W \in \calS_{t,\ell,n} \ \middle\vert \ \bigoplus_{i \in [\ell]} W_i = 0 \text{ and } \forall T \subsetneq [\ell], \bigoplus_{i \in T} W_i \ne 0\right\}
\]
and
\[
    \SD_\ell(G^*) = \{W \in \SD_\ell \mid W G^* = 0\}.
\]
For any $G^*$, over $P^* \from \calP(G^*)$ the expected number of simple dependencies in $P^*$ is
\begin{align*}
    \E_{P^* \from \calP(G^*)} \sum_{\ell=1}^r \sum_{S \in \binom{[r]}{\ell}} \mathbbm{1}[P^*_S \in \SD_\ell] &= \sum_{\ell=1}^r \sum_{S \in \binom{[r]}{\ell}} \Pr_{P^* \from \calP(G^*)}[P^*_S \in \SD_\ell] \\
    &= \sum_{\ell=1}^r \binom{r}{\ell} \frac{\abs{\SD_\ell(G^*)}}{\eta(G^*)^\ell}. \addtocounter{equation}{1}\tag{\theequation} \label{eq:num-simp-deps}
\end{align*}
Now since any $\ell-1$ of the rows of any simple dependency are linearly independent, we have
\[
    \Pr_{G^* \from \F_2^{n \times g}}[W G^* = 0] = \frac{1}{2^{(\ell-1) g}}
\]
for any $W \in \SD_\ell$. Therefore,
\begin{align*}
    \E_{G^* \from \F_2^{n \times g}} \abs{\SD_\ell(G^*)} &= \sum_{W \in \SD_\ell} \Pr_{G^* \from \F_2^{n \times g}}[W G^* = 0] \\
    &= \frac{\abs{\SD_\ell}}{2^{(\ell-1) g}}.
\end{align*}
By Markov's inequality it follows that for any $q > 0$, with probability $1-1/q$ over $G^* \from \F_2^{n \times g}$,
\[
    \abs{\SD_\ell(G^*)} \le \frac{\abs{\SD_\ell}}{2^{(\ell-1) g}} \cdot q.
\]
Together with \Cref{lemma:eta-concentration}, we have that for any $q, \varepsilon > 0$ the expected number of simple dependencies computed in \Cref{eq:num-simp-deps} is at most
\begin{align*}
    \sum_{\ell=1}^r \binom{r}{\ell} \frac{\abs{\SD_\ell(G^*)}}{\eta(G^*)^\ell} &\le q \sum_{\ell=1}^r \binom{r}{\ell} \frac{\abs{\SD_\ell}}{\eta(G^*)^\ell \cdot 2^{(\ell-1) g}} \\
    &\le \frac{q}{(1-\varepsilon)^n} \sum_{\ell=1}^r \binom{r}{\ell} \frac{2^{g\ell} \abs{\SD_\ell}}{\binom{n}{t}^\ell 2^{(\ell-1) g}} \\
    &= \frac{q \cdot 2^g}{(1-\varepsilon)^n} \sum_{\ell=1}^r \binom{r}{\ell} \frac{\abs{\SD_\ell}}{\binom{n}{t}^\ell}
\end{align*}
with probability $1-\frac{1}{q}-\frac{2^g}{\binom{n}{t} \cdot \varepsilon^2}$ over $G^*$. Since
\[
    \frac{\abs{\SD_\ell}}{\binom{n}{t}^\ell} \le \Pr_{w_1, \dots, w_\ell \from \calS_{t,n}}[w_1 \oplus \cdots \oplus w_\ell = 0],
\]
the remainder of the proof is devoted to showing that
\begin{equation} \label{eq:no-dependencies}
    \sum_{\ell=1}^r \binom{r}{\ell} \Pr_{w_1, \dots, w_\ell \from \calS_{t,n}}[w_1 \oplus \cdots \oplus w_\ell = 0] \le 2^{-c t^2}.
\end{equation}
for some constant $c > 0$. Setting $\tilde{g} = c t^2/4$, $q = 2^{ct^2/4}$, and $\varepsilon = 2^{-t \log(n/t)/8}$ will then complete the proof of \Cref{lemma:full-rank}.

Let $A$ be the transition matrix for the random walk on $\F_2^n$ where, at each step, we sample a random $w \from \calS_{t,n}$ and move $x \mapsto x \oplus w$. Observe that $\Pr_{w_1, \dots, w_\ell \from \calS_{t,n}}[w_1 \oplus \cdots \oplus w_\ell = 0]$ is equal to the probability that $\ell$ steps of this walk form a (not necessarily simple) cycle, i.e.,
\[
    \Pr_{w_1, \dots, w_\ell \from \calS_{t,n}}[w_1 \oplus \cdots \oplus w_\ell = 0] = \frac{1}{2^n} \Tr(A^\ell).
\]

Let $H$ be the transition matrix for the random walk on the hypercube graph on $\F_2^n$. In \Cref{claim:sparse-to-hypercube}, we bound the probability that $\ell$ steps of $A$ form a cycle in terms of the probability that $\ell t$ steps of $H$ form a cycle.

\begin{numberedclaim} \label{claim:sparse-to-hypercube}
If $\ell = O(n)$ and $t = o(n^{1/3})$, then $\Tr(A^\ell) \le O\left(e^{t^2 \ell/n} \cdot \Tr(H^{\ell t})\right)$.
\end{numberedclaim}
\begin{proof}
    Let $\calE_H$ denote the event that $\ell t$ steps of $H$ form a cycle, and let $\calE_A$ denote the event that $\ell$ steps of $A$ form a cycle. Observe that 
    $$\frac{1}{2^n} \Tr(A^\ell) = \Pr[\calE_A] = \Pr\left[\calE_H \ | \ \text{every } t \text{-block is distinct} \right]$$
    where we consider the walk in $H$ as consisting of $\ell$ consecutive blocks of $t$ steps each, and we say a block is ``distinct'' if a different index is changed in each step.

    Furthermore,
    $$\frac{1}{2^n} \Tr(H^{\ell t}) = \Pr[\calE_H] \ge \Pr\left[\calE_H \big| \ \text{every } t \text{-block is distinct} \right] \cdot \Pr[\text{every } t \text{-block is distinct}]$$
    So $\Pr[\calE_A] \le \frac{\Pr[\calE_H]}{\Pr[\text{every } t \text{-block is distinct}]}$.
    Thus, it suffices to see that
    \begin{align*}
        \Pr[\text{every } t \text{-block is distinct}] &= \Pr[\text{a given } t \text{-block is distinct}]^\ell \\
        &\ge (1-t/n)^{t \ell} \\
        &\ge \left(e^{-t} \left(1 - \frac{t^2}{n}\right)\right)^{t \ell/n} \\
        &\ge e^{-t^2 \ell/n} \cdot (1-t^3 \ell/n^2) \\
        &\ge \Omega(e^{-t^2 \ell/n}).
    \end{align*} \qedhere
\end{proof}

Applying \Cref{claim:sparse-to-hypercube}, we have reduced the problem of proving \Cref{eq:no-dependencies} to showing that
\begin{equation} \label{eq:no-dependencies-reformulated}
    \sum_{\ell=1}^r \binom{r}{\ell} \cdot \frac{e^{t^2 \ell/n}}{2^n} \Tr(H^{\ell t}) \le 2^{-\Omega(t^2)}.
\end{equation}
Fortunately, the eigenvalues of $H$ are simple to compute:
\begin{equation} \label{eq:hypercube-eigenvalues}
    \frac{1}{2^n} \Tr(H^{\ell t}) = \E_{x \from \F_2^n} \left[\left(1 - \frac{2 \wt(x)}{n}\right)^{\ell t}\right]
\end{equation}
where $\wt(x)$ denotes the Hamming weight of $x \in \F_2^n$. We analyze this expression separately depending on how large $\ell$ is.

\paragraph{Small $\ell$ ($\ell \le (e-\Omega(1)) \cdot n/t)$.} We can rewrite \Cref{eq:hypercube-eigenvalues} as a moment of a simple random walk:
\[
\E_{x \from \F_2^n} \left[\left(1 - \frac{2 \wt(x)}{n}\right)^{\ell t}\right] = n^{-\ell t} \E_{X_1, \dots, X_n \from \{1,-1\}} \left[\left(\sum_{i=1}^n X_i\right)^{\ell t}\right].
\]
Let $X = \sum_{i=1}^n X_i$ where $X_1, \dots, X_n \from \{1,-1\}$. For small $\ell$, the Gaussian approximation to these moments is good enough. For even $p$,
\begin{align*}
    \E \left[X^p\right] &= \sum_{i_1, \dots, i_p \in [n]} \E[X_{i_1} \cdots X_{i_p}] \\
    &= \abs{\{(i_1, \dots, i_p) \in [n]^p \mid \text{ each $i_j$ appears an even number of times}\}} \\
    &\le n^{p/2} \cdot (p-1)!!
\end{align*}
For odd $p$, $\E[X^p] = 0$. Therefore, for $\ell = O(n/t)$ we have
\begin{align*}
    \binom{r}{\ell} \cdot \frac{e^{t^2 \ell/n}}{2^n} \Tr(H^{\ell t}) &\le e^{O(t)} \binom{n}{\ell} \cdot \frac{(\ell t-1)!!}{n^{\ell t/2}} \\
    &\le e^{O(t)} \left(\frac{e n}{\ell}\right)^\ell \cdot \left(\frac{\ell t}{e n}\right)^{\ell t / 2} \cdot O(\sqrt{\ell t}) \\
    &= e^{O(t)} \left(\frac{\ell}{e n}\right)^{\ell (t/2-1)} \cdot t^{\ell t/2}\cdot O(\sqrt{\ell t}) \\
\intertext{If $\ell \le t$, then this yields}
    \binom{r}{\ell} \cdot \frac{e^{t^2 \ell/n}}{2^n} \Tr(H^{\ell t}) &\le e^{O(t)} \left(\frac{t^{t-1}}{(e n)^{t/2-1}}\right)^\ell = n^{-\Omega(t)}
\intertext{and if $t < \ell \le (e-\Omega(1)) \cdot n/t$,}
    \binom{r}{\ell} \cdot \frac{e^{t^2 \ell/n}}{2^n} \Tr(H^{\ell t}) &\le e^{O(t)} \left(\frac{e-\Omega(1)}{e t}\right)^{\ell (t/2-1)} \cdot t^{\ell t/2}\cdot O(\sqrt{\ell t}) = 2^{-\Omega(t^2)}.
\end{align*}
In either case we have a bound of $2^{-\Omega(t^2)}$.

\paragraph{Large $\ell$ ($\ell \ge (e-o(1)) \cdot n/t$).} Using the binomial theorem,
\begin{align*}
        \frac{1}{2^n} \Tr(H^{\ell t}) &= \E_{x \from \F_2^n} \left[\left(1 - \frac{2 \wt(x)}{n}\right)^{\ell t}\right] \\
        &= \frac{1}{2^n} \sum_{s=0}^n \binom{n}{s} \cdot \left(1-\frac{2 s}{n}\right)^{\ell t} \\
        &\le \frac{1}{2^n} \sum_{s=0}^n \binom{n}{s} \cdot e^{-2 s \ell t /n} \\
        &= \left(\frac{1 + e^{-2 \ell t/n}}{2} \right)^n.
\end{align*}
For $\ell \ge (e-o(1)) \cdot n/t$, this is at most $\left(\frac{1 + e^{o(1)-2e}}{2}\right)^n$. Since $\left(\frac{1 + e^{-2e}}{2}\right) < 2^{-0.99}$, we have for $r \le 0.99n$ that
\[
    \sum_{\ell=1}^r \binom{r}{\ell} \cdot \frac{e^{t^2 \ell/n}}{2^n} \Tr(H^{\ell t}) \le 2^r \cdot e^{t^2} \cdot \left(\frac{1 + e^{o(1)-2e}}{2}\right)^n = 2^{-\Omega(n)}. \qedhere
\]
\end{proof}

\begin{theorem} \label{theorem:ldpc-prc-lpn}
    \sloppy
    For any $p, \eta \in (0,1/2)$, there exists $g = \Omega(\log^2 n)$, $t = \Theta(\log n)$ such that $\LDPCPRC_0[n, g, t, 0.99 n, \eta, (0.99n)^{-1/4}]$ is a zero-bit public-key PRC (\Cref{def:pkPRC}) that is robust to every $p$-bounded channel, where pseudorandomness rests on the $\LPN_{g, \eta}$ assumption (\Cref{assumption:LPN}).
\end{theorem}
\begin{proof}
    By \Cref{lemma:zero-bit-ldpc-decoding}, there exists $t = \Theta(\log n)$ such that for any $g > 0$, $\LDPCPRC_0[n, g, t, 0.99n, \eta, (0.99n)^{-1/4}]$ is robust to every $p$-bounded channel.
    
    By \Cref{lemma:low-rate-random-ldpc}, there exists a function $g = \Omega(\log^2 n)$ such that the generator matrix of this code --- that is, $G$ for $(P,G) \from \LDPC[n,g,t,0.99n]$ --- is $\negl[n]$-close to uniformly random.

    Since $G$ has dimensions $n \times g$, the $\LPN_{g,\eta}$ assumption implies that $\LDPCPRC_0[n, g, t, 0.99n, \eta, (0.99n)^{-1/4}]$ is a public-key PRC.
\end{proof}

\ifeprint
\newpage
\else \fi
\newpage

\section{Boosting the rate and robustness of any pseudorandom code} \label{sec:improved-prcs}
\subsection{Multi-bit pseudorandom codes} \label{subsec:many-bit}
We show how to construct a multi-bit PRC from any zero-bit PRC.
The high-level idea is to encode a given message bit-by-bit, where we use codewords from the zero-bit PRC to represent bits of the message that are 1, and uniformly random strings to represent bits of the message that are 0.
We say this encoding consists of many \emph{blocks}, where each block is either a codeword from the zero-bit PRC or a random string.
Since our zero-bit PRC allows a decoder with the secret key to distinguish uniform strings from noisy codewords, we can use this decoder to recover each bit of the message from each block.

However, as described so far, there are two issues with this scheme. 
The first is that this scheme encodes the all-0 string as a uniformly random string, but the error correction property of a PRC requires that the decoder can distinguish encodings (of any message) from random strings.
Thus, we modify the above scheme to append a codeword from the zero-bit PRC to the end of every encoding.

The other issue is that this scheme may lose the zero-bit PRC's robustness.
In particular, consider a zero-bit PRC that is robust to all $p$-bounded channels; we would like for our multi-bit PRC to retain this robustness.
However, consider the channel that flips each bit in only the first block of the encoding, independently with probability $\frac{1}{2}$. 
The decoder will now be unable to recover the first bit of the message, and this channel is $p$-bounded since it changes only a sub-constant fraction of the bits.
The issue here is that our $p$-bounded channel was not bounded at all on the first block; we'd like for the channel's effect on every block of the encoding to be $p$-bounded.
We solve this issue by randomly permuting the encoding, which ensures that the errors introduced by the channel cannot be too concentrated in any block. 
The decoder now inverts the permutation before decoding block-by-block as before.

Although the constructions are essentially the same, we separately present multi-bit secret-key and public-key PRCs as they have slightly different syntax. 

\begin{construction}[Multi-bit secret-key PRC] \label{const:multi-bit-sk}
Let $\PRC_0 = (\KeyGen_0, \Encode_0, \Decode_0)$ be a zero-bit secret-key PRC with block length $n$. We define an $\ell$-bit secret-key PRC, $\PRC_\ell = (\KeyGen_\ell, \Encode_\ell, \Decode_\ell)$, as follows:
\begin{itemize}
    \item $\KeyGen_\ell(1^\secpar)$: Sample $\sk_0 \gets \KeyGen_0(1^\secpar)$. Sample a random permutation $\pi : [n \cdot (\ell+1)] \to [n \cdot (\ell+1)]$. Output $\sk = (\sk_0, \pi)$.
    \item $\Encode_\ell(\sk, \m)$: Given as input a message $\m \in \{0,1\}^\ell$, for each $i \in [\ell + 1]$, let
    \begin{align*}
        c_i = \begin{cases}
            c \from \{0,1\}^n &\text{ if } \m_i = 0 \text{ and } i \neq \ell + 1\\
            c \gets \Encode_0(\sk_0) &\text{ if } \m_i = 1 \text{ and } i \neq \ell + 1\\
            c \gets \Encode_0(\sk_0) &\text{ if } i = \ell+1
        \end{cases}
    \end{align*}
    Let $y = c_1 || \ldots || c_{\ell+1} \in \{0,1\}^{n(\ell+1)}$.
    Output $x = y_{\pi(1)} || \ldots || y_{\pi(n(\ell + 1))}$.
    \item $\Decode_\ell(\sk, x_1 || \ldots || x_{n(\ell+1)})$:
    For each $i \in [\ell + 1]$, let $\hat{y}_i = x_{\pi^{-1}(1 + (i-1) n)} || x_{\pi^{-1}(2 + (i-1) n)} || \ldots || x_{\pi^{-1}(n + (i-1) n)}$.
    \begin{align*}
        \hat{x}_i = \begin{cases}
            1 &\text{ if } \Decode_0(\sk_0, \hat{y}_i) = 1\\
            0 &\text{ otherwise }
        \end{cases}
    \end{align*}    
    If $\hat{x}_{\ell+1} \neq 1$, output $\bot$. Otherwise, output $\hat{\m} = \hat{x}_{\pi^{-1}(1)} || \ldots || \hat{x}_{\pi^{-1}(\ell)}$.
\end{itemize}
\end{construction}

\begin{construction}[Multi-bit public-key PRC] \label{const:multi-bit-pk}
Let $\PRC_0 = (\KeyGen_0, \Encode_0, \Decode_0)$ be a zero-bit public-key PRC with block length $n$. We define a $\ell$-bit public-key PRC, $\PRC_\ell = (\KeyGen_\ell, \Encode_\ell, \Decode_\ell)$, as follows:
\begin{itemize}
    \item $\KeyGen_\ell(1^\secpar)$: Sample $(\sk_0, \pk_0) \gets \KeyGen_0(1^\secpar)$. Sample a random permutation $\pi : [n \cdot (\ell+1)] \to [n \cdot (\ell+1)]$. Output $\sk = (\sk_0, \pi)$ and $\pk = (\pk_0, \pi)$.
    \item $\Encode_\ell(\pk, m)$: Given as input a message $\m \in \{0,1\}^\ell$, for each $i \in [\ell + 1]$, let
    \begin{align*}
        c_i = \begin{cases}
            c \from \{0,1\}^n &\text{ if } \m_i = 0 \text{ and } i \neq \ell + 1\\
            c \gets \Encode_0(\pk_0) &\text{ if } \m_i = 1 \text{ and } i \neq \ell + 1\\
            c \gets \Encode_0(\pk_0) &\text{ if } i = \ell+1
        \end{cases}
    \end{align*}
    Let $y = c_1 || \ldots || c_{\ell+1} \in \{0,1\}^{n(\ell+1)}$.
    Output $x = y_{\pi(1)} || \ldots || y_{\pi(n(\ell + 1))}$.
    \item $\Decode_\ell(\sk, x_1 || \ldots || x_{n(\ell+1)})$:
    For each $i \in [\ell + 1]$, let $\hat{y}_i = x_{\pi^{-1}(1 + (i-1) n)} || x_{\pi^{-1}(2 + (i-1) n)} || \ldots || x_{\pi^{-1}(n + (i-1) n)}$.
    \begin{align*}
        \hat{x}_i = \begin{cases}
            1 &\text{ if } \Decode_0(\sk_0, \hat{y}_i) = 1\\
            0 &\text{ otherwise }
        \end{cases}
    \end{align*}    
    If $\hat{x}_{\ell+1} \neq 1$, output $\bot$. Otherwise, output $\hat{\m} = \hat{x}_{\pi^{-1}(1)} || \ldots || \hat{x}_{\pi^{-1}(\ell)}$.
\end{itemize}
\end{construction}

\begin{numberedclaim} \label{claim:many-bit-from-zero}
    If $\PRC_0 = (\KeyGen_0, \Encode_0, \Decode_0)$ is a zero-bit secret-key (public-key) PRC robust to $p$-bounded channels, $\PRC_\ell$ is a $\ell$-bit secret-key (public-key) PRC robust to $(p-\varepsilon)$-bounded channels for any constant $\varepsilon \in (0,p)$.
\end{numberedclaim}

\begin{proof}

\par{\bf Error correction.}
We first show that with overwhelming probability, a $(p-\varepsilon)$-bounded channel applied to a codeword of $\Encode_\ell$ is $p$-bounded on each of its blocks.
Since $\PRC_0$ is robust against $p$-bounded channels, this implies that each block will be correctly decoded by $\Decode_0$.
First, by definition of $p$-boundedness, the entire codeword from $\PRC_\ell$ will have at most $(p-\varepsilon)n(\ell+1)$ errors with overwhelming probability.
Let $\alpha \leq p-\varepsilon$ be the actual fraction of errors.
Now, fix any $n$-length block of the codeword after the permutation has been inverted.
Observe that over the randomness of the permutation, the number of errors in this block is a random variable $X \sim \Hyp(n(\ell+1), \alpha n(\ell+1), n)$.
By \Cref{theorem:hyper},
$$\Pr\left[X \geq (\alpha+\varepsilon)n \right] \leq e^{-2 \varepsilon^2 n},$$
which is negligible in $n$.
Since $\alpha \leq p - \varepsilon$,
the probability that there are at least $pn$ errors in our block is at most the probability that there are $(\alpha + \varepsilon)n$ errors, which we've just shown is negligible.
By a union bound, the probability that \emph{any} block has more than $pn$ errors is negligible.

We now show the second property of error correction, that an unrelated string $c$ decodes to $\bot$ with overwhelming probability.
Consider the last block $\hat{y}_{\ell+1}$ of $c$ after the permutation has been inverted.
By error correction of $\PRC_0$,
the probability over $\sk_0$ that $\Decode_0(\sk_0, \hat{y}_{\ell+1}) = 1$ is negligible.
Therefore, with overwhelming probability, $\Decode_\ell(\sk, c) = \bot$.

\par{\bf Pseudorandomness.}
We prove pseudorandomness of public-key $\PRC_\ell$. We observe that the same proof holds for secret-key $\PRC_\ell$, by omitting the public keys as input to the adversaries $\adv$ and $\bdv$, and instead giving them oracle access to $\Encode_0(\sk_0, \cdot)$ and $\Encode_\ell(\sk, \cdot)$ respectively.

Suppose that an adversary $\adv$ given $\pk$ can distinguish between oracle access to $\Encode_\ell(\pk, \cdot)$ and the uniform distribution.
Then $\bdv$ given $\pk_0$ can distinguish between oracle access to $\Encode_0$ and the uniform distribution as follows. 
$\bdv$ samples a random permutation $\pi : [n(\ell+1)] \to [n(\ell+1)]$ and gives $\pk = (\pk_0, \pi)$ to $\adv$ as input.
Whenever $\adv$ queries $\m$ to its oracle, $\bdv$ computes a response $x$ by following $\Encode_\ell$, but querying its own oracle when $\Encode_\ell$ requires a call to $\Encode_0$, and using $\pi$ as the permutation.
Observe that if $\bdv$'s oracle is the uniform distribution, the resulting $x$ is permutation of a uniform string, so $x$ is uniform.
If $\bdv$'s oracle is $\Encode_0$, the resulting $x$ is drawn from $\Encode_\ell(\sk, \m)$.
Therefore, $\bdv$'s advantage is exactly $\adv$'s advantage.
\end{proof}

\begin{remark*}
    Unfortunately, \Cref{const:multi-bit-sk,const:multi-bit-pk} yield codes of rate roughly $1/n$: If the underlying zero-bit PRC has block length $n$, then in order to encode a $\ell$-bit message we need codewords of length $n \cdot (\ell+1)$. Ideally, we would like to have constant rate --- i.e., to encode $\ell$-bit messages into $O(\ell)$-bit codewords. While typical LDPC codes have constant rate, and there is a simple modification of our pseudorandom LDPC codes that have constant rate, at densities of $\Omega(\log n)$ it is not known how to decode from a constant rate of errors. Instead, we build constant-rate PRCs generically from any multi-bit PRC in \Cref{subsec:constant-rate-prcs}.
\end{remark*}

\subsection{Constant-rate pseudorandom codes}
\label{subsec:constant-rate-prcs}

We now show how to build constant-rate PRCs from any multi-bit PRC (as in \Cref{subsec:many-bit}) and any constant-rate error-correcting code.
We state the public-key version of the following construction only; as usual, the secret-key version is similar.

\begin{construction}[Constant-rate public-key PRC] \label{const:constant-rate-prcs}
    Let $\secpar$ be a security parameter and $\PRC_\secpar$ be a $\secpar$-bit public-key PRC with block length $n'$. Let $(\Enc, \Dec)$ be any error-correcting code with block length $n > \secpar$ and messages of length $k$. Let $\PRG : \{0,1\}^\secpar \to \{0,1\}^n$ be any pseudorandom generator. We define $\PRC[\PRC_\secpar, (\Enc, \Dec)]$ which is a $k$-bit public-key PRC as follows:
    \begin{itemize}
        \item $\PRC.\KeyGen(\secparam)$: Sample $(\sk',\pk') \gets \PRC_\secpar.\KeyGen(1^\secpar)$. Sample a random permutation $\pi : [n' + n] \to [n' + n]$. Let $\sk = (\sk', \pi)$ and $\pk = (\pk', \pi)$; output $(\sk, \pk)$.
        \item $\PRC.\Encode(\pk,\m)$: Given as input a message $\m \in \{0,1\}^k$, let $r \gets \{0,1\}^\secpar$, and let
        \[
            x \gets \PRC_\secpar.\Encode(\pk, r) || \PRG(r) \oplus \Enc(\m).
        \]
        Let $x_1, \ldots, x_{n' + n}$ denote the bits of $x$.
        The output is $x_{\pi(1)} || x_{\pi(2)}|| \ldots || x_{\pi(n' + n)}$.
        \item $\PRC.\Decode(\sk, c)$: Letting $c_1, \ldots, c_{n' + n}$ denote the bits of $c$, the decoder first computes
        \[
            y = c_{\pi^{-1}(1)} || \ldots || c_{\pi^{-1}(n' + n)}.
        \]
        The decoder then parses $y$ as $y_1 || y_2$, where $\abs{y_1} = n'$ and $\abs{y_2} = n$. It computes $r \gets \PRC_\secpar.\Decode(\sk, y_1)$ and outputs $\m = \Dec(\PRG(r) \oplus y_2)$.
    \end{itemize}
\end{construction}
The block length of the resulting PRC is $n + n'$, so the rate is $k / (n+n')$. Since $n'$ is only a function of the security parameter $\secpar$ and does not depend on $k$, for large $k$ the rate approaches the rate $k/n$ of the underlying error-correcting code $(\Enc, \Dec)$.

\begin{theorem} \label{theorem:constant-rate-prcs}
    Let $\PRC_\secpar$ be a $\secpar$-bit public-key PRC with block length $n'$, and let $(\Enc, \Dec)$ be any error-correcting code with block length $n > \secpar$ and messages of length $k$.
    
    Then $\PRC = \PRC[\PRC_\secpar, (\Enc, \Dec)]$ of \Cref{const:constant-rate-prcs} is a public-key PRC where
    \begin{itemize}
        \item the rate of $\PRC$ is $k/(n+n')$; and
        \item for any constants $p, \varepsilon \in (0, 1/2)$, if $\PRC_\secpar$ and $(\Enc,\Dec)$ are robust to channels that introduce at most a $p$ fraction of errors at random locations, then $\PRC$ is robust to all $(p-\eps)$-bounded channels.
    \end{itemize}
\end{theorem}

\begin{proof}
    \par{\bf Pseudorandomness.} By pseudorandomness of $\PRC_\secpar$, no polynomial-time adversary can distinguish between $\PRC.\Encode$ and a hybrid where $x \gets \PRC_\secpar.\Encode(\pk, r)$ is replaced by a uniformly random $x \gets \{0,1\}^{n'}$.
    By pseudorandomness of the PRG, this is computationally indistinguishable from a hybrid where $\PRG(r)$ is replaced by a uniformly random $s \gets \{0,1\}^n$, making $(x || s \oplus \Enc(\m))$ uniform. 
    Therefore, in this hybrid, which is indistinguishable from oracle access to $\PRC.\Encode$, each query to the oracle outputs an independently drawn uniform string in $\{0,1\}^{n + n'}$.

    \par{\bf Robustness.} We first show that for any $p$-bounded channel $\calE$, if the decoder receives $c \gets \calE(\PRC.\Encode(\pk, \m))$, then the resulting $y_1$ and $y_2$ are equal to $\calE_1(\Encode(r))$ and $\calE_2(\Enc(\m)) \oplus \PRG(r)$ respectively, where $\calE_1$ and $\calE_2$ are both $p$-bounded. 
    Observe that the number of errors introduced by $\calE_1$ is distributed as a hypergoemetric random variable with a population of size $n + n'$, $p(n + n')$ success elements (representing the errors), and $n'$ draws. 
    By \Cref{theorem:hyper}, the probability that $\calE_1$ introduces at least $(p-\eps) n'$ errors is at most $e^{-2\eps^2 n'}$, which is negligible. 
    Similarly, the number of errors introduced by $\calE_2$ is distributed as a hypergoemetric random variable with a population of size $n + n'$, $p(n + n')$ success elements (representing the errors), and $n$ draws.
    By \Cref{theorem:hyper}, the probability that $\calE_2$ introduces at least $(p-\eps) n'$ errors is at most $e^{-2\eps^2 n}$, which is negligible as $n > \secpar$.
    By a union bound, with overwhelming probability both $\calE_1$ and $\calE_2$ introduce errors with at most a rate of $(p-\eps)$, and therefore they are $(p-\eps)$-bounded.

    Furthermore, because of the random permutation $\pi$, the locations of the errors on $y_2$ are random.
    Therefore by robustness of $\PRC_\secpar$, we have that $\Decode(\sk, y_1) = r$ with overwhelming probability; and by error correction of $(\Enc, \Dec)$, $\Dec(\PRG(r) \oplus y_2) = \Dec(\BSC_p(\Enc(\m))) = \m$ with overwhelming probability as well.
\end{proof}

If we instantiate \Cref{const:constant-rate-prcs} with
\begin{itemize}
    \item $\PRC_\secpar$ from \Cref{const:multi-bit-pk}, using the LDPC-based PRCs of \Cref{const:zero-bit-ldpc-prc} as $\PRC_0$; and
    \item the error-correcting codes $(\Enc,\Dec)$ of \cite{ABN+92,NN93,TaShma17},
\end{itemize}
then we obtain constant-rate PRCs with robustness to every $p$-bounded channel for $p \in (0,1/2)$.

\subsection{Pseudorandom codes for the deletion channel} \label{subsec:prcs-for-deletions}

In this section, our primary goal is to construct a PRC, $\PRCdel$, that is robust against a constant-rate deletion channel.
That is, the message can be recovered from an encoding under $\PRCdel$ when each of its bits is deleted independently with some constant probability $p$.
The actual robustness guarantee we obtain is even stronger: we can recover the message even when this constant-rate deletion channel is composed with a constant-rate binary symmetric channel.
Our construction makes black-box use of any PRC that is robust against $p$-bounded channels.

Our high-level idea is to take an encoding of this original PRC and make it redundant, while preserving pseudorandomness.
Let $x = x_1, \dots, x_n$ be a PRC encoding of some message.
We could make $x$ robust to the deletion channel by duplicating each bit $x_i$ some number of times $m$.
However, the result would no longer be pseudorandom.
Instead of duplicating $x_i$, we sample a random $y \in \{0,1\}^m$, conditioned on the majority of the bits of $y$ being $x_i$.
Observe that if $x_i$ is $\Ber\left(\frac{1}{2} \right)$, $y$ is uniform over $\{0,1\}^m$.
Since $x_i$ is a bit of our PRC, it is pseudorandom, and $y$ is indistinguishable from uniform.
The resulting codeword is $y_1, \dots, y_n$, where each $y_i \in \{0,1\}^m$.
We call this the ``majority code.''

Observe that given $y = y_1, \dots, y_n$, one can recover $x$ by computing the majority of each length-$m$ block $y_i$.
If each bit of $y$ is deleted independently with some probability $p$ to yield $y'$, we can recover an approximate version of $x$ as follows.
We partition $y'$ into equal-length blocks $y'_1, \dots, y'_n$, and we compute the majority of each block.
Intuitively, we expect most of the new blocks $y'_i$ to be subsets of the corresponding original blocks $y_i$.
Since the deletions are random, we expect them to preserve the majority for most blocks.
Therefore, the message $x'$ that we recover should be approximately $x$, with some bounded number of bits flipped.
Recalling that $x$ itself was a PRC codeword, and our PRC is robust to bounded-weight channels, we can recover the original message from $x'$.

It may be tempting to apply this majority code to other encoding schemes such as pseudorandom encryption, and the result would indeed be pseudorandom.
However, it would not be robust against a constant-rate deletion channel.
Since decoding from the deletion channel introduces bounded-weight errors, the message that the majority code is applied to must be error-correcting against bounded-weight channels. 
Therefore, our PRC is crucial to this approach. 

We first introduce some notation. Let $\Maj(y)$ denote the majority function for bitstrings $y$, where ties are resolved by choosing a random bit. For a given positive odd integer $m \in \N$, let $\MajEncode_m : \{0,1\} \to \{0,1\}^k$ be a randomized function that takes an input bit $b \in \{0,1\}$ and outputs a random string $y \in \{0,1\}^m$ conditioned on $\Maj(y) = b$. For $x \in \{0,1\}^n$, we define $\MajEncode_m(x) := (\MajEncode_m(x_1), \dots, \MajEncode_m(x_n)) \in \{0,1\}^{nm}$. The decoding algorithm $\MajDecode_n : \{0,1\}^* \to \{0,1\}^n$ evenly partitions the received string into $n$ blocks and outputs a string of the blocks' majorities.

\indent $\MajEncode_m(x)$:
\begin{enumerate}
    \item For each $i \in [\len{x}]$, let $y^i = \MajEncode_m(x_i)$.
    \item Output $(y^1 || \cdots || y^{\len{x}}) \in \{0,1\}^{\len{x} \cdot m}$.
\end{enumerate}

\indent $\MajDecode_n(z)$:
\begin{enumerate}
    \item Partition $z$ into $n$ blocks $z = (z^1 || \cdots || z^n)$ such that $\abs{\len(z^i) - \len(z)/n} \le 1$ for all $i \in [n]$.
    \item Output $(\Maj(z^1), \dots, \Maj(z^n))$.
\end{enumerate}

In \Cref{lemma:deletion-channel}, we show that majority encoding a message allows us to recover an approximation of the message with bounded-weight error after a deletion channel is applied.
This holds even if this deletion channel is composed with a binary symmetric channel.

\begin{lemma} \label{lemma:deletion-channel}
   Assume $m = \Omega(n \log^6 n)$. For any constant deletion rate $p \in (0,1)$ and error rate $q \in (0,1/2)$, there exists a constant $\varepsilon > 0$ such that
    \[
        \Pr_{x \from \{0,1\}^n}[\wt(x \oplus \MajDecode_n \circ \emph{\BSC}_q \circ \emph{\BDC}_p \circ \MajEncode_m(x)) \le (1/2-\varepsilon) \cdot n] \ge 1-\negl[n].
    \]
\end{lemma}
\begin{proof}
    Let $y \from \MajEncode_m(x)$ and let $z \from \BSC_q \circ \BDC_p(y)$. Let $D \subseteq [nm]$ be the set of indices deleted by $\BDC_p$; $D$ is a random variable where each $i \in [nm]$ is included in $D$ independently with probability $p$. Let $f : [nm - \abs{D}] \to [nm]$ be the function that maps indices of $z$ to those of $y$, i.e., $z_j = \BSC_q(y_{f(j)})$.

    Let $\calR = R_1, \dots, R_n \subseteq [nm]$ be the partition of $[nm]$ into $n$ blocks of length $m$. Let $\calS = S_1, \dots, S_n \subseteq [nm - \abs{D}]$ be the partition of $[nm - \abs{D}]$ into almost-equal-sized blocks from the definition of $\MajDecode_n(z)$. $\calR$ and $\calS$ define partitions of $y$ and $z$, respectively. All of 

    \begin{figure}
        \centering
        \makebox[\textwidth][c]{\includegraphics{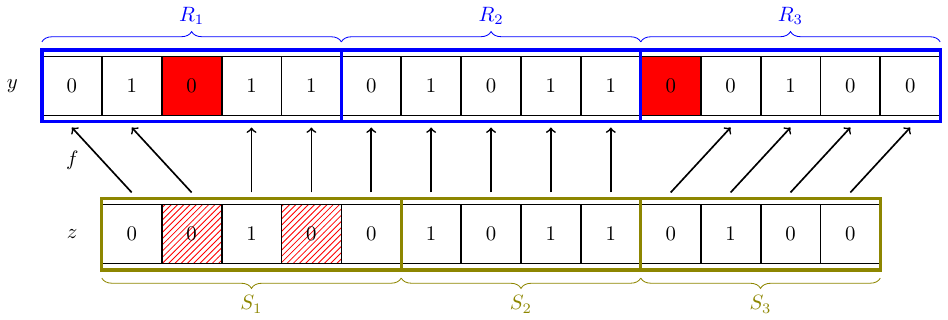}}
        \caption{Objects from the proof of \Cref{lemma:deletion-channel}: $y \from \MajEncode_m(x)$, $z \from \BSC_q \circ \BDC_p(y)$, the function $f$, and the partitions $\calR, \calS$ for $n = 3, m = 5$. In this illustration $x = (1, 1, 0)$ and $\MajDecode_3(z) = (0, 1, 0)$. There are deletions at locations $D = \{3, 11\}$ (indicated by solid red), and there are errors from $\BSC_q$ at locations 2 and 4 (indicated by hatched red). The arrows indicate the mapping $f$, i.e. an arrow points from an index in $z$ to the index in $y$ from which it originated.}
    \end{figure}

    Let $a, b \in [nm]$ be such that $a \le b$. Applying Hoeffding's inequality to $\abs{[a,b] \setminus D} = \sum_{j=a}^b \mathbbm{1}[j \not\in D]$ yields
    \[
        \Pr\left[\Big|(b-a+1)(1-p) - \abs{[a,b] \setminus D}\Big| > \sqrt{nm} \log n\right] \le e^{-\Omega(\log^2 n)}.
    \]
    By a union bound, with probability $1-\negl[n]$ we have that for every $a,b \in [nm]$
    \begin{align}
        \Big|(b-a+1)(1-p) - \abs{[a,b] \setminus D}\Big| \le \sqrt{nm} \log n. \label{eq:ab-length}
    \end{align}
    That is, the number of deletions in every contiguous interval $[a,b]$ of $y$ concentrates and is roughly $p(b-a) \pm \sqrt{nm} \log n$.
    
    We now argue in \Cref{claim:small-misalignment} that because of this concentration, each $S_i$ in the decoder's partition corresponds to a subset of the original message $y$ that is largely contained in $R_i$.
    That is, the symmetric set difference between $f(S_i)$ and $(R_i \setminus D)$ is small.
    \begin{numberedclaim} \label{claim:small-misalignment}
        If \Cref{eq:ab-length} holds for all $a,b \in [nm]$, then for all $i \in [n]$ we have $\abs{f(S_i) \triangle (R_i \setminus D)} = O(\sqrt{nm} \log n)$ (where $\triangle$ denotes the symmetric set difference).
    \end{numberedclaim}
    \begin{proof}
        Observe that $\abs{f(S_i) \triangle (R_i \setminus D)} \le \abs{(i-1) m  - \min f(S_i)} + \abs{im - \max f(S_i)}$.
        We will show that $\abs{(i-1)m - \min f(S_i)} = O(\sqrt{nm} \log n)$. A nearly identical proof shows that $\abs{im - \max f(S_i)} = O(\sqrt{nm} \log n)$ as well, so this will complete the proof.
        
        Setting $[a,b] = [nm]$ to be all of $y$ in \Cref{eq:ab-length}, we have that $\abs{(1-p)nm - \len z} \le \sqrt{nm} \log n$. Therefore, $\abs{(1-p)m - \abs{S_i}} \le \left\lceil \frac{\sqrt{nm} \log n}{n }\right\rceil$ for every $i \in [n]$. It follows that for every $i \in [n]$,
        \begin{equation} \label{eq:min-s}
            \abs{(1-p) (i-1) m - \min S_i} = O(\sqrt{nm} \log n).
        \end{equation}
        For any $j \in [nm-\abs{D}]$, setting $[a,b] = [f(j)]$ in \Cref{eq:ab-length} we have that $\abs{j - (1-p) f(j)} = O(\sqrt{nm} \log n)$. Applying this fact to $j = \min S_i$ for any $i \in [n]$,
        \[
            \abs{\min S_i  - (1-p) \min f(S_i)} = O(\sqrt{nm} \log n).
        \]
        Together with \Cref{eq:min-s} and the fact that $p \in (0,1)$ is a constant, we have
        \[
            \abs{(i-1) m  - \min f(S_i)} = O(\sqrt{nm} \log n)
        \]
        as desired.
    \end{proof}

    So far, we've argued that each block $z^i$ in the decoder's partition consists mostly of bits from block $y^i$ of the encoding.
    We now argue that with high probability, the bits of $y^i$ that were excluded and the bits from other blocks that were erroneously included do not affect the majority of $z^i$.

    To do so, we define $\Delta_i$ to be the difference of the erroneously excluded bits and the erroneously included bits, converted to values in $\{-1, 1\}$ so we can later reason about this as a random walk.
    Let
    \[
        \Delta_i = \sum_{j \in (R_i \setminus D) \setminus f(S_i)} (-1)^{y_j} - \sum_{j \in f(S_i) \setminus (R_i \setminus D)} (-1)^{y_j}.
    \]
    Conditioned on \Cref{eq:ab-length} holding for all $a,b \in [nm]$, \Cref{claim:small-misalignment} implies that $\Delta_i$ is a sum of $O(\sqrt{nm} \log n)$ random $\{1,-1\}$ variables. 
    Furthermore, these random variables are independent because $y$ is uniform over $\{0,1\}^{nm}$.
    By Hoeffding's inequality,
    \begin{equation} \label{eq:delta-small}
        \Pr[\abs{\Delta_i} \geq (nm)^{1/4} \log^{3/2} n] \le 2\exp\left(\frac{-\sqrt{nm} \log^3 n}{O(\sqrt{nm} \log n)} \right) \leq \negl[n].
    \end{equation}
    By a union bound, the probability that $\abs{\Delta_i} \le (nm)^{1/4} \log^{3/2} n$ for all $i \in [n]$ is at least $1 - \negl[n]$.

    We will complete the proof by applying \Cref{claim:bsc-signs} with $N = \abs{S_i}$, $N' = \abs{R_i \cap D}$, and $\Delta = \Delta_i$. The variables $X_j$ are the values of $y_j$ for $j \in f(S_i) \cup (R_i \cap D)$ and $Z_j$ are the errors introduced by $\BSC_q$. Since the sets $f(S_i) \cup (R_i \cap D)$ are disjoint for distinct $i$, the events that $\Maj(z_i) = x_i$ are independent once we condition on the values of $y_j$ for $j \not\in \bigcup_{i \in [n]} f(S_i) \cup (R_i \cap D)$; therefore the result will follow from a Chernoff bound.

    \begin{numberedclaim} \label{claim:bsc-signs}
        Let $N, N' \in \N$ and $\Delta \in \R$ be such that $N' = O(N)$ and $\Delta = O(\sqrt{N})$. Let $X_1, \dots, X_N, Y_1, \dots, Y_{N'}$ be independent, uniform $\{-1, 1\}$ random variables and $Z_1, \dots, Z_N$ be independent $\{-1,1\}$ random variables with expectation $\delta = \Omega(1)$. Then for sufficiently large $N$,
        \[
            \Pr\left[\sgn\left(\sum_{j \in [N]} X_j \cdot Z_j\right) = \sgn\left(\sum_{j \in [N]} X_j + \sum_{j \in [N']} Y_j + \Delta\right)\right] = \frac{1}{2} + \Omega(1),
        \]
        where we define $\sgn(0)$ to be a random value in $\{-1, 1\}$.
    \end{numberedclaim}
    \begin{proof}
        We can equivalently view $Z_j$ as a random variable which is 1 with probability $\delta$, and a uniformly random $\{-1,1\}$ value with probability $1-\delta$. Let $F$ be the set of indices $j \in [N]$ where $Z_j$ is deterministically 1. By a Chernoff bound, $\abs{\delta N - \abs{F}} \le \delta N/2$ with probability $1-\negl[n]$; fix such an $F$. We can write
        \begin{align*}
            \sum_{j \in [N]} X_j \cdot Z_j &= \sum_{j \in F} X_j + \sum_{j \in [N] \setminus F} X_j \cdot Z_j \\
            &=: \sum_{j \in F} X_j + \sum_{j \in [N_a]} X_j^a
        \end{align*}
        where we have defined $N_a := N - \abs{F}$ and $\{-1,1\}$ variables $\{X_j^a\}_{j \in [N_a]} = \{X_j \cdot Z_j\}_{j \in [N] \setminus F}$. We also write
        \begin{align*}
            \sum_{j \in [N]} X_j + \sum_{j \in [N']} Y_j + \Delta &= \sum_{j \in F} X_j + 
            \left(\sum_{j \in [N] \setminus F} X_j + \sum_{j \in [N']} Y_j\right) + \Delta \\
            &=: \sum_{j \in F} X_j + \sum_{j \in [N_b]} X_j^b + \Delta
        \end{align*}
        where we have defined $N_b := N - \abs{F} + N'$ and $\{-1,1\}$ variables $\{X_j^b\}_{j \in [N_b]} = \{X_j\}_{j \in [N] \setminus F} \cup \{Y_j\}_{j \in [N']}$. Observe that $\{X_j\}_{j \in F} \cup \{X_j^a\}_{j \in [N_a]} \cup \{X_j^b\}_{j \in [N_b]}$ are all independent, uniformly random $\{-1,1\}$ variables. We wish to show that
        \[
            \Pr\left[\sgn\left(\sum_{j \in F} X_j + \sum_{j \in [N_a]} X_j^a\right) = \sgn\left(\sum_{j \in F} X_j + \sum_{j \in [N_b]} X_j^b + \Delta\right)\right] = \frac{1}{2} + \Omega(1).
        \]
        Let $X = \sum_{j \in F} X_j$, $A = \sum_{j \in [N_a]} X_j^a$, and $B = \sum_{j \in [N_b]} X_j^b$. The above equation we wish to show becomes
        \[
            \Pr[\sgn(X+A) = \sgn(X+B+\Delta)] = \frac{1}{2} + \Omega(1).
        \]
        If $\abs{X} > \max\{\abs{A}, \abs{B+\Delta}\}$, then $\sgn(X+A) = \sgn(X+B+\Delta) = \sgn(X)$. On the other hand,
        \[
            \Pr\left[\sgn(X+A) = \sgn(X+B+\Delta) \Bigm| \abs{X} \le \max\{\abs{A}, \abs{B+\Delta}\}\right] = \frac{1}{2}.
        \]
        Therefore, it suffices to show that $\Pr[\abs{X} > \max\{\abs{A}, \abs{B+\Delta}\}] = \Omega(1)$. We do this by bounding each of the three factors below separately:
        \ifeprint
        \[
            \Pr\Big[\abs{X} > \max\{\abs{A}, \abs{B+\Delta}\}\Big] \ge \Pr\left[\abs{X} > \sqrt{\abs{F}}+\abs{\Delta}\right] \cdot \Pr\left[\abs{A} \le \sqrt{\abs{F}}+\abs{\Delta}\right] \cdot \Pr\left[\abs{B+\Delta} \le \sqrt{\abs{F}}+\abs{\Delta}\right].
        \]
        \fi
        \begin{itemize}
            \item Since $\sqrt{\abs{F}} + \abs{\Delta} = O(\sqrt{\abs{F}})$, the central limit theorem implies that $\Pr\left[\abs{X} > \sqrt{\abs{F}}+\abs{\Delta}\right] = \Omega(1)$.
            \item Since $N_a = O(\abs{F})$, Hoeffding's inequality implies that $\Pr\left[\abs{A} \le \sqrt{\abs{F}}+\abs{\Delta}\right] \ge \Pr\left[\abs{A} \le \sqrt{\abs{F}}\right] = \Omega(1)$.
            \item By the triangle inequality, $\Pr\left[\abs{B+\Delta} \le \sqrt{\abs{F}}+\abs{\Delta}\right] \ge \Pr\left[\abs{B} \le \sqrt{\abs{F}}\right]$. Since $N_b = O(\abs{F})$, Hoeffding's inequality again implies that $\Pr\left[\abs{B} \le \sqrt{\abs{F}}\right] = \Omega(1)$. \qedhere
        \end{itemize}
    \end{proof}
    It only remains to show that the conditions of \Cref{claim:bsc-signs} are satisfied with probability $1-\negl[n]$. That is, for $N = \abs{S_i}$, we need to show that $N' = \abs{R_i \cap D} = O(N)$ and $\Delta = \Delta_i = O(\sqrt{N})$ with probability $1-\negl[n]$.
    Conditioned on \Cref{eq:ab-length} holding for all $a, b \in [nm]$, we have that $\abs{pm - \abs{R_i \cap D}} \le \sqrt{nm} \log n$ and $\abs{(1-p)m - \abs{S_i}} = O(\sqrt{nm} \log n)$ for all $i \in [n]$. Therefore $\abs{R_i \cap D} = \Theta(m)$ and $\abs{S_i} = \Theta(m)$. By \Cref{eq:delta-small}, $\abs{\Delta_i} = O((nm)^{1/4} \log^{3/2} n)$ for all $i \in [n]$ with probability $1-\negl[n]$. Therefore, $\abs{\Delta_i} = O(\sqrt{\abs{S_i}})$ with probability $1-\negl[n]$ if $m = \Omega(n \log^6 n)$.
\end{proof}

We now show that we can compose this majority code with a PRC for bounded-weight channels to yield a PRC for the deletion channel.

\begin{construction}[PRC for deletions]
Let $\PRC = (\KeyGen, \Encode, \Decode)$ be a PRC with block length $n$.
$\PRCdel[m, \PRC] = (\KeyGendel, \Encodedel, \Decodedel)$ is defined as follows:
\begin{itemize}
    \item $\KeyGendel(1^\secpar)$: Output $\sk \gets \KeyGen(1^\secpar)$.
    \item $\Encodedel(\sk, \m)$: Output $\MajEncode_m(\Encode(\sk,\m))$.
    \item $\Decodedel(\sk, z)$: Output $\Decode(\sk, \MajDecode_n(z))$.
\end{itemize}    
\end{construction}

\begin{theorem} \label{theorem:deletion-code}
    For any constants $p \in (0,1)$ and $q\in (0, 1/2)$, there exists $\varepsilon > 0$ such that the following holds.
    If $\PRC$ is a PRC for $\left(\frac{1}{2} - \varepsilon\right)$-bounded channels with block length $n$,
    there exists an $m \in \N$, where $m = \Omega(n \log^6 n)$, such that $\PRCdel[m, \PRC]$ is a PRC for the channel $\BSC_q \circ \BDC_p$.
\end{theorem}

\begin{proof}
    \par{\bf Error correction.}
    Let $\m$ be any message, and let $x = \Encode(\sk, \m)$ for $\sk \gets \KeyGen(1^\secpar)$.
    By pseudorandomness of $\PRC$, $x$ is indistinguishable from a random string in $\{0,1\}^n$.
    By \Cref{lemma:deletion-channel}, there exists $\varepsilon \in (0,1/2)$ such that
    \[
        \Pr_{x \from \{0,1\}^n}[\wt(x \oplus \MajDecode_n \circ \emph{\BSC}_q \circ \emph{\BDC}_p \circ \MajEncode_m(x)) \le (1/2-\varepsilon \cdot n)] \ge 1-\negl[n].
    \]
    Let $\calE' = \MajDecode_n \circ \BSC_q \circ \BDC_p \circ \MajEncode_m$ be a channel, and observe that the above implies that $\calE'$ is $\left(1/2 - \varepsilon \right)$-bounded.
    
    Since $\PRC$ is error-correcting against $\left(1/2 - \varepsilon\right)$-bounded channels, 
    \[\Pr_{\sk \gets \KeyGen(1^\secpar)} [\Decode(\sk, \calE'(x)) = s : x \gets \Encode(\sk, \m)] \geq 1 - \negl.\]
    Now, we rewrite this guarantee to be in terms of $\PRCdel$. 
    Observe that by definition, for all $\m$, $\Decodedel(\sk, \BSC_q \circ \BDC_p(\Encodedel(\sk,\m)))$ is distributed identically to $\Decode(\sk, \calE'(\Encode(\sk,\m)))$. Therefore, 
    \ifeprint
    \[
    \Pr_{\sk \gets \KeyGendel(1^\secpar)} [\Decodedel(\sk, \BSC_q \circ \BDC_p(x)) = \m : x \gets \Encodedel(\sk, s)] \geq 1 - \negl.
    \]
    \fi
    
    We finally show that unrelated strings $c$ decode to $\bot$ under $\PRCdel$.
    Let $c' = \MajDecode_n(c)$.
    By the analogous property for $\PRC$,
    \[\Pr_{\sk \gets \KeyGen(1^\secpar)} [\Decode(\sk, c') = \bot] \geq 1 - \negl\]
    as desired.

    \par{\bf Pseudorandomness.}
    By pseudorandomness of $\PRC$, its codewords are indistinguishable from random strings from $\{0,1\}^n$.
    Thus, let $x$ be drawn uniformly from $\{0,1\}^n$. The distribution $\MajEncode_m(x)$ is uniform over $\{0,1\}^{nm}$ because for odd $m \in \N$, $\abs{\{y \in \{0,1\}^m : \Maj(y) = 0\}} = \abs{\{y \in \{0,1\}^m : \Maj(y) = 1\}} = 2^{m-1}$.
\end{proof}

\newpage

\section{Application: watermarking for language models} \label{sec:water}
In this section we show that PRCs can be used to build quality-preserving (undetectable) and robust watermarks.
For an overview of our approach, see either the introduction or \Cref{subsec:techo-water}.

\subsection{Watermarking preliminaries} \label{subsec:water-prelims}
We follow \cite{CGZ23} in our definition of a \emph{language model}. We will often refer to language models simply as \emph{models}.
\begin{definition}\label{def:model}
A language model $\Model$ over token set~$\cal{T}$ is a deterministic algorithm that takes as input a prompt $\prompt$ and tokens previously output by the model $\t=(\t_1, \ldots, \t_{i-1})$, and outputs a probability distribution $\p_i = \Model(\prompt, \t)$ over $\cal{T}$.
\end{definition}
We write $\p_i(t)$ to denote the probability that the model outputs a token $t \in \calT$ as specified by the distribution $\p_i$.
When $\p_i$ is supported on $\{0,1\}$, we write $\hat{\p}_i := \E[\p_i] \in [0,1]$.

A language model~$\Model$ is used to generate text as a response to a prompt by iteratively sampling from the returned distribution until a special terminating token~$\done \in \cal{T}$ is drawn. 
We assume an upper bound $L^*$ on the length of the token sequence comprising any response.

\begin{definition}\label{def:modelresp}
    A language model's \emph{response} to~$\prompt$ is a random variable $\RModel(\prompt)\in \mathcal{T}^\star$ that is defined algorithmically as follows.
    We begin with an empty list of tokens~$\t=()$. As long as the last token in~$\t$ is not~$\done$, we draw a token~$\t_i$ from the distribution~$\Model(\prompt,\t)$ and append it to~$\t$. Finally, we set~$\RModel(\prompt) = \t$.
\end{definition}

For a probability distribution~$D$ over elements of a finite set~$X$, we define the Shannon \emph{entropy} of~$D$ as 
\[H(D) = \E_{x\sim D}[-\log D(x)],\]
where~$D(x)$ is the probability of~$x$ in the distribution~$D$.
The \emph{empirical entropy} (also known as Shannon information or surprisal) of~$x$ in~$D$ is simply~$-\log D(x)$.
The expected empirical entropy of~$x\sim D$ is exactly~$H(D)$.
Intuitively, the empirical entropy of~$x$ (with respect to~$D$) is the number of random bits that were required to draw~$x$ out of the distribution~$D$.

\begin{definition}[Empirical entropy]\label{def:empentropy}
    For a language model~$\Model$, a prompt $\prompt$, and a possible response~$\t\in\calT^\star$, we define the \emph{empirical entropy} of~$\RModel$ responding with~$\t$ to~$\prompt$ as
    \[
        \empH(\Model,\prompt,\t) := -\log \Pr\Big[\RModel\left(\prompt\right)=\t\Big].
    \]
\end{definition}

We next generalize the definition of empirical entropy from whole outputs to \emph{substrings} of a model's output. This will quantify how much entropy was involved in the generation of a particular contiguous substring of the output.

\begin{definition}\label{def:subseq-empentropy}
    For a language model~$\Model$, a prompt~$\prompt$, a possible response~$\t\in\calT^\star$, and indices $i, j \in [\len{\t}]$ with $i \le j$ we define the \emph{empirical entropy on substring $[i,j]$} of~$\RModel$ responding with~$\t$ to~$\prompt$ as
    \ifeprint
    \begin{align*}
        \empH^{[i,j]}(\Model,\prompt,\t) := -\log \Pr\Big[&\RModel\left(\prompt\right)[i:j]=\t[i:j] \ \\
        &\Big| \ \RModel\left(\prompt\right)[1:(i-1)] = \t[1:(i-1)]\Big].
    \end{align*}
    \fi
\end{definition}
For convenience, in settings where the model and prompt are clear we simply write $\empH^{[i,j]}(\t)$ to mean $\empH^{[i,j]}(\Model, \prompt,\t)$.
We sometimes write $H_e^{i} := H_e^{[i,i]}$ to denote the empirical entropy of a single token $i$.

We formally define a watermarking scheme as follows.
\begin{definition}[Watermarking scheme] \label{def:watermark}
    A \emph{watermarking scheme} for a model $\Model$ over~$\cal{T}$ is a tuple of polynomial-time algorithms $\calW = (\Setup, \Wat, \Detect)$ where:
    \begin{itemize}
        \item $\Setup(\secparam) \to \sk$ outputs a secret key, with respect to a security parameter~$\secpar$.
        \item $\Wat_{\sk}(\prompt)$ is a randomized algorithm that takes as input a prompt $\prompt$ and generates a response in~$\cal{T}^\star$.
        \item $\Detect_{\sk}(\t) \to \{\true, \false\}$ is an algorithm that takes as input a sequence~$\t\in\cal{T}^\star$ outputs~$\true$ or~$\false$.
    \end{itemize}
\end{definition}

Ideally, $\Detect_{\sk}(\t)$  should output $\true$ if~$\t$ is generated by~$\Wat_{\sk}(\prompt)$, and should output $\false$ if~$\t$ is generated independently of $\sk$. 
The former property is called \emph{completeness} and the latter \emph{soundness}.

\begin{definition}[Soundness] \label{def:wm_sound}
    A watermarking scheme~$\calW$ is~\emph{sound} if for every security parameter $\secpar$ and token sequence~$\t\in \calT^\star$ of length $\poly$,
    \[\Pr_{\sk \gets \Setup(\secparam)}[\Detect_{\sk}(\t) = \true] \leq \negl.\]
\end{definition}

\begin{definition}[Completeness] \label{def:wm_complete}
    A watermarking scheme~$\calW$ is~$b(L)$-\emph{complete} if for every security parameter $\secpar$ and prompt~$\prompt$ of length $\poly$,
    \[
        \Pr_{\substack{\sk \gets \Setup(\secparam) \\ \t \gets \Wat_\sk(\prompt)}}[\Detect_\sk(\t) = \false \text{ and }  \empH\left(\Model,\prompt,\t\right) \geq b\left(\len \t\right)] \le \negl.
    \]
\end{definition}

\begin{definition}[Substring completeness] \label{def:wm_complete_ss}
    A watermarking scheme~$\calW$ is~$b(L)$-\emph{substring-complete} if for every prompt~$\prompt$ and security parameter $\secpar$,
    \begin{align*}
        \Pr_{\substack{\sk \gets \Setup(\secparam) \\ \t \gets \Wat_\sk(\prompt)}}
        \Big[\;
        \exists\; i,L \in [\len{\t}] &\text{ such that }
        \Detect_\sk(\t[i:i+L]) = \false \\
        &\text{ and }  \empH^{[i:i+L]}\left(\Model,\prompt,\t\right) \geq b(L)\;\Big] \le \negl.
    \end{align*}
\end{definition}

If the watermarked model is indistinguishable from the un-watermarked model (and therefore perfectly quality-preserving), we say that the watermarking scheme is \emph{undetectable}.

\begin{definition}[Undetectability] \label{def:wm_undet}
    A watermarking scheme $\calW = (\Setup, \Wat, \Detect)$ is \emph{undetectable} if for every security parameter $\secpar$ and all polynomial-time distinguishers $D$,
    \[
        \abs{\Pr[D^{\Model,\RModel}(\secparam) \to 1] - \Pr_{\sk \gets \Setup(\secparam)}[D^{\Model,\Wat_\sk}(\secparam) \to 1]} \le \negl,
    \]
    where the notation $D^{\calO_1, \calO_2}$ means that $D$ is allowed to adaptively query both~$\calO_1$ and $\calO_2$ with arbitrary prompts.
\end{definition}

We also define robustness for a watermarking scheme. For robustness, the detector should be able to identify the watermark even if the text has undergone corruption by some channel $\calE$. Substring robustness says that the watermark should be robust to both cropping as well as $\calE$.

\begin{definition}[Substring robustness] \label{def:wm_robust_ss}
    A watermarking scheme $\calW = (\Setup, \Wat, \Detect)$ is $b(L)$-\emph{substring-robust} against a channel $\calE$ if for every security parameter $\secpar$ and prompt~$\prompt$ of length $\poly$,
    \begin{align*}
        \Pr_{\substack{\sk \gets \Setup(\secparam) \\ \t \gets \Wat_\sk(\prompt)}}
        \Big[\;
        \exists\; i,L \in [\len{\t}] &\text{ such that }
        \Detect_\sk(\calE(\t[i:i+L])) = \false \\
        &\text{ and }  \empH^{[i:i+L]}\left(\Model,\prompt,\t\right) \geq b(L)\;\Big] \le \negl.
    \end{align*}
\end{definition}

\paragraph{Reducing to a binary alphabet.} 
Recall that a language model operates over an arbitrary token alphabet $\calT$. 
In constructing our PRC-based watermarks, it will be convenient to take $\calT$ to be $\{0,1\}$, since our pseudorandom LDPC codes have binary codewords.
Prior work \cite{CGZ23} suggests a black-box transformation from any model with an arbitrary-token alphabet to a model with a binary-token alphabet, by reinterpreting the probability vectors output by the model.
Here, we describe that transformation in more detail.

Let $\Model$ be any model with token alphabet $\calT$. $\Model$, when queried with a prompt and token sequence, outputs a probability vector $\p \in [0,1]^{\abs{\calT}}$.
We construct a model $\Model'$ that has token alphabet $\calT' = \{0,1\}$; that is, $\Model'$ outputs probability vectors $\p' \in [0,1]^2$.
This construction will make only black-box use of $\Model$ and will ensure that there is some decoding function $f$ converting binary outputs of $\RModel'$ into sequences of tokens from $\calT$, such that the distributions of $\RModel$ and $f\left(\RModel'\right)$ are identical.

Let $\Enc : \calT \to \{0,1\}^*$ be any prefix-free encoding function, and let $\Dec : \{0,1\}^* \to \calT^* \times \{0,1\}^*$ be the corresponding decoding function such that for any $t \in \calT$, $\Dec(\Enc(t)) = (t, \bot)$.
Furthermore, if $\Dec$ is given as input an encoded token $\Enc(t)$ concatenated with some binary string $s \in \{0,1\}^*$ such that no prefix of $s$ is a valid codeword, $\Dec$ outputs the token and the string; that is, $\Dec(\Enc(t)||s) = (t,s)$.
We will construct $\Model'$ such that the distribution of $\Dec\left(\RModel'\right)$ is identical to that of $\RModel$.
That is, $\RModel'$ will output a binary encoding of a token sequence in $\calT$. So $\Model'$ will compute its distribution $\p'$ over the next binary token by querying $\Model$ for its distribution $\p$ over $\calT$, and computing (according to $\p$) the distribution over the next bit of the binary encoding of the next token.

Suppose that $\Model'$ is given as input a prompt $\prompt$ and a binary token sequence $(\t_1', \ldots, \t_{\ell}') \in \{0,1\}^{\ell}$.
$\Model'$ must now output a probability distribution over the next (binary) token, which must match $\Model$ under the appropriate transformation.
Let $((\t_1, \ldots, \t_i), s) \gets \Dec(\t_1', \ldots, \t_{\ell}')$, and let $\p = \Model(\prompt, (\t_1, \ldots, \t_i))$, where $\p \in [0,1]^{\abs{\calT}}$.
$\Model'$ computes $\p'$ as
\begin{align*}
\p'(0) &= \sum_{\substack{t \in \calT\\
\Enc(t)[1:\len{s}+1] = s|| 0}} \p(t)\\
\p'(1) &= 1 - \p'(0)
\end{align*}

where $\Enc(t)[1:\len{s}+1]$ denotes the first $(\len{s} + 1)$ bits of $\Enc(t)$.
Observe that $\p'(0)$ is exactly 
the probability that the next bit of $\Enc(t)$ is 0, where $t \gets \p$ from $\Model$, conditioned on the event that the previous bits of $\Enc(t)$ match $s$.
Therefore, the probability that $t$ is drawn as the next token from $\Model$ is exactly the same as the probability that $\Enc(t)$ is drawn as the next sequence of binary tokens from $\Model'$.

This transformation exactly preserves the empirical entropy of the response, since the distribution of $\Dec\left(\RModel'\right)$ is identical to that of $\RModel$.
However, it does reduce the rate of entropy per token, since responses from $\RModel'$ are longer.
This reduction will depend on the length of codewords, and using an efficient encoding such as a Huffman encoding can help mitigate this rate reduction.

For the remainder of this paper, we assume that the token alphabet of the underlying model is binary. In practice, we can embed the watermark in a model $\Model$ with an arbitrary token alphabet $\calT$ by using this transformation to compute $\Model'$, generating a watermarked binary response from $\Model'$, and transforming this response back to tokens in $\calT$ using the decoding function $D$.
This transformation preserves the undetectability of our watermark.
Recall that by undetectability, the distribution of watermarked (binary) responses from $\Model'$ is indistinguishable from the original (binary) distribution of $\Model'$.
Let $\Wat$ denote the distribution of watermarked responses from $\Model'$.
Since our transformation guarantees that $\Dec \left(\RModel'\right) \equiv \RModel$, 
we have that $\Dec \left(\Wat \right) \approx \RModel$, meaning that our watermarked responses after this transformation are indistinguishable from those of the original model.

Since $\Model'$ outputs distributions $\p$ over $\{0,1\}$, we often specify these distributions by their expectations $\hat{\p} := \E[\p] \in [0,1]$.

\begin{remark*}
    In our robustness theorems, we assume that the bit-errors are random. Using the above transformation, this may not be realistic because bit-errors within the representation of a given token may be correlated. We note that an alternative transformation that avoids this issue simply uses the first bit of the above representation.
\end{remark*}

\subsection{Robust watermarking from PRCs} \label{subsec:water-simple}
Watermarking schemes for language models typically embed the watermark by sampling each token with some bias.
In this section, we show that using a PRC to determine this bias yields a watermarking scheme that inherits robustness from the PRC.
In more detail, our watermark generator first samples a PRC codeword $x \in \{0,1\}^n$.
It then samples each token $\t_i \in \{0,1\}$ to be biased toward $x_i$, yielding a response that is a noisy codeword.
By pseudorandomness of PRCs, this biased sampling does not noticeably change the distribution of $\t_i$ in expectation over $x_i$.
Yet the detector, which knows the PRC secret key, can check if the given text is a noisy codeword to detect the watermark.
Furthermore, one can use our scheme to embed an arbitrary message $\m$ in the response, by choosing $x$ to be $\Encode(\pk, \m)$ under a multi-bit PRC.
This yields a language model steganography scheme, since steganographic secrecy is implied by undetectability.

Our watermarking scheme can tolerate additional changes to the response depending on the robustness of the PRC used.
In particular, if the PRC is robust to deletions, our scheme evades the \emph{emoji attack} which succeeds against all existing undetectable watermarking schemes.
Here, the attacker asks the model to answer the prompt and randomly insert an emoji in between words, then deletes the emojis. This attack succeeds because the detectors in existing schemes must be given \emph{contiguous} text from the watermarked model. However, modeling this attack as a random deletion channel, our codes from \Cref{subsec:prcs-for-deletions} give rise to watermarking schemes that are robust to this attack.

\begin{construction}[Watermark from PRCs]
    Let $\PRC$ be a PRC of block length $n$ with security parameter $\secpar$. $\calW[\PRC] = (\Setup, \Wat, \Detect)$ is the watermarking scheme whose algorithms are specified in \Cref{fig:w-algos}.
\end{construction}

\begin{remark*}
    In \Cref{alg:watermarking-det}, we specify the detector to do a brute-force search over $O(Ln)$ index pairs to find the start and end locations of the perturbed codeword, since deletions may change its length.
    If the perturbed codeword length is known, e.g., when one wants robustness only to substitutions, one can use a more efficient detector that searches only over $O(L)$ indices to find the start location of the codeword.
\end{remark*}

Recall that we write $\hat{\p}$ for the expectation $\E[\p]$ of a distribution $\p$ over $\{0,1\}$.

\begin{figure}

\begin{minipage}[t]{\textwidth}
\begin{algorithm}[H]
\caption{Watermark setup procedure $\Setup(\secparam)$}
\label{alg:watermarking-setup}
    \KwIn{A security parameter $\secpar$}
    \KwResult{A watermark key $\sk$}
    $\PRC.\sk \from \PRC.\KeyGen(\secparam)$\;
    $a_1, \ldots, a_{\lceil \frac{L^*}{n} \rceil} \gets \{0,1\}^{n}$\;
    \Return $\sk = \left(\PRC.\sk, (a_1, \ldots, a_{\lceil \frac{L^*}{n} \rceil})\right)$\;
\end{algorithm}
\begin{algorithm}[H]
\caption{Watermarked text generator $\Wat_\sk$}
\label{alg:watermarking-gen}
    \KwIn{A prompt ($\prompt$) and a key $\sk$}
    \KwResult{Watermarked text $\t_1, \hdots, \t_L$}
    $(\PRC.\sk, a) \gets \sk$\;
    $x \from \PRC.\Encode(\PRC.\sk) \oplus a_1$\; \label{line:r1}
    $i := 1$\;
    $j := 1$\;
    \While{$\texttt{done} \notin \{\t_1, \ldots, \t_{i-1}\}$}{
        $\p_i := \Model(\prompt, \t_1, \hdots, \t_{i-1})$\; \label{line:p_i}
        $\t_i \from \Ber(\hat{\p}_i - (-1)^{x_j} \cdot \min\{\hat{\p}_i,1-\hat{\p}_i\})$\;
        $i \gets i + 1$\;
        $j \gets j+1$\;
        \If{$j > n$}{
            \tcp{Re-sample from the PRC}
            $x \from \PRC.\Encode(\PRC.\sk) \oplus a_{\lceil \frac{i}{n} \rceil}$\; \label{line:r2}
            $j \gets 1$\;
        }
    }
    \Return $\t_1, \dots, \t_L$\;
\end{algorithm}
\end{minipage}
\hspace{0.03\textwidth}
\begin{minipage}[t]{\textwidth}
\begin{algorithm}[H]
\caption{Watermark detector $\Detect_\sk$}
\label{alg:watermarking-det}
    \KwIn{Text $\t_1, \ldots, \t_L$ and watermark key $\sk$}
    \KwResult{$\true$ or $\false$}
    $(\PRC.\sk, a) \gets \sk$\;
    \For{$i \in [L], j \in [i, \min\{i+n-1, L\}], \ell \in \left[\lceil \frac{L^*}{n} \rceil\right]$}{
        \If{$\PRC.\Decode(\PRC.\sk, (\t_i, \dots, \t_{j}) \oplus a_{\ell}) \neq \bot$}{
            \Return $\true$\;
        }
    }
    \Return $\false$\;
\end{algorithm}
\end{minipage}
\caption{The algorithms $(\Setup, \Wat, \Detect)$ of $\calW[\PRC]$. \label{fig:w-algos}}
\end{figure}

\paragraph{Soundness and undetectability}
We first show that $\calW[\PRC]$ is sound and undetectable, which follows quickly from pseudorandomness and error correction of PRCs.

\begin{lemma}[Soundness] \label{lemma:soundness}
    $\calW[\PRC]$ is sound.
\end{lemma}
\begin{proof}
    Let $\t \in \calT^*$ be any token sequence of length $(\len{\t}) \leq \poly$.
    $\Detect_\sk(\t)$ iterates over all $i \in [\len{\t}], j \in [i, \min\{i+n-1, \len{\t}\}]$, checks if $\PRC.\Decode(\PRC.\sk, (\t_i, \ldots, \t_{j}) \oplus a_\ell) = 1$ for any $i,j,\ell$, and outputs true if so.
    Recall that the error correction of $\PRC$ ensures that for any $c \in \Sigma^* = \calT^*$,
        \[
            \Pr_{\prcsk \from \PRC.\KeyGen(\secparam)}[\Decode(\prcsk,c) \neq \bot] \leq \negl.
        \]
    Setting $c_i = (\t_i, \ldots, \t_{j}) \oplus a_\ell$ for each $i,j,\ell$, the probability that $\Decode$ returns 1 for any $c_i$ is negligible by a union bound. 
    Consequently, the probability over choice of $\sk$ that $\Detect_\sk$ outputs true is negligible.
\end{proof}

\begin{lemma}[Undetectability] \label{lemma:undetectability}
$\calW[\PRC]$ is undetectable.
\end{lemma}

\begin{proof}
    We prove that if the output of $\Encode$ is replaced with uniform randomness, this watermark is undetectable. It then follows from pseudorandomness of the PRC that the actual watermark is undetectable.

    We begin by showing that if $x_j \gets \Ber(\frac{1}{2})$, then $\Pr[\t_i = 1] = \hat{\p}_i$, where $\t_i \from \Ber(\hat{\p}_i - (-1)^{x_j} \cdot \min\{\hat{\p}_i,1-\hat{\p}_i\})$.
    Suppose first that $\min\{\hat{\p}_i, 1-\hat{\p}_i\} = \hat{\p}_i$. Then 
    \[\hat{\p}_i - (-1)^{x_j} \cdot \min\{\hat{\p}_i,1-\hat{\p}_i\} = 
    \begin{cases}
        0 &\text{ if } x_j = 0\\
        2 \hat{\p}_i &\text{ if } x_j = 1
    \end{cases} \]
    so $\Pr_{x_j}[\t_i = 1] = \frac{1}{2} \cdot 0 + \frac{1}{2} \cdot 2 \hat{\p}_i = \hat{\p}_i$.
    Now, suppose that $\min\{\hat{\p}_i, 1-\hat{\p}_i\} = 1-\hat{\p}_i$ Then 
    \[\hat{\p}_i - (-1)^{x_j} \cdot \min\{\hat{\p}_i,1-\hat{\p}_i\} = 
    \begin{cases}
        2 \hat{\p}_i - 1 &\text{ if } x_j = 0\\
        1 &\text{ if } x_j = 1
    \end{cases} \]
    so $\Pr_{x_j}[\t_i = 1] = \frac{1}{2} \cdot (2 \hat{\p}_i-1) + \frac{1}{2} \cdot 1 = \hat{\p}_i$. 
    Therefore, if the $x_j$'s are independent $\Ber(\frac{1}{2})$, the distribution of every $\t_i$ under the watermark is identical to its distribution under the original model.

    Suppose for the sake of contradiction that $\calW$ is not undetectable; that is, there is a polynomial-time distinguisher $D$ such that for some constant $c > 0$,
    $$\abs{\Pr[D^{\Model,\RModel}(\secparam) \to 1] - \Pr_{\sk \gets \Setup(\secparam)}[D^{\Model,\Wat_\sk}(\secparam) \to 1]} \ge \frac{1}{\secpar^c}.$$
    We'll construct an adversary $\adv$ that uses $D$ to break pseudorandomness of $\PRC$.
    Recall that in the PRC pseudorandomness experiment, $\adv$ has access to $\calO$, which is either an oracle for $\Encode(\PRC.\sk, \cdot)$ or the uniform distribution $\calU$, and its goal is to distinguish between these two cases.
    $\adv$ interacts with $D$ and simulates $D$'s oracle $\calO_2$, which is either $\RModel$ or $\Wat_\sk$.
    When $D$ queries a prompt $\prompt$ to $\calO_2$, $\adv$ returns a response according to \Cref{alg:watermarking-gen}, \emph{with one key modification}: It draws each $x$ by querying its oracle $\calO$.
    If $\calO$ is $\Encode(\PRC.\sk)$, the responses returned by $\adv$ are exactly those of $\Wat_\sk$.
    By our earlier observation that if the $x_j$'s are uniform, the $\t_i$'s distributions are unchanged, if $\calO$ is $\calU$ the responses returned by $\adv$ are exactly those of $\RModel$.
    Therefore, if we set $\adv$ to output 1 if and only if $D$ outputs 1, we have that 
        \[
            \abs{\Pr_{(\PRC.\sk) \from \KeyGen(\secparam)}[\adv^{\Encode(\PRC.\sk, \cdot)}(1^\secpar) = 1] - \Pr_{\calU}[\adv^{\calU}(1^\secpar) = 1]} \geq \frac{1}{\secpar^c},
        \]
    contradicting pseudorandomness of the PRC.
\end{proof}

\paragraph{Robustness of our scheme} 
In the watermarked text generator (\Cref{alg:watermarking-gen}), we embed a PRC output $r$ into the model's response to a prompt $\prompt$ by sampling $\t_i$ from $\Ber(\hat{\p}_i - (-1)^{x_j} \cdot \min \{\hat{\p}_i, 1-\hat{\p}_i\})$.
Observe that if $\hat{\p}_i = \frac{1}{2}$, this is equal to $\Ber(x_j)$, so $\t_i$ will \emph{always} equal $x_j$.
If $\hat{\p}_i \ne \frac{1}{2}$, then $\t_i$ is sampled with bias toward $x_j$, but it sometimes will not equal $x_j$.
One can think of this embedding process as taking as input $x$, applying some random error, and returning the noisy response $x$ as output.

We model this process as an \emph{embedding channel} $\Eemb : \{0,1\}^n \to \{0,1\}^n$, where for $x \from \{0,1\}^n$, $\Eemb(x) \sim \Model(\prompt)[i:i+n-1]$ (i.e., the $n$ bits of the watermarked model's response to $\prompt$ where $x$ is embedded).
For ease of notation, we write $\t \gets \Eemb(x)$, leaving $\prompt$ and the index $i$ implicit. 
We will show that $\Eemb$ introduces a bounded number of errors when the text has non-zero entropy.

The relevant notion of entropy is the empirical entropy (\Cref{def:empentropy}). We will also make use of the \emph{truncated empirical entropy} $\empHtr$, which is defined as follows.
For a model output $\t$ (given some prompt that we leave implicit),
$\empHtr(\t) = \sum_{i \in \len{\t}} \min\{1, \empH^i(\t)\}$.
As with empirical entropy, we define truncated empirical entropy for partial responses:
$\empHtr^{[i,j]}(\t) = \sum_{\ell = i}^j \min\{1, \empH^\ell(\t)\}$, and $\empHtr^{i}(\t) = \min\{1, \empH^i(\t)\}$.

\begin{lemma} \label{lemma:truncated-surprisal}
    For any $c > 0$ and any indices $a,b \in \N$ where $a \leq b$,
    \[
        \Pr_{\t \from \Model}\left[\empH^{[a,b]}(\t) > c \cdot (b-a) \text{ and } \empHtr^{[a,b]}(\t) \le \frac{c^2}{2} (b-a)\right] \leq \negl[b-a].
    \]
\end{lemma}
\begin{proof}
    We will show that with probability $1-\negl[b-a]$ over $\t \from \Model$,
    \[
        \sum_{i=a}^b \left(\frac{\empH^i(\t)}{\empHtr^i(\t)}\right)^2 \le 2(b-a).
    \]
    By the Cauchy-Schwarz inequality, we have
    \[
        \sum_{i=a}^b \empH^i(\t) \le \sqrt{\sum_{i=a}^b \left(\frac{\empH^i(\t)}{\empHtr^i(\t)}\right)^2} \cdot \sqrt{\sum_{i=a}^b \empHtr^i(\t)^2}.
    \]
    Combining these two inequalities, it will follow that with probability $1-\negl[b-a]$ if $\sum_{i=a}^{b} \empH^i(\t) \ge c \cdot (b-a)$ then
    \begin{align*}
        \sum_{i=a}^{b} \left(\empHtr^i(\t)\right)^2 &\geq \left(\sum_{i=a}^{b} \empH^i(\t) \right)^2 \left(\sum_{i=a}^{b} \left(\frac{\empH^i(\t)}{\empHtr^i(\t)}\right)^2 \right)^{-1}\\
        &\geq \frac{c^2 \cdot (b-a)^2}{2(b-a)}\\
        &= \frac{c^2}{2}{(b-a)}.
    \end{align*}
    And since $\empHtr^i(\t) \leq 1$ by definition, $\empHtr^i(\t) \geq \left(\empHtr^i(\t)\right)^2$ for every $i$, and $\sum_{i=a}^{b} \empHtr^i(\t) \geq \frac{c^2}{2}(b-a)$ as desired.

    It remains to show that with overwhelming probability, $\sum_{i=a}^{b} \left(\frac{\empH^i(\t)}{\empHtr^i(\t)}\right)^2 \le 2 (b-a)$.
    For $i \in [a,b]$, let $X_i := (\empH^i(\t)/\empHtr^i(\t))^2$ and $\q_i := \min\{\p_i(0), \p_i(1)\}$. If $\q_i \ge 1/e$, $X_i = 1$; if $\q_i < 1/e$ then
    \[
        X_i = \begin{cases}
            1 & \text{w.p. } 1-\q_i \\
            \ln^2 \q_i & \text{w.p. } \q_i
        \end{cases}
    \]
    and $\E[X_i \mid \t_a, \dots, \t_{i-1}] = 1-\q_i + \q_i \ln^2 \q_i \le 3/2$. For any fixed $\t_a, \dots, \t_{i-1}$,
    \[
        \Pr_{\t_i \from \p_i}[X_i \ge \ln^4 (b-a)] = \q_i \cdot \mathbbm{1}[\q_i \le e^{-\ln^2 (b-a)}] \le e^{-\ln^2 (b-a)}.
        \]
    By a union bound, $\Pr[\exists i \in [a,b] \text{ s.t. } X_i \ge \ln^4 (b-a)] = \negl[b-a]$.

    Let $Z_{a-1} = 0$ and for $i \in [a,b]$ let $Z_i = Z_{i-1} + X_i - (1-\q_i + \q_i \ln^2 \q_i)$. Observe that $Z_a, \dots, Z_{b}$ is a martingale with respect to $\t_a, \dots, \t_{b}$. Since the differences $\abs{Z_i - Z_{i-1}}$ are bounded by $\ln^4 (b-a)$ with probability $1-\negl[b-a]$, Azuma's inequality (\Cref{theorem:azuma}) implies that
    \[
        \Pr[Z_{b} > {(b-a)}/2] \le \exp{\left(\frac{-{(b-a)}}{8\ln^8 (b-a)}\right)} + \negl[b-a] \leq \negl[b-a].
    \]
    Since $Z_{b} = \sum_{i=a}^{b} [X_i - (1-\q_i + \q_i \ln^2 \q_i)] \ge \sum_{i=a}^{b} X_i - \frac{3}{2} (b-a)$, it follows that
    \[
        \Pr\left[\sum_{i=a}^{b} X_i > 2(b-a)\right] \leq \negl[b-a]. \qedhere
    \]
\end{proof}

We show in the following lemma that the embedding channel introduces bounded-weight errors into any contiguous substring of a response with sufficient empirical entropy.
\begin{lemma}[Entropy and $\Eemb$] \label{lemma:encoding-errors}
    Let $\kappa > 0$ be any constant, and let $a,b \in \N$ be indices such that $b = a + \lceil \kappa n \rceil - 1$. For the embedding channel $\Eemb$ of our watermarked model,
    \ifeprint
    \[
        \Pr_{x \from \{0,1\}^{\lceil \kappa n \rceil }} \left[
        \begin{array}{rcl}
        \wt(\t \oplus x) >  \left(\frac{1}{2} - \frac{c^2}{16} \right) \cdot \len{\t} \text{ and } \empH^{[a,b]}(\t') > c \cdot \len{\t} &: &\t' \gets \Eemb(x) \\
        & &\t \gets \t'[a:b] 
        \end{array}
        \right]
        \leq \negl[n].
    \]
    \fi
\end{lemma}
\begin{proof}
    For ease of notation, in this proof we use $\p_i(\hat{\t}_i)$ to denote the probability that the model outputs a bit $\hat{\t}_i \in \{0,1\}$ as the next $(a+i)^{\text{th}}$ token given that it has output tokens $\t'[:a + i - 1]$ so far.
    We also let $\empH(\t)$ denote $\empH^{[a,b]}(\t')$ and $\empHtr(\t)$ denote $\empHtr^{[a,b]}(\t')$.
    
    For each $i \in \len{\t}$, recall that the watermark samples $\t_i \from \Ber(\hat{\p}_i - (-1)^{x_i} \cdot \min \{\hat{\p}_i, 1-\hat{\p}_i\})$, and for $\hat{\t}_i \in \{0,1\}$, $\Pr_{x_i \from \{0,1\}}[\t_i = \hat{\t}_i] = \p_i(\hat{\t}_i)$.

    For fixed $\hat{\t}_i \in \{0,1\}$, we have
    \begin{align*}
         \frac{\Pr[\t_i = \hat{\t}_i \mid x_i=\hat{\t}_i]}{2} &= \Pr[x_i = \hat{\t}_i \text{ and } \t_i = \hat{\t}_i] \\
        &= \Pr[\t_i = \hat{\t}_i] \cdot \Pr[x_i = \hat{\t}_i \mid \t_i = \hat{\t}_i] \\
        &= \p_i(\hat{\t}_i) \cdot \Pr[x_i = \hat{\t}_i \mid \t_i = \hat{\t}_i].
    \end{align*}
    Because $\t_i$ is distributed as $\Ber(\hat{\p}_i - (-1)^{x_i} \cdot \min \{\hat{\p}_i, 1-\hat{\p}_i\}) = \Ber(\hat{\p}_i - (-1)^{x_i} \cdot \min \{\hat{\p}_i, 1-\hat{\p}_i\})$, we also have
    \[
        \Pr[\t_i = \hat{\t}_i \mid x_i=\hat{\t}_i] = \p_i(\hat{\t}_i) + \min\{\p_i(\hat{\t}_i), 1-\p_i(\hat{\t}_i)\}.
    \]
    Therefore,
    \begin{align*}
        \Pr[x_i = \hat{\t}_i \mid \t_i = \hat{\t}_i] &= \frac{1}{2} + \frac{1}{2} \min\left\{1, \frac{1}{\p_i(\hat{\t}_i)} - 1\right\} \\
        &\ge \frac{1}{2} + \frac{1}{2} \min\left\{1, - \ln \p_i(\hat{\t}_i)\right\}
    \end{align*}
    where for the inequality we have used the fact that $\ln z \ge 1-1/z$ for $z>0$. 
 
    Let $Y_1, \ldots, Y_{\len{\t}}$ be Bernoulli random variables where $Y_i = 1$ if and only if $x_i = \hat{\t}_i$. The expected number of correct bits given that $\t = \hat{\t}$ is 
    \begin{align*}
        \E\left[\sum_{i=1}^{\len{\t}} Y_i \ \middle\lvert \ \t = \hat{\t} \right] &= \sum_{i=1}^{\len{\t}} \Pr[Y_i \mid \t_i = \hat{\t}_i]\\
        &\geq \frac{{\len{\t}}}{2} + \frac{1}{2} \sum_{i=1}^{\len{\t}} \min \{1, -\ln \p_i(\hat{\t}_i)\}\\
        &\ge \frac{{\len{\t}}}{2} + \frac{\ln 2}{2} \sum_{i=1}^{\len{\t}} \min \{1, -\log \p_i(\hat{\t}_i)\}\\
        &\ge \frac{{\len{\t}}}{2} + \frac{1}{4} \empHtr(\t).
    \end{align*}
    By \Cref{lemma:truncated-surprisal}, if $\empH(\t) > c \cdot {\len{\t}}$, this is at least $\left(\frac{1}{2} + \frac{c^2}{8}\right) {\len{\t}}$. Since the events $x_i = \hat{\t}_i$ are independent for fixed $\hat{\t}_i$, we can apply a Chernoff bound to see that $\wt(x \oplus \t) \le \left( \frac{1}{2} - \frac{c^2}{16}\right) \cdot \len{\t}$ with probability $1-\negl[n]$.
\end{proof}

In the following proofs of substring-completeness and substring-robustness, it will be useful to consider the model's response as consisting of \emph{blocks}.
A block is a contiguous substring of the response where a PRC codeword was embedded.
More formally, recall that the generator $\Wat_\sk$ first chooses a PRC codeword of length $n$ that corresponds to the first $n$ tokens of the response, $(\t_1, \ldots, \t_n)$, and applies a one-time pad $a_1$.
It then chooses a new PRC codeword for the next $n$ tokens of the response and applies a fresh one-time pad,
continuing to do so until it terminates, having output the final response $\t_1, \ldots, \t_{cn + i}$ for some $c, i \in \Z_{\geq 0}$.
This response consists of $(c-1)$ \emph{full} blocks $(\t_1, \ldots, \t_n), \ldots, (\t_{(c-1)n + 1}, \ldots, \t_{cn})$.
It also consists of a possibly \emph{incomplete} block, $(\t_{cn + 1}, \ldots, \t_{cn + i})$.
The substring $\t_{n-1}, \ldots, \t_{2n+2}$ of this response consists of two incomplete blocks: $(\t_{n-1}, \t_n)$ and $(\t_{2n+1}, \t_{2n+2})$, and one full block: $(\t_{n+1}, \ldots, \t_{2n})$.

\begin{lemma}[Substring completeness] \label{lemma:substring-completeness}
    Let $\varepsilon > 0$ be any constant. If $\PRC$ is a zero-bit PRC with block length $n$ and robustness to every $(1/2-\varepsilon)$-bounded channel, then $\calW[\PRC]$ is $(4\sqrt{\varepsilon} \cdot L + 2\sqrt{2} \cdot n)$-substring-complete.
\end{lemma}
\begin{proof}
    Let $\t \gets \Wat_\sk(\prompt)$ be a watermarked response for some arbitrary prompt $\prompt$, and let $\t'$ be a length-$L$ contiguous substring of $\t$, where $\empH(\t') \geq 4\sqrt{\varepsilon} \cdot L + 2\sqrt{2} \cdot n$.
    
    The start and end of $\t'$ may be part of incomplete blocks, with some number of full blocks in between.
    Recall that the truncated empirical entropy of any single-bit distribution is at most 1, so each of these possibly-incomplete blocks contains at most $n$ truncated empirical entropy.
    Consider padding these possibly-incomplete blocks, appending deterministic bits so they are each of length $n$ but have the same truncated empirical entropy.
    It now follows from \Cref{lemma:truncated-surprisal} that with overwhelming probability in the length of these blocks (which is now $n$), they each have at most $n\sqrt{2}$ empirical entropy.
    The full blocks therefore contain at least $4\sqrt{\varepsilon} \cdot L$ empirical entropy.

    By an averaging argument, $\t'$ must contain some full block $\tilde{\t} \subseteq \t'$ with at least $4\sqrt{\varepsilon} \cdot n$ empirical entropy.
    By \Cref{lemma:encoding-errors}, the embedding channel applied to this block is $(1/2-\varepsilon)$-bounded. 
    Furthermore, because of the one-time pads, these channels' errors do not depend on other codewords.
    Since the PRC used is robust against such bounded channels, the detector will see that $\PRC.\Decode(\PRC.\sk, \tilde{\t} \oplus \alpha_\ell) \neq \bot$ for some $\ell$ and output $\true$.
\end{proof}

\begin{lemma}[Substring-robustness against substitutions] \label{lemma:watermark-robustness}
    Let $\varepsilon, \delta > 0$ be any constants. If $\PRC$ is a zero-bit PRC with block length $n$ and robustness to every $(1/2 - \varepsilon \cdot \delta)$-bounded channel, then $\calW[\PRC]$ is $(4\sqrt{\varepsilon} \cdot L + 2\sqrt{2} \cdot n)$-substring-robust against $\BSC_{1/2-\delta}$.
\end{lemma}
\begin{proof}
    As in the proof of \Cref{lemma:substring-completeness}, the decoder receives the output of the composition of the error channel and embedding channel on the sufficiently high-entropy block. The embedding channel is again $(1/2-\varepsilon)$-bounded on this block by \Cref{lemma:encoding-errors}.
    Therefore, the composition of the error channel $\BSC_{1/2-\delta}$ with the embedding channel has expected error rate
    \[
        (1/2-\delta) \cdot (1/2+\varepsilon) + (1/2-\varepsilon) \cdot (1/2+\delta) = \frac{1}{2} - 2 \varepsilon \cdot \delta
    \]
    and is in particular $(1/2-\varepsilon \cdot \delta)$-bounded. Since $\PRC$ is robust to such channels, it follows that $\calW[\PRC]$ is $(4\sqrt{\varepsilon} \cdot L + 2\sqrt{2} \cdot n)$-substring-robust against $\BSC_{1/2-\delta}$.
\end{proof}

We showed above that when a substring $\t'$ of a response has at least $4\sqrt{\varepsilon} \cdot \t'$ empirical entropy, the embedding channel $\Eemb$ applied to some block in that substring is $\left(\frac{1}{2} - \varepsilon \right)$-bounded.
For the following theorem, we need to assume something further about $\Eemb$: restricted to that block in the substring, $\Eemb$ is equal to $\BSC_{\alpha(\varepsilon)}$ for some $\alpha(\varepsilon) \in \left(0, \frac{1}{2} - \varepsilon\right]$.

\begin{assumption} \label{assumption:embedding}
    There is a function $\alpha : [0,1/2) \to [0,1/2)$ such that for any constant $\varepsilon \in [0,1/2)$, if a response substring of length $L$ has at least $(4\sqrt{\varepsilon} \cdot L + 2\sqrt{2} \cdot n)$ empirical entropy, the embedding channel can be modeled as $\BSC_{\alpha(\varepsilon)}$.
\end{assumption}

As noted in the remark at the end of \Cref{subsec:water-prelims}, this assumption is not realistic if we use the standard transformation from tokens to bits because the bit-errors within the representation of a given token may be correlated.
However, an alternative transformation that avoids this issue simply uses one bit for each token.
In this case, \Cref{assumption:embedding} roughly states that the entropy is spread throughout the substring, and that the substring does not repeat too many tokens.

\begin{lemma}[Substring-robustness against deletions] \label{lemma:watermark-robustness-deletions}
    Suppose that \Cref{assumption:embedding} holds with function $\alpha$. Let $\varepsilon > 0$, $q \in (0,1/2)$, and $p \in (0,1)$ be any constants. If $\PRCdel$ is a zero-bit PRC with block length $n$ and robustness to $\BDC_p \circ \BSC_q \circ \BSC_{\alpha(\varepsilon)} = \BDC_p \circ \BSC_{q + \alpha(\varepsilon) - q \cdot \alpha(\varepsilon)}$, then $\calW[\PRCdel]$ is $(4\sqrt{\varepsilon} \cdot L + 2\sqrt{2} \cdot n)$-substring-robust against $\BDC_p \circ \BSC_q$.
\end{lemma}

\begin{proof}
    Consider the block of the response over which $\Eemb$ is equivalent to $\BSC_{\alpha(\varepsilon)}$.
    The decoder receives the output of the composition of the error channel $\BDC_p \circ \BSC_q$ and embedding channel $\BSC_{\alpha(\varepsilon)}$ applied to the codeword used in this block.
    The composition of these channels is $\BDC_p \circ \BSC_q \circ \BSC_{\alpha(\varepsilon)}$, which $\PRCdel$ is robust to. 

    Therefore, the detector of $\calW[\PRCdel]$ will output $\true$ when it runs $\Decode$ on this block of the response.
\end{proof}

By applying the PRCs of \Cref{theorem:ldpc-prc-xor,theorem:ldpc-prc-lpn} to \Cref{lemma:watermark-robustness,lemma:watermark-robustness-deletions}, we obtain the following results.

\begin{theorem} \label{theorem:robust-water}
    Let $\delta > 0$ be any constant and $n$ be a security parameter. Under \Cref{assumption:combined}, there exists a language model watermarking scheme that is $O(L+n)$-substring-robust against $\BSC_{1/2-\delta}$.
\end{theorem}

\begin{theorem} \label{theorem:deletion-robust-water}
    Suppose that \Cref{assumption:embedding} holds with function $\alpha$, let $q \in (0,1/2)$ and $p \in (0,1)$ be any constants, and let $n$ be a security parameter. Under \Cref{assumption:combined}, there exists a language model watermarking scheme that is $O(L+n)$-substring-robust against $\BDC_p \circ \BSC_q$.
\end{theorem}

We remark that although we state our theorems for error rates in $(0,1/2)$ for convenience, one can easily extend these schemes to error rates for arbitrary constants in $(0,1) \setminus \{1/2\}$.

\subsection{Language model steganography} \label{subsec:langauge-model-stego}
We have proven completeness, robustness, soundness, and undetectability of $\calW$, where the random strings $x$ used in $\Wat_\sk$ were computed as $\PRC.\Encode(\PRC.\sk)$.
These proofs relied only on our ability to distinguish PRC codewords from unrelated strings, so it sufficed to use a zero-bit PRC.
Of course, one could obtain a multi-bit watermarking scheme by replacing $\PRC.\Encode(\PRC.\sk)$ with $\PRC'.\Encode(\PRC'.\sk, \m)$ for any message $\m$, using $\PRC'$ which is multi-bit and has the same robustness as $\PRC$. 
\Cref{lemma:watermark-robustness,lemma:watermark-robustness-deletions} both apply identically in this case.
This yields a robust language model steganography scheme, where steganographic secrecy is implied by undetectability of $\calW[\PRC']$.
This steganography scheme has the same robustness as the watermarking scheme.

Applying this observation and using the constant-rate PRCs of \Cref{subsec:constant-rate-prcs}, we obtain the following result.

\begin{theorem} \label{theorem:language-model-stego}
    Let $\delta > 0$ be any constant and $n$ be a security parameter. Under \Cref{assumption:combined}, there exists a language model steganography scheme that is $O(L+n)$-substring-robust against $\BSC_{1/2-\delta}$.
\end{theorem}

In \Cref{subsec:public-attr}, we will see how this observation about encoding arbitrary messages can be used to build watermarks with public attribution.
Since the proofs of security and robustness of the language model steganography scheme are completely identical to those for the watermarking scheme $\calW[\PRC]$, we do not formally prove them separately.
However, in \Cref{sec:stego} we show that PRCs can be used for \emph{universal} steganography (which is more general than language model steganography).

\subsection{Watermarks with public attribution} \label{subsec:public-attr}
See \Cref{subsec:techo-attribution} for a description of publicly attributable watermarks.

We define public attribution in terms of a function called $\ForgeDetect$.
Recall that we intentionally design this function to \emph{not} be robust.
$\ForgeDetect$, given a text $\t$ and a public detection key, outputs a pair $(\t', \b)$, where $\b \in \{\true, \false\}$ indicates whether $\t$ contains verbatim a significant prefix of text output by the model, and if so, $\t'$ is that prefix.
Although $\ForgeDetect$ is intentionally not robust, our scheme's $\Detect$ function retains all properties (undetectability, robustness, soundness) of our standard watermarks.

\Cref{alg:attr-watermarking-setup,alg:attr-watermarking-gen,alg:attr-watermarking-attr,alg:watermarking-det}, given in \Cref{fig:watt-algos}, comprise a watermarking scheme with unforgeable public attribution from any secret-key PRC and any digital signature scheme.

\begin{figure}

\begin{algorithm}[H]
\caption{Publicly attributable watermark setup procedure $\Setup(\secparam)$}
\label{alg:attr-watermarking-setup}
    \KwIn{A security parameter $\secpar$.}
    \KwResult{A watermark public key $\pk$ and secret key $\sk$}
    $\prcsk \from \PRC.\KeyGen(\secparam)$\;
    $\sigpk, \sigsk \from \Sig.\KeyGen(\secparam)$\;
    $a_1, \ldots, a_{\lceil \frac{L^*}{n} \rceil} \gets \{0,1\}^{n}$\;
    $\pk \gets \left(\prcsk, \sigpk, (a_1, \ldots, a_{\lceil \frac{L^*}{n} \rceil})\right)$\;
    $\sk \gets \sigsk$\;
    \Return $\pk, \sk$\;
\end{algorithm}
\vspace{0.1cm}

\begin{algorithm}[H]
\caption{Publicly attributable watermarked text generator $\Wat_\sk$}
\label{alg:attr-watermarking-gen}
    \KwIn{A prompt ($\prompt$), a watermark public key $\pk = \left(\prcsk, \sigpk, (a_1, \ldots, a_{\lceil \frac{L^*}{n} \rceil})\right)$, and a watermark secret key $\sk = \sigsk$.}
    \KwResult{Watermarked text $\t_1, \hdots, \t_L$}
    $(\PRC.\sk, \sigpk, a) \gets \pk$\;
    $x \from \PRC.\Encode(\prcsk) \oplus a_1$\;
    $i := 1, j := 1$\;
    \While{$\texttt{done} \notin \{\t_1, \ldots, \t_{i-1}\}$}{
        \tcp{Embed the current message if first block, or signature otherwise}
        $\p_i := \Model(\prompt, \t_1, \hdots, \t_{i-1})$\;
        $\t_i \from \Ber(\hat{\p}_i - (-1)^{x_j} \cdot \min\{\hat{\p}_i,1-\hat{\p}_i\})$\;
        $i \gets i + 1$, $j \gets j+1$\;       
        \tcp{If done embedding the current codeword, generate a new signature}
        \If{$j > n$}{
            $\m \gets \t_1, \ldots, \t_i$\;
            $\sigma \gets \Sign_\sigsk(\m)$\;
            $x \gets \PRC.\Encode(\prcsk,\sigma) \oplus a_{\lceil \frac{i}{n} \rceil}$\;
            $j \gets 1$\;
        }
    }
    \Return $\t_1, \ldots, \t_L$\;
\end{algorithm}
\vspace{0.1cm}

\begin{algorithm}[H]
\caption{Publicly attributable watermarked text attributor $\ForgeDetect_\pk$}
\label{alg:attr-watermarking-attr}
    \KwIn{Text $\t_1, \ldots, \t_L$ and a watermark public key $\pk = \left(\prcsk, \sigpk, (a_1, \ldots, a_{\lceil \frac{L^*}{n} \rceil})\right)$.}
    \KwResult{$\true$ or $\false$}
    $(\PRC.\sk, \sigpk, a) \gets \pk$\;
    \For{$i \in [L-n+1], \ell \in \left[\lceil \frac{L^*}{n} \rceil \right]$}{
        $\m \gets \t_1, \ldots, \t_{i-1}$\;
        $\sigma \gets \PRC.\Decode(\prcsk, (\t_{i}, \ldots, \t_{i+n-1}) \oplus a_\ell)$\;
        \If{$\Vrfy_{\sigpk}(\m, \sigma)$}{
            \Return $(\m, \true)$\;
        }
    }
    \Return $(\bot, \false)$\;
\end{algorithm}

\caption{Algorithms $(\Setup, \Wat, \ForgeDetect)$ of $\Watt[\PRC]$. The detector $\Detect$ is the same as that of $\calW[\PRC]$, given in \Cref{alg:watermarking-det}. Here, we let $n$ be a codeword length sufficient to encode signatures from $\Sig$ using $\PRC$. \label{fig:watt-algos}}
\end{figure}

The watermarked text generator of $\Watt$ (\Cref{alg:attr-watermarking-gen}) is exactly the same as that of $\calW$ (\Cref{alg:watermarking-gen}), but $x$ is now generated as the output of $\PRC.\Encode$ on specific messages.
In particular, the first $x$ is sampled as $\PRC.\Encode(\sk)$, and every subsequent $x$ is the encoding of a signature on the response output thus far.

\begin{figure}[t]
    \centering
    \hfpages{
        0.6
    }{
        \underline{$\AttrForge_{\adv, \calW_{\sf att}}(\secpar)$}\\
        $\pk, \sk \gets \KeyGen(\secparam)$\\
        $\t^* \gets \adv^{\Model, \Wat_\sk}(\secparam, \pk)$\\
        Let $\calQ$ denote the set of responses returned to $\adv$ by $\Wat_\sk$\\
        $(\m, \b) \gets \ForgeDetect(\pk, \t^*)$\\
        \texttt{return} $(\b = \true) \land (\forall \t \in \calQ, \m$ is not a contiguous substring of $\t$)
    }
    \caption{The attribution forgery experiment $\AttrForge_{\adv, \calW_{\sf att}}(\secpar)$\label{fig:attrforge}}
\end{figure}

\begin{definition}[Unforgeable public attribution]
    A watermarking scheme $\calW_{\sf att}$ for a model $\Model$ has $b(L)$ \emph{unforgeable public attribution} if there is a function $\ForgeDetect_\pk(\t)$ satisfying:
    \begin{itemize}
        \item (Syntax): $\ForgeDetect_\pk(\t) \to (\t', {\sf b})$ takes in a token sequence $\t$ and outputs a (token sequence, boolean) pair. $\sf b = \true$ indicates that a prefix of $\t$ was output verbatim by the model; the corresponding $\t'$ that it outputs is this prefix. $\sf b = \false$ indicates that no sufficiently long prefix of $\t$ was output verbatim by the model; in this case $\t'= \bot$.
        \item ($b(L)$ public attribution): 
        For every security parameter $\secpar$ and prompt~$\prompt$ of length $\poly$,
        \ifeprint
        \begin{align*}
            \Pr_{\substack{\pk, \sk \gets \Setup(\secparam) \\ \t \gets \Wat_\sk(\prompt)}}
            \Big[\;
            & \exists\; i,L \in [\len{\t}] \text{ such that }
            \len \t' < i-1 \text{ and } \\
            &  \empH^{[i:i+L]}\left(\Model,\prompt,\t\right) \geq b\left(L\right) \mid (\t', \b) \from \ForgeDetect_\pk(\t)] \le \negl.
        \end{align*}
        \fi
        \item (Unforgeability): For all polynomial-time adversaries $\adv$,
            $$\Pr[\AttrForge_{\adv, \calW_{\sf att}}(\secpar) = 1] \leq \negl,$$
            where the $\AttrForge$ experiment is given in \Cref{fig:attrforge}. 
    \end{itemize}
\end{definition}

\begin{construction}[Watermark with public attribution]
    Let $\Sig$ be a digital signature scheme with signatures of length $n$, and let $\PRC$ be a PRC of block length $n$.
    $\Watt[\PRC]$ is the watermarking scheme whose algorithms are specified in \Cref{fig:watt-algos}. Its detector is the same as that of $\calW[\PRC]$.
\end{construction}

\begin{lemma}[Unforgeability] \label{lemma:unforgeability}
    $\Watt[\PRC]$ is unforgeable if the underlying signature scheme is unforgeable.
\end{lemma}
\begin{proof}    
    Suppose there exists an adversary $\adv$ that wins $\AttrForge_{\adv, \calW_{\sf att}}(\secpar)$ with non-negligible probability. 
    We construct $\bdv$ that uses $\adv$ to break unforgeability of the underlying signature scheme $\Sig$. 
    $\bdv$ acts as the challenger in $\AttrForge_{\adv, \calW_{\sf att}}(\secpar)$. It receives the signature public key $\Sig.\pk$ used in $\calW_{\sf att}$ from its own challenger for $\SigForge_{\bdv,\Sig}$. $\bdv$ generates the other parameters for $\calW_{\sf att}$ (that is, the PRC parameters) itself. 
    Let $\pk$ denote the resulting public key of the watermarking scheme; $\bdv$ passes $\adv$ $\pk$ as input.
    When responding to $\adv$'s queries to $\Wat_\sk$, $\bdv$ generates the necessary signature using its signing oracle in $\SigForge_{\bdv,\Sig}$.
    Notice that every query that $\bdv$ makes to its signing oracle is a contiguous substring of some response $\t \in \calQ$, where $\calQ$ is the set of responses returned to $\adv$ by $\Wat_\sk$. 
    At the end of $\AttrForge_{\adv, \calW_{\sf att}}(\secpar)$, $\adv$ outputs $\t^*$. $\bdv$ computes $(\m, \b) \gets \ForgeDetect(\pk, \t^*)$; if $\b = \true$, $\bdv$ outputs the corresponding $\m$ and verifying signature $\sigma$.
    
    Recall that if $\adv$ wins, it outputs $\t^*$ such that $(\m,\b) \gets \ForgeDetect(\pk, \t^*)$ where $\b = \true$ and $\m$ is not a contiguous substring of any response in $\calQ$. Since every query made by $\bdv$ to its signing oracle was a substring of some response in $\calQ$, $\bdv$ did not query this $\m$. 
    However, $\bdv$ is able to output a verifying signature for $\m$, so $\bdv$ wins $\SigForge_{\bdv,\Sig}$.
\end{proof}

\begin{lemma}[Public attribution] \label{lemma:public-attribution}
    Let $\varepsilon > 0$ be any constant.
    If $\PRC$ has block length $n$ and is robust to $\left(\frac{1}{2}-\varepsilon\right)$-bounded channels, then $\Watt[\PRC]$ satisfies $(4\sqrt{\varepsilon} \cdot L + 2\sqrt{2} \cdot n)$ public attribution.
\end{lemma}
\begin{proof}
    Let $\t \gets \Wat_\sk(\prompt)$ be a response, and let $i$ be such that $H_e^{[i:i+L]}(\t) \ge 4\sqrt{\varepsilon} \cdot L + 2\sqrt{2} \cdot n$.

    By the exact same proof as in \Cref{lemma:substring-completeness}, the robustness of $\PRC$ implies that there is some block $\t'$ with starting index at least $i$, which is correctly decoded by $\PRC.\Decode$ in \Cref{alg:attr-watermarking-attr}.
    That is, $\PRC.\Decode$ yields a signature computed on the substring of the response up to the start of block $\t'$.
    Since $\t'$ starts at index at least $i$, this signed portion has length at least $i-1$.       
\end{proof}

\begin{theorem}[A watermark with unforgeable public attribution] \label{theorem:watermark-att}
        Let $\varepsilon > 0$ be any constant.
        If the underlying signature scheme is unforgeable, and $\PRC$ is a PRC with block length $n$ and robustness to every $\left(\frac{1}{2}-\varepsilon\right)$-bounded channel, $\Watt[\PRC]$ is unforgeable and $(4\sqrt{\varepsilon} \cdot L + 2\sqrt{2} \cdot n)$-publicly-attributable.
        
        Furthermore, $\Watt[\PRC]$ retains the same soundness, undetectability, substring-completeness, and substring-robustness properties as $\calW[\PRC]$. That is, \Cref{lemma:soundness,lemma:undetectability,lemma:substring-completeness,lemma:watermark-robustness,lemma:watermark-robustness-deletions} all apply to $\Watt[\PRC]$.
\end{theorem}

\begin{proof}
    Unforgeabile public attribution follows from  \Cref{lemma:unforgeability,lemma:public-attribution}.

    As noted earlier, the proofs of \Cref{lemma:soundness,lemma:undetectability,lemma:substring-completeness,lemma:watermark-robustness,lemma:watermark-robustness-deletions} all hold for any choice of message input to $\PRC.\Encode$. 
    The only difference between $\calW[\PRC]$ and $\Watt[\PRC]$ is the choice of message input to $\PRC.\Encode$.
    Therefore, those proofs hold for $\Watt[\PRC]$.
\end{proof}

\paragraph{Related public watermarking work.}
Concurrent work \cite{public} constructs a watermark that similarly can be detected with a public key but requires a secret key for generation. 
Their scheme, like $\Watt$, embeds a digital signature in the response and includes the signature verification key in the public key. However, their scheme has two key differences from $\Watt$:
\begin{itemize}
    \item They assume that each token sequence of a fixed length $\ell$ has at least $\alpha$ min-entropy.
    \item Their scheme has a single detector, analogous to $\ForgeDetect$, which is robust only to cropping. In addition to $\ForgeDetect$, our scheme also has $\Detect$, which is robust to deletions and/or the binary symmetric channel, depending on the underlying PRC.
\end{itemize}

\cite{public} also discuss using error-correcting codes to improve robustness. 
They suggest applying an error-correcting code to the message and signature before embedding them. However, this would make the scheme no longer undetectable or distortion-free.

\newpage

\section{Application: universal steganography} \label{sec:stego}
Our main result in this section is the construction of a robust stateless steganography scheme using PRCs.
We first prove that robustness follows immediately from applying a PRC to a non-robust steganography scheme appearing in several prior works \cite{And98,HLvA02,vAH04}.
We then show that the assumptions necessary for this scheme can be weakened, which sacrifices some robustness but still yields the most robust scheme in its regime.

\subsection{Steganography preliminaries}
\paragraph{Setting (from \cite{HLvA02}).}
A \emph{steganographic channel} (or \emph{covertext distribution}) $\calC$ is a distribution on sequences of symbols from some alphabet $\Sigma$. These symbols $d \in \Sigma$ are often called \emph{documents}, and sequences of symbols are often called \emph{covertexts}. 
It is often convenient to consider the conditional distribution over the next symbol, given some subsequence already output by $\calC$.
More precisely, given a \emph{history} $h \in \Sigma^*$, we use $\calC_h$ to denote the conditional distribution of the next symbol from $\calC$ given that thus far $\calC$ has output $h$.
For the steganographic encoder, we assume existence of an efficient oracle $M(\cdot)$ that can sample $\calC_h$ for any $h \in \Sigma^*$ of the encoder's choice.
Our steganographic encoder makes use of a rejection sampling function denoted $RS^{M,f} : \{0,1\} \times \mathbb{N} \to \Sigma$, which on input $(x, \kappa)$, samples $c \gets M$ until $f(c) = x$, taking at most $\kappa$ samples. 
If the sampler reaches the $\kappa^{\text{th}}$ sample $c_{\kappa}$, it outputs $c_{\kappa}$ regardless of whether $f(c_{\kappa}) = x$.
We say a function $f : \Sigma \to R$ is an \emph{unbiased function on a steganographic channel $\calC$} if for all $r \in R$ and all histories $h \in \Sigma^*$,
$\Pr_{d \gets \calC_h}[f(d) = r] = \frac{1}{\abs{R}}$.

We recall the definition of a symmetric steganography scheme as presented in \cite{meteor}, which is equivalent to that of \cite{HLvA02}.

\begin{definition}[Symmetric steganography scheme \cite{meteor}]
    A symmetric steganography scheme $\Sigma_{\calC}$ is a triple of possibly probabilistic algorithms $\Pi_{\calC} = ({\sf KeyGen}_{\calC}, \Encode_{\calC}, \Decode_{\calC})$ parameterized by a covertext channel distribution $\calC$. 
    \begin{itemize}
        \item ${\sf KeyGen}_{\calC} (1^\secpar)$ generates the key $\sk$.
        \item $\Encode_{\calC} (\sk, \m, h)$ is a (possibly probabilistic) algorithm that takes the key $\sk$ and a plaintext message $\m$, called a \emph{hiddentext}.
        Additionally, the algorithm can optionally take in a message history $h = \{h_0,h_1, \ldots, h_{\abs{h}-1}\} \in \Sigma^*$ of covertext messages. It returns a symbol sequence $c_i \in \Sigma^*$, called a \emph{stegotext}.
        \item $\Decode_{\calC}(\sk, c, h)$ is a (possibly probabilistic) algorithm that takes as input the key $\sk$, a stegotext $c$, and optionally a history $h \in \Sigma^*$.
        It returns a plaintext message $\m$ on success or the empty string $\epsilon$ on failure.
    \end{itemize}
\end{definition}
The subscript $\calC$ indicates that these algorithms have access to $\calC$ via the oracle $M(\cdot)$; we often omit this subscript for convenience.

A scheme in which the decoding algorithm requires the history is \emph{stateful}; a scheme where decoding is possible without knowledge of the history is \emph{stateless}.
This shared history in stateful schemes is very powerful, allowing the sender and receiver to maintain a shared counter of messages sent thus far, which they can use to generate a shared one-time pad for each message.

In a \emph{public-key steganography} scheme, $\KeyGen$ outputs a public-secret key pair. Analogously to public-key encryption, the encoder takes as input the public key (and not the secret key), and the decoder still takes as input the secret key. We present formal definitions for secret-key steganography here, and refer the reader to \cite{vAH04} for formal definitions of public-key steganography. 

A steganography scheme must satisfy \emph{correctness} and \emph{security}.

\begin{definition}[Steganographic correctness]
A steganography scheme $\Pi_{\calC} = ({\sf KeyGen}_{\calC}, \Encode_{\calC}, \Decode_{\calC})$ is \emph{correct} if for any history $h \in \Sigma^*$ and any message $\m$,
$$\Pr_{\sk \gets \KeyGen(1^\secpar)} \left[\Decode_{\calC}(\sk, \Encode_{\calC}(\sk,\m,h),h) = \m \right] \geq 1 - \negl.$$
\end{definition}
We note that some other works require this probability to be 1.

Security requires that an adversary cannot distinguish between oracle access to $\Encode_{\calC}(\sk, \cdot, \cdot)$ or a sampling oracle $O_{\calC}(\cdot, \cdot)$ for the channel distribution.
\begin{definition}[Steganographic security against chosen hiddentext attacks]
    $\Pi_{\calC}$ is secure against \emph{chosen hiddentext attacks} if for all polynomial-time adversaries $\calA$, for all $\sk \gets {\sf KeyGen}_{\calC}(1^\secpar)$,
    $$\abs{\Pr\left[\adv^{\Encode_{\calC}(\sk, \cdot, \cdot), M(\cdot)} = 1\right] - \Pr\left[\adv^{O_{\calC}(\cdot,\cdot), M(\cdot)} = 1 \right]} \leq \negl.$$
\end{definition}

This definition can be modified for public-key steganography by giving the adversary the public key as input.

\subsection{Robust stateless steganography}
We first present $\sf Steg$, a steganography scheme that appeared originally in \cite{And98}.
$\sf Steg$ is parameterized by an encryption scheme $(\KeyGen,\Enc,\Dec)$ and a function $f$.

\begin{construction}[{${\sf Steg}[(\KeyGen,\Enc,\Dec), f, \kappa]$}]
    Let $\sf Steg.\KeyGen = \KeyGen$.
    Let $\sf Steg.\Encode$ and $\sf Steg.\Decode$ be defined as in \Cref{fig:steg}. 
\end{construction}

\cite{HLvA02} and \cite{vAH04} prove security of the secret- and public-key versions of $\sf Steg$, respectively.

\begin{numberedclaim}[\cite{HLvA02,vAH04}] \label{claim:stego-orig}
Let $(\KeyGen, \Enc, \Dec)$ be a (public-key) encryption scheme with ciphertexts indistinguishable from random bits under a chosen plaintext attack, and let $f: \Sigma \to \{0,1\}$ be a function which is unbiased on $\calC$. Then {${\sf Steg}[(\KeyGen,\Enc,\Dec), f,\secpar]$} is a secure (public-key) steganography scheme.
\end{numberedclaim}

Our first application of PRCs to steganography will be in showing that $\sf Steg$ is robust to errors when a PRC is used as the encryption scheme. To our knowledge, this gives the first robust stateless steganography scheme.

 Here, we define robustness of a steganography scheme against a channel $\calE$ analogously to robustness of a PRC. 
In particular, $\calE$ may be a deletion channel, rather than just making substitutions.
Robustness requires that given the output of the error channel applied to a stegotext, the hiddentext can be recovered with overwhelming probability.

\begin{definition}[Robustness to an error channel]
    Let $\calC$ be a steganographic channel with symbol alphabet $\Sigma$, and let $\calE : \Sigma^* \to \Sigma^*$ be an error channel.
    A steganography scheme $(\KeyGen,\Encode, \Decode)$ for $\calC$ is robust to $\calE$ if for all messages $\m$ and all histories $h$,
\[
\Pr_{\sk, \pk \gets \KeyGen(1^\secpar)} \left[\Decode(\sk,\calE(x),h) = \m : x \gets \Encode(\pk,\m,h) \right] = 1 - \negl.
\]
\end{definition}
This definition can be adapted for secret-key schemes by replacing $\pk$ with $\sk$.

We say that the \emph{rate} of a steganography scheme is the ratio of the number of message bits to the number of symbols of stegotext that are needed to decode the message.

\begin{theorem}[Robust stateless steganography] \label{theorem:stego}
Let $\PRC = (\KeyGen, \Encode, \Decode)$ be a (public-key) PRC that is robust to some binary channel $\calE$. Let $f: \Sigma \to \{0,1\}$ be a public function which is unbiased on $\calC$.
Then $\sf Steg[\PRC,f,\secpar]$ is a (public-key) steganography scheme, with the same rate as $\PRC$, that is robust to any $\calE'$ such that $f \circ \calE' = \calE \circ f$.
\end{theorem}

\begin{remark*}
    We use $\calE \circ f$ to denote the channel where $f$ is applied individually to each symbol of the input, and then $\calE$ is applied to the resulting bitstring. Similarly, $f \circ \calE'$ is the channel where $\calE'$ is applied, and then $f$ is applied individually to each symbol of the resulting $\Sigma$-string. The condition that $f \circ \calE' = \calE \circ f$ means that $\calE'$ acts as $\calE$ in the binary representation. For instance, if $\calE'$ is a deletion channel over $\Sigma$, then $\calE$ is the corresponding deletion channel over $\{0,1\}$; if $\calE'$ is the channel that independently replaces each symbol with a random symbol with probability $p$, and $f$ is a random function, then $\calE$ is approximately $\BSC_{p/2}$.
\end{remark*}

\begin{proof}
    Steganographic security follows immediately from \Cref{claim:stego-orig} and the fact that a PRC is a pseudorandom encryption scheme.

    For robustness, let $\m$ be an arbitrary message and $h$ an arbitrary history, and let $c = c_1|| \ldots || c_n \gets \Steg.\Encode(\pk,\m,h)$.
    Recall that each $c_i$ is chosen as $RS^{M(h),f}(x_i,\secpar)$, where $x = x_1 || \ldots || x_n$ is the output of $\PRC.\Encode$.
    Since $f$ is unbiased, each sample $x'$ from $RS$ satisfies $f(x') = x_i$ independently with probability $1/2$, and therefore with overwhelming probability within $\secpar$ draws the sampler will output $c_i$ such that $f(c_i) = x_i$.

    Given $c' = \calE'(c)$, decoder first computes $x' = f(c'_1) || \ldots || f(c'_n)$.
    Observe that $x' = f \circ \calE'(c)$, which by assumption is equal to $\calE \circ f(c) = \calE(x)$, where $x$ is the output of $\PRC.\Encode$.
    By robustness of the PRC, $\PRC.\Decode(\sk,x',h) = \m$.
\end{proof}

Applying \Cref{theorem:stego} with our PRCs from \Cref{subsec:prcs-for-deletions}, we obtain a stateless public-key steganography scheme that is robust to random deletions and substitutions.

\begin{figure}
\centering
\begin{minipage}{0.48\textwidth}
\begin{algorithm}[H]
\caption{$\sf Steg.Encode$}
    \KwIn{Key $\pk$, message $\m$, history $h$.}
    Let $x \gets \Enc(\pk,\ m)$\;
    \For{$i = 1, \ldots, n$}{
    $c_i \gets RS^{M(h), f}(x_i, \kappa)$\;
    Set $h = h||c_i$\;
    }
    \Return{$c_1 || c_2 || \ldots || c_n$}
\end{algorithm}
\end{minipage}
\hspace{0.04\textwidth}
\begin{minipage}{0.46\textwidth}
\begin{algorithm}[H]
\caption{$\sf Steg.Decode$}
    \KwIn{Key $\sk$ and stegotext $c'$.}
    Parse $c'$ as $c'_1||c'_2|| \ldots ||c'_n$\;
    \For{$i = 1, \ldots, n$}{
        Set $x'_i = f(c'_i)$\;
    }
    Set $x' = x'_1||x'_2||\ldots ||x'_n$\;
    \Return{$\Dec(\sk, x')$}
\end{algorithm}
\end{minipage}
\caption{Encoding and decoding algorithms of ${\sf Steg}[(\KeyGen,\Enc,\Dec),f,\kappa]$.}
\label{fig:steg}
\end{figure}

\subsection{Stateless steganography from weaker modeling assumptions}

The assumption of the existence of an unbiased function $f$ over $\calC$ is quite strong. 
One way to relax this assumption is to assume instead that $f(\calC)$ meets some min-entropy requirement.
However, now $f$ may be unbalanced; for example, it could be the case that $\Pr_{c_i \gets \calC_h^n}[f(c_i) = 1] = 2/3$ for all $h$.
Now, the encoder in \Cref{fig:steg} would no longer be secure:
Consider sampling each symbol $c_i$ in the stegotext to \emph{always} satisfy $f(c_i) = x_i$. Then, since each $x_i$ is uniform in $\{0,1\}$, half of the symbols $c_i$ in the stegotext would satisfy $f(c_i) = 1$.
Therefore, one could distinguish between stegotexts and samples from $\calC$ using $f$.

A natural way (used by \cite{HLvA02}) to build a secure stateful steganography scheme under only this min-entropy assumption, is to let $c$ be an error-corrected version of the message and use in the encoder a rejection sampler that takes at most two samples, sometimes outputting $f(c_i) \neq x_i$.
This rejection sampler preserves the channel distribution as long as $c_i$ is random, which is not the case for generic error-correcting codes.
Therefore, \cite{HLvA02} relies on shared state to let the sender and receiver generate a fresh one-time pad per stegotext, letting $c$ be an error-corrected message with this one-time pad applied.
This shared state significantly simplifies the problem and is a strong assumption of its own.

Relaxing the assumption of an unbiased function poses more challenges for \emph{stateless} steganography, where this error-correction approach fails.
Although \cite{hopper2008provably} does construct a stateless steganography scheme under only a min-entropy assumption, this scheme has poor robustness. 
The idea behind this scheme is that the encoder first samples a symbol sequence $d$ that is long enough to have $\lambda$ min-entropy, and includes this sequence in the stegotext. 
Since both the encoder and decoder know $d$, $d$ can now act as their shared state, and the remainder of the stegotext is formed using a stateful scheme such as that of \cite{HLvA02}. 
The min-entropy assumption ensures that the state will never be used twice.

\begin{theorem} \label{theorem:stego-weaker-assumption}
    Let $\calC$ be a steganographic channel with alphabet $\Sigma$, and let $f : \Sigma \to \{0,1\}$ be a function. Suppose that there exists some $\alpha > 0$ such that for all histories $h \in \Sigma^*$, $\min_{b \in \{0,1\}}\Pr_{c \gets \calC_h}[f(c) = b] \geq \alpha$.
    Then for any $p \in (0,1/2)$, there exists $q \in (0,1/2)$ such that if $\PRC$ is any (public-key) PRC for every $q$-bounded channel, then ${\sf Steg}[\PRC,f,2]$ is a (public-key) stateless steganography scheme for $\calC$, with the same rate as $\PRC$, that is robust to any channel $\calE$ such that $f \circ \calE = \BSC_p \circ f$.
\end{theorem}
\begin{proof}
    We begin by showing that ${\sf Steg}[\PRC,f,2]$ is steganographically secure.
    As in \Cref{fig:steg}, suppose that we wish to encode a message $m$, let $x \from \PRC.\Encode(\pk, m)$, and let $c_i \gets RS^{M(h), f}(x_i, 2)$ be the $i^{\text{th}}$ symbol output by the steganography scheme.
    Let $c_i^1$ denote the first sample from $M(h)$ and let $c_i^2$ denote the second sample (which is irrelevant in the event that $f(c_i^1) = x_i$).

    By pseudorandomness of $x \from \PRC.\Encode(\pk, m)$, it suffices to show that for any $i \in [n]$, $h \in \Sigma^{i-1}$, and $c^* \in \Sigma$,
    \[
        \Pr_{\substack{x_i \from \{0,1\} \\ c_i \gets RS^{M(h), f}(x_i, 2)}}[c_i = c^*] = \Pr_{c \from M(h)}[c = c^*].
    \]
    Indeed,
    \ifeprint
    \begin{align*}
        \Pr_{\substack{x_i \from \{0,1\} \\ c_i \gets RS^{M(h), f}(x_i, 2)}}[c_i = c^*] &= \Pr_{\substack{x_i \from \{0,1\} \\ c_i^1 \gets M(h)}}[c_i^1 = c^* \land f(c_i^1) = x_i] + \Pr_{\substack{x_i \from \{0,1\} \\ c_i^1, c_i^2 \gets M(h)}}[c_i^2 = c^* \land f(c_i^1) \neq x_i] \\
        &= \frac{1}{2}\Pr_{c_i^1 \gets M(h)}[c_i^1 = c^*] + \frac{1}{2} \Pr_{c_i^2 \gets M(h)}[c_i^2 = c^*] \\
        &= \Pr_{c \from M(h)}[c = c^*].
    \end{align*}
    \fi

    Now we turn to showing that ${\sf Steg}[\PRC,f,2]$ is robust to any channel $\calE$ such that $f \circ \calE = \BSC_p \circ f$, provided that $\PRC$ is robust to every $q$-bounded channel for appropriate choice of $q$.
    
    Again, pseudorandomness of $x \from \PRC.\Encode(\pk, \m)$ allows us to assume $x$ is random. Observe that the bits $f(c_i)$ are correlated with the bits $x_i$:
    \ifeprint
    \begin{align*}
        \Pr_{\substack{x_i \from \{0,1\} \\ c_i \gets RS^{M(h), f}(x_i, 2)}}[f(c_i) = x_i] &= \Pr_{\substack{x_i \from \{0,1\} \\ c_i^1 \gets M(h)}}[f(c_i^1) = x_i] + \Pr_{\substack{x_i \from \{0,1\} \\ c_i^1, c_i^2 \gets M(h)}}[f(c_i^2) = x_i \land f(c_i^1) \neq x_i] \\
        &= \frac{1}{2} + \frac{1}{2} \Pr_{\substack{x_i \from \{0,1\} \\ c_i^1, c_i^2 \gets M(h)}}[f(c_i^2) = x_i \mid f(c_i^1) \neq x_i] \\
        &\ge \frac{1}{2} + \frac{1}{2} \alpha.
    \end{align*}
    \fi
    For $i \in [n]$ let $\hat{c}$ be the content generated by ${\sf Steg}$ (before the error channel $\calE$ is applied), defined by $\hat{c} = \hat{c}_1 || \cdots || \hat{c}_n$ for $\hat{c}_i = f(c_i)$. By Azuma's inequality (\Cref{theorem:azuma}), $\wt(\hat{c} \oplus x) \le (1/2-\alpha/4) \cdot n$ with probability $1-\negl[n]$. The decoder for the steganography scheme will compute $(f \circ \calE)(c)$, which is distributed identically to $(\BSC_p \circ f)(c) = \BSC_p(\hat{c})$ since $f \circ \calE = \BSC_p \circ f$. By a Chernoff bound and straightforward computation,
    \[
        \wt(\BSC_p(\hat{c}) \oplus x) \le \left(\frac{1}{2} - \frac{\alpha}{8} (1-2p)\right) \cdot n
    \]
    with probability $1-\negl[n]$. Therefore, if $\PRC$ is robust to any $q$-bounded channel for $q = 1/2 - (\alpha/8) (1-2p)$, then ${\sf Steg}[\PRC,f,2]$ is robust to $\calE$.
\end{proof}

\newpage

\bibliographystyle{alpha}
\bibliography{references}

\newcommand{\etalchar}[1]{$^{#1}$}
\begin{thebibliography}{KGW{\etalchar{+}}23b}

\bibitem[Aar22]{scott}
Scott Aaronson.
\newblock {My AI Safety Lecture for UT Effective Altruism}.
\newblock \url{https://scottaaronson.blog/?p=6823}, November 2022.
\newblock Accessed May 2023.

\bibitem[ABN{\etalchar{+}}92]{ABN+92}
Noga Alon, Jehoshua Bruck, Joseph Naor, Moni Naor, and Ron~M. Roth.
\newblock Construction of asymptotically good low-rate error-correcting codes
  through pseudo-random graphs.
\newblock {\em {IEEE} Trans. Inf. Theory}, 38(2):509--516, 1992.

\bibitem[ACI{\etalchar{+}}20]{agrikola2020pseudorandom}
Thomas Agrikola, Geoffroy Couteau, Yuval Ishai, Stanis{\l}aw Jarecki, and Amit
  Sahai.
\newblock On pseudorandom encodings.
\newblock In {\em Theory of Cryptography: 18th International Conference, TCC
  2020, Durham, NC, USA, November 16--19, 2020, Proceedings, Part III 18},
  pages 639--669. Springer, 2020.

\bibitem[ADI{\etalchar{+}}17]{applebaum2017secure}
Benny Applebaum, Ivan Damg{\aa}rd, Yuval Ishai, Michael Nielsen, and Lior
  Zichron.
\newblock Secure arithmetic computation with constant computational overhead.
\newblock In {\em Annual International Cryptology Conference}, pages 223--254.
  Springer, 2017.

\bibitem[Ale03]{Alekhnovich03}
Michael Alekhnovich.
\newblock More on average case vs approximation complexity.
\newblock In {\em 44th Symposium on Foundations of Computer Science {(FOCS}
  2003), 11-14 October 2003, Cambridge, MA, USA, Proceedings}, pages 298--307.
  {IEEE} Computer Society, 2003.

\bibitem[AP98]{And98}
Ross~J Anderson and Fabien~AP Petitcolas.
\newblock On the limits of steganography.
\newblock {\em IEEE Journal on selected areas in communications},
  16(4):474--481, 1998.

\bibitem[ARC{\etalchar{+}}01]{atallah2001natural}
Mikhail~J Atallah, Victor Raskin, Michael Crogan, Christian Hempelmann, Florian
  Kerschbaum, Dina Mohamed, and Sanket Naik.
\newblock Natural language watermarking: Design, analysis, and a
  proof-of-concept implementation.
\newblock In {\em Information Hiding: 4th International Workshop, IH 2001
  Pittsburgh, PA, USA, April 25--27, 2001 Proceedings 4}, pages 185--200.
  Springer, 2001.

\bibitem[ARH{\etalchar{+}}02]{atallah2002natural}
Mikhail~J Atallah, Victor Raskin, Christian~F Hempelmann, Mercan Karahan, Radu
  Sion, Umut Topkara, and Katrina~E Triezenberg.
\newblock Natural language watermarking and tamperproofing.
\newblock In {\em International workshop on information hiding}, pages
  196--212. Springer, 2002.

\bibitem[Ari09]{Ari09}
Erdal Arikan.
\newblock Channel polarization: a method for constructing capacity-achieving
  codes for symmetric binary-input memoryless channels.
\newblock {\em {IEEE} Trans. Inf. Theory}, 55(7):3051--3073, 2009.

\bibitem[ASS{\etalchar{+}}23]{ASSVV23}
Shweta Agrawal, Sagnik Saha, Nikolaj~Ignatieff Schwartzbach, Akhil Vanukuri,
  and Prashant~Nalini Vasudevan.
\newblock $k$-sum in the sparse regime.
\newblock Cryptology ePrint Archive, Paper 2023/488, 2023.
\newblock \url{https://eprint.iacr.org/2023/488}.

\bibitem[BC05]{backes2005public}
Michael Backes and Christian Cachin.
\newblock Public-key steganography with active attacks.
\newblock In {\em Theory of Cryptography Conference}, pages 210--226. Springer,
  2005.

\bibitem[BCG{\etalchar{+}}19]{boyle2019efficient}
Elette Boyle, Geoffroy Couteau, Niv Gilboa, Yuval Ishai, Lisa Kohl, and Peter
  Scholl.
\newblock Efficient pseudorandom correlation generators: Silent ot extension
  and more.
\newblock In {\em Advances in Cryptology--CRYPTO 2019: 39th Annual
  International Cryptology Conference, Santa Barbara, CA, USA, August 18--22,
  2019, Proceedings, Part III 39}, pages 489--518. Springer, 2019.

\bibitem[BH23]{BH23}
Diane Bartz and Krystal Hu.
\newblock {OpenAI}, {Google}, others pledge to watermark {AI} content for
  safety, {White House} says, July 2023.
\newblock
  \href{https://www.reuters.com/technology/openai-google-others-pledge-watermark-ai-content-safety-white-house-2023-07-21}{https://www.reuters.com/technology/openai-google-others-pledge-watermark-ai-content-safety-white-house-2023-07-21}.

\bibitem[Cac98]{cachin1998information}
Christian Cachin.
\newblock An information-theoretic model for steganography.
\newblock In {\em International Workshop on Information Hiding}, pages
  306--318. Springer, 1998.

\bibitem[CGZ23]{CGZ23}
Miranda Christ, Sam Gunn, and Or~Zamir.
\newblock Undetectable watermarks for language models.
\newblock {\em {IACR} Cryptol. ePrint Arch.}, page 763, 2023.

\bibitem[Cra98]{craver1998public}
Scott Craver.
\newblock On public-key steganography in the presence of an active warden.
\newblock In {\em International Workshop on Information Hiding}, pages
  355--368. Springer, 1998.

\bibitem[DGG{\etalchar{+}}15]{dodis2015formal}
Yevgeniy Dodis, Chaya Ganesh, Alexander Golovnev, Ari Juels, and Thomas
  Ristenpart.
\newblock A formal treatment of backdoored pseudorandom generators.
\newblock In {\em Advances in Cryptology--EUROCRYPT 2015: 34th Annual
  International Conference on the Theory and Applications of Cryptographic
  Techniques, Sofia, Bulgaria, April 26-30, 2015, Proceedings, Part I 34},
  pages 101--126. Springer, 2015.

\bibitem[DIRR09]{dedic2009upper}
Nenad Dedi{\'c}, Gene Itkis, Leonid Reyzin, and Scott Russell.
\newblock Upper and lower bounds on black-box steganography.
\newblock {\em Journal of Cryptology}, 22:365--394, 2009.

\bibitem[dWSK{\etalchar{+}}22]{de2022perfectly}
Christian~Schroeder de~Witt, Samuel Sokota, J~Zico Kolter, Jakob Foerster, and
  Martin Strohmeier.
\newblock Perfectly secure steganography using minimum entropy coupling.
\newblock {\em arXiv preprint arXiv:2210.14889}, 2022.

\bibitem[FGJ{\etalchar{+}}23]{public}
Jaiden Fairoze, Sanjam Garg, Somesh Jha, Saeed Mahloujifar, Mohammad Mahmoody,
  and Mingyuan Wang.
\newblock Publicly detectable watermarking for language models.
\newblock {\em Cryptology ePrint Archive}, 2023.

\bibitem[Fri09]{fridrich2009steganography}
Jessica Fridrich.
\newblock {\em Steganography in digital media: principles, algorithms, and
  applications}.
\newblock Cambridge University Press, 2009.

\bibitem[Gal62]{Gal62}
Robert~G. Gallager.
\newblock Low-density parity-check codes.
\newblock {\em {IRE} Trans. Inf. Theory}, 8(1):21--28, 1962.

\bibitem[GKVZ22]{backdoors}
Shafi Goldwasser, Michael~P Kim, Vinod Vaikuntanathan, and Or~Zamir.
\newblock Planting undetectable backdoors in machine learning models.
\newblock In {\em 2022 IEEE 63rd Annual Symposium on Foundations of Computer
  Science (FOCS)}, pages 931--942. IEEE, 2022.

\bibitem[HLVA02]{HLvA02}
Nicholas~J Hopper, John Langford, and Luis Von~Ahn.
\newblock Provably secure steganography.
\newblock In {\em Advances in Cryptology—CRYPTO 2002: 22nd Annual
  International Cryptology Conference Santa Barbara, California, USA, August
  18--22, 2002 Proceedings 22}, pages 77--92. Springer, 2002.

\bibitem[Hoe94]{hoeffding1994probability}
Wassily Hoeffding.
\newblock Probability inequalities for sums of bounded random variables.
\newblock {\em The collected works of Wassily Hoeffding}, pages 409--426, 1994.

\bibitem[Hop05]{hopper2005steganographic}
Nicholas Hopper.
\newblock On steganographic chosen covertext security.
\newblock In {\em International Colloquium on Automata, Languages, and
  Programming}, pages 311--323. Springer, 2005.

\bibitem[HvAL08]{hopper2008provably}
Nicholas Hopper, Luis von Ahn, and John Langford.
\newblock Provably secure steganography.
\newblock {\em IEEE Transactions on Computers}, 58(5):662--676, 2008.

\bibitem[KAAL23]{openaidetector}
Jan~Hendrik Kirchner, Lama Ahmad, Scott Aaronson, and Jan Leike.
\newblock New ai classifier for indicating {AI}-written text.
\newblock
  \url{https://openai.com/blog/new-ai-classifier-for-indicating-ai-written-text},
  January 2023.
\newblock Accessed May 2023.

\bibitem[KGW{\etalchar{+}}23a]{KGW+23}
John Kirchenbauer, Jonas Geiping, Yuxin Wen, Jonathan Katz, Ian Miers, and Tom
  Goldstein.
\newblock A watermark for large language models.
\newblock {\em arXiv preprint arXiv:2301.10226}, 2023.

\bibitem[KGW{\etalchar{+}}23b]{kirchenbauer2023reliability}
John Kirchenbauer, Jonas Geiping, Yuxin Wen, Manli Shu, Khalid Saifullah, Kezhi
  Kong, Kasun Fernando, Aniruddha Saha, Micah Goldblum, and Tom Goldstein.
\newblock On the reliability of watermarks for large language models.
\newblock {\em arXiv preprint arXiv:2306.04634}, 2023.

\bibitem[KJGR21]{meteor}
Gabriel Kaptchuk, Tushar~M. Jois, Matthew Green, and Aviel~D. Rubin.
\newblock Meteor: Cryptographically secure steganography for realistic
  distributions.
\newblock In Yongdae Kim, Jong Kim, Giovanni Vigna, and Elaine Shi, editors,
  {\em {CCS} '21: 2021 {ACM} {SIGSAC} Conference on Computer and Communications
  Security, Virtual Event, Republic of Korea, November 15 - 19, 2021}, pages
  1529--1548. {ACM}, 2021.

\bibitem[KKRT16]{kolesnikov2016efficient}
Vladimir Kolesnikov, Ranjit Kumaresan, Mike Rosulek, and Ni~Trieu.
\newblock Efficient batched oblivious prf with applications to private set
  intersection.
\newblock In {\em Proceedings of the 2016 ACM SIGSAC Conference on Computer and
  Communications Security}, pages 818--829, 2016.

\bibitem[KL07]{katzlindell}
Jonathan Katz and Yehuda Lindell.
\newblock {\em Introduction to modern cryptography: principles and protocols}.
\newblock Chapman and hall/CRC, 2007.

\bibitem[KTHL23]{KTHL23}
Rohith Kuditipudi, John Thickstun, Tatsunori Hashimoto, and Percy Liang.
\newblock Robust distortion-free watermarks for language models.
\newblock {\em arXiv preprint arXiv:2307.15593}, 2023.

\bibitem[MLK{\etalchar{+}}23]{detectgpt}
Eric Mitchell, Yoonho Lee, Alexander Khazatsky, Christopher~D. Manning, and
  Chelsea Finn.
\newblock Detect{GPT}: Zero-shot machine-generated text detection using
  probability curvature.
\newblock {\em CoRR}, abs/2301.11305, 2023.

\bibitem[MTSB13]{MTSB13}
Rafael Misoczki, Jean-Pierre Tillich, Nicolas Sendrier, and Paulo~SLM Barreto.
\newblock {MDPC-McEliece: New McEliece variants from moderate density
  parity-check codes}.
\newblock In {\em 2013 IEEE international symposium on information theory},
  pages 2069--2073. IEEE, 2013.

\bibitem[NN93]{NN93}
Joseph Naor and Moni Naor.
\newblock Small-bias probability spaces: Efficient constructions and
  applications.
\newblock {\em {SIAM} J. Comput.}, 22(4):838--856, 1993.

\bibitem[PSF{\etalchar{+}}23]{piet2023mark}
Julien Piet, Chawin Sitawarin, Vivian Fang, Norman Mu, and David Wagner.
\newblock Mark my words: Analyzing and evaluating language model watermarks.
\newblock {\em arXiv preprint arXiv:2312.00273}, 2023.

\bibitem[RBB03]{rogaway2003ocb}
Phillip Rogaway, Mihir Bellare, and John Black.
\newblock Ocb: A block-cipher mode of operation for efficient authenticated
  encryption.
\newblock {\em ACM Transactions on Information and System Security (TISSEC)},
  6(3):365--403, 2003.

\bibitem[Sim84]{simmons1984prisoners}
Gustavus~J Simmons.
\newblock The prisoners’ problem and the subliminal channel.
\newblock In {\em Advances in Cryptology: Proceedings of Crypto 83}, pages
  51--67. Springer, 1984.

\bibitem[Ta{-}17]{TaShma17}
Amnon Ta{-}Shma.
\newblock Explicit, almost optimal, epsilon-balanced codes.
\newblock In Hamed Hatami, Pierre McKenzie, and Valerie King, editors, {\em
  Proceedings of the 49th Annual {ACM} {SIGACT} Symposium on Theory of
  Computing, {STOC} 2017, Montreal, QC, Canada, June 19-23, 2017}, pages
  238--251. {ACM}, 2017.

\bibitem[TCH23]{tang2023science}
Ruixiang Tang, Yu-Neng Chuang, and Xia Hu.
\newblock The science of detecting {LLM}-generated texts.
\newblock {\em arXiv preprint arXiv:2303.07205}, 2023.

\bibitem[Tia23]{gptzero}
Edward Tian.
\newblock {GPTZero} update {V1}.
\newblock \url{https://gptzero.substack.com/p/gptzero-update-v1}, January 2023.
\newblock Accessed May 2023.

\bibitem[TTDI05]{topkara2005natural}
Mercan Topkara, Cuneyt~M Taskiran, and Edward~J Delp~III.
\newblock Natural language watermarking.
\newblock In {\em Security, Steganography, and Watermarking of Multimedia
  Contents VII}, volume 5681, pages 441--452. SPIE, 2005.

\bibitem[vAH04]{vAH04}
Luis von Ahn and Nicholas~J. Hopper.
\newblock Public-key steganography.
\newblock In Christian Cachin and Jan Camenisch, editors, {\em Advances in
  Cryptology - {EUROCRYPT} 2004, International Conference on the Theory and
  Applications of Cryptographic Techniques, Interlaken, Switzerland, May 2-6,
  2004, Proceedings}, volume 3027 of {\em Lecture Notes in Computer Science},
  pages 323--341. Springer, 2004.

\bibitem[VV83]{trapdoor}
Umesh~V. Vazirani and Vijay~V. Vazirani.
\newblock Trapdoor pseudo-random number generators, with applications to
  protocol design.
\newblock In {\em 24th Annual Symposium on Foundations of Computer Science
  (sfcs 1983)}, pages 23--30, 1983.

\bibitem[ZALW23]{ZALW23}
Xuandong Zhao, Prabhanjan Ananth, Lei Li, and Yu-Xiang Wang.
\newblock Provable robust watermarking for {AI}-generated text.
\newblock {\em arXiv preprint arXiv:2306.17439}, 2023.

\bibitem[Zam24]{zamir2024excuse}
Or~Zamir.
\newblock Excuse me, sir? your language model is leaking (information).
\newblock {\em arXiv preprint arXiv:2401.10360}, 2024.

\bibitem[ZDR19]{ZieglerDR19}
Zachary~M. Ziegler, Yuntian Deng, and Alexander~M. Rush.
\newblock Neural linguistic steganography.
\newblock In Kentaro Inui, Jing Jiang, Vincent Ng, and Xiaojun Wan, editors,
  {\em Proceedings of the 2019 Conference on Empirical Methods in Natural
  Language Processing and the 9th International Joint Conference on Natural
  Language Processing, {EMNLP-IJCNLP} 2019, Hong Kong, China, November 3-7,
  2019}, pages 1210--1215. Association for Computational Linguistics, 2019.

\bibitem[ZEF{\etalchar{+}}23]{zhang2023watermarks}
Hanlin Zhang, Benjamin~L Edelman, Danilo Francati, Daniele Venturi, Giuseppe
  Ateniese, and Boaz Barak.
\newblock Watermarks in the sand: Impossibility of strong watermarking for
  generative models.
\newblock {\em arXiv preprint arXiv:2311.04378}, 2023.

\bibitem[ZHR{\etalchar{+}}19]{zellers2019defending}
Rowan Zellers, Ari Holtzman, Hannah Rashkin, Yonatan Bisk, Ali Farhadi,
  Franziska Roesner, and Yejin Choi.
\newblock Defending against neural fake news.
\newblock {\em Advances in neural information processing systems}, 32, 2019.

\end{thebibliography}

\appendix

\section{Related work}
\label{sec:full-related-work}

\paragraph{Comparison to code-based cryptography.} Our approach to constructing pseudorandom codes is closely related to code-based cryptography, and in particular the McEliece cryptosystem. The McEliece cryptosystem is a public-key encryption scheme in which the public key is a generator matrix for a linear code and the secret key is a trapdoor for efficient decoding. In this sense, our constructions of pseudorandom codes can be viewed as instantiations of the McEliece cryptosystem. However, there are two key differences between our setting and that of code-based cryptography that render code-based cryptography results inapplicable for our purposes.

The first difference is in the source of the errors, or noise. PRCs are required to correct uncontrolled errors introduced by an adversarial channel; in McEliece, all of the errors are introduced by an honest user for the purpose of securely sending the message over a perfect channel. Therefore, the channel itself is specified as part of the cryptosystem and the code only needs to be robust to \emph{this channel}.
In pseudorandom codes, it is crucial that we can correct from an arbitrary constant rate $p < 1/2$ of additional errors introduced by a noisy channel on a binary alphabet. 
Existing code-based cryptosystems (such as binary Goppa codes, Reed-Solomon codes, Reed-Muller codes, or Algebraic Geometry codes) either do not enjoy such strong robustness or rely on a larger alphabet.

The second difference is that we require noisy codewords to be \emph{pseudorandom}, rather than just hiding. If every (noisy) codeword satisfies some efficiently-computable property $P$, but a uniformly random string does not satisfy $P$, then the code cannot be a PRC. However, such a code might still define a secure encryption scheme, as long as $P$ does not distinguish between codewords for different messages.

Our primary construction of PRCs is based on low-density parity-check (LDPC) codes, with somewhat-higher $\Omega(\log n)$ density than is typically considered.
The work of \cite{MTSB13} consider ``moderate-density parity-check'' codes, with density $\Omega(\sqrt{n \log n})$.
While their codes suffice for code-based cryptography, they do not enjoy as strong robustness as ours because of the higher density.
In particular, they are not robust to any constant rate of substitutions.

\paragraph{PRC-like cryptographic tools.}
Error-correcting codes with some pseudorandomness properties have proven useful in cryptographic applications such as pseudorandom correlation generators \cite{boyle2019efficient} and oblivious linear evaluation (OLE) \cite{applebaum2017secure}.
However, in all of these applications, decoding requires new side information for each message, such as the locations of the errors.
The fundamental difficulty in constructing our pseudorandom codes is that the only extra information provided to the decoder is a single, unchanging secret key; without this key the messages must be pseudorandom.

\emph{Pseudorandom encodings}, defined and studied in \cite{agrikola2020pseudorandom}, are similar in name to PRCs but are very different objects.
A pseudorandom encoding for a distribution $X$ is a pair of unkeyed algorithms $(\Encode, \Decode)$ such that $\Encode(x)$ is pseudorandom (i.e., appears uniform) for a random $x \gets X$, and $\Decode(\Encode(x)) = x$ with overwhelming probability.
For pseudorandom encodings there is no robustness requirement.
Furthermore, pseudorandomness for PRCs requires that $\Encode(x)$ is pseudorandom for any $x$ of the adversary's choice.

\paragraph{Language watermarking schemes.}
Classical watermarking considers the problem of embedding a signal into a \emph{fixed} object, such as a given text, in such a way that the watermark is hard to remove, and the quality of the original object is preserved.
In that setting, studied specifically for natural language in \cite{topkara2005natural,atallah2001natural,atallah2002natural}, planting the watermark must alter the text, and the goal is to minimize some distance measure between the original text and the watermarked text.
In contrast, since generative models are randomized, the quality goal of watermarks for generative models 
is to preserve certain properties of the \emph{distribution} of the model's outputs.
Due to this difference, new techniques have been developed for watermarking AI-generated content, particularly text output by language models, which is the focus of our watermarking work in this paper.

A language model takes as input a \emph{prompt} and randomly samples a \emph{response} using an iterative process.
Given a prompt, it computes a probability distribution $\p_1$ for the first \emph{token} and samples a token $\t_1$ from this distribution.
At each subsequent step, it takes as input the token sequence $\t_1, \dots, \t_{i-1}$ output thus far, computes the distribution $\p_i$, and samples the $i^{\text{th}}$ token in the response $\t_i$ from $\p_i$.
It continues doing so until it samples a special ``done'' token, at which point it terminates and outputs the response. 
The randomness in this process is important for the usefulness of the model---a model should produce a wide variety of useful responses given the same prompt---and it is crucial for watermarking.

Recently, there have been several language model watermarking schemes proposed (e.g., \cite{scott,KGW+23,CGZ23,ZALW23,KTHL23,public}), which embed the watermark by changing the way that each $\t_{i}$ is sampled from $\p_i$. 
At a very high level, these watermark embedders give preference to certain tokens in the sampling process. The corresponding detectors check for the presence of more preferred tokens than would be expected in independently generated text.

The simplest of these schemes, \cite{ZALW23}, fixes a random partition of tokens into equal-sized green and red lists, and it embeds the watermark by increasing the probabilities of tokens on the green list when it samples its response.
The detector computes the fraction of green tokens, which should be appreciably greater than $\frac{1}{2}$ for watermarked text and roughly $\frac{1}{2}$ for all other text.
This red-green partition is used for every sampled token across all responses, which results in a significant change to the model's distribution. 
For example, if ``computer'' is on the red list, the model now prefers not to talk about computers.

\cite{scott,KGW+23} mitigate this distributional shift by generating the red-green partitions dynamically.
Rather than using the same red-green partition for every token, \cite{scott,KGW+23} select a partition for each token using a PRF evaluated at the previous $k$ tokens.
That is, these schemes use a PRF $F_\sk$ with secret key $\sk$, and add weight to tokens $\t_i$ such that $F_\sk(\t_{i-k}, \dots, \t_{i-1}, \t_i) = 1$.
This effectively generates a new red-green partition for each token, assuming that no prefix $\t_{i-k}, \dots, \t_{i-1}$ ever appears twice.
The detector computes the fraction of tokens on these green lists, which it can compute by evaluating the PRF, provided that the seeds are intact in the given text.
The length $k$ of the token sequence used as these seeds can be tuned to trade off between robustness and frequency of seed reuse.
Longer seeds result in less frequent seed reuse but make the watermark easier to remove, since the adversary can destroy every seed and prevent the detector from computing the green lists by changing only $1/k$ tokens in a watermarked response.

Although \cite{KGW+23} and \cite{scott} both generate randomness using this seeded PRF, they use this randomness to sample from $\p_i$ differently, achieving different quality guarantees as a result.
The watermarking algorithm in \cite{scott} samples tokens such that, for each token of the response, the expected watermarked distribution (over the randomness of $\sk$) is identical to the original model's distribution.
However, while individual tokens' distributions are preserved, the watermark introduces correlations between tokens' distributions as seeds can be reused (even within the same response).
For example, if the same prompt is queried multiple times, the watermark will result in the same bias for the first token of the response each time.
\cite{KGW+23} achieves a weaker guarantee and changes the distribution of the model's outputs.

\emph{Distortion-freeness}, a quality notion defined by \cite{KTHL23}, is the property that the expected distribution of a single response from the watermarked model is identical to the distribution of a single response from the original model.
\cite{KTHL23} constructs a distortion-free watermarking scheme with an interesting level of robustness.
To generate each response, this watermarking scheme samples a string $x^*$ from a fixed collection of seeds $x_1, \dots, x_\ell$.
The $i^{\text{th}}$ token in the response is chosen to be correlated with the $i^{\text{th}}$ bit of $x^*$.
The detector checks whether the given text is watermarked by computing its edit distance from each of the strings $x_1, \dots, x_\ell$, returning true if any of these distances is below some threshold.

This detection algorithm is quite slow in practice: It runs in time $O(n^2 \ell)$, where $n$ is the length of the random strings.
While using a smaller $\ell$ improves the detector's efficiency, it deteriorates the model's quality by reducing the response variability.
Our watermark avoids this tradeoff between variability of the model and detector efficiency, while still offering robustness.

The strongest quality guarantee is \emph{undetectability}, defined by \cite{CGZ23}, which is the quality notion we focus on in this work.
Undetectability requires that the watermarked model and original model are computationally indistinguishable to an adversary without knowledge of the detection key, even if the adversary can make an unbounded number of adaptive queries.
In contrast, distortion-freeness of \cite{KTHL23} says nothing about multiple responses from the model.
In particular, a watermark can be distortion-free but render the response to a given prompt entirely deterministic --- even if the original model had a great deal of variability on the same prompt.

Like \cite{KGW+23,scott}, the watermark in \cite{CGZ23} uses a PRF seeded with previously output tokens to generate the randomness used to bias the token sampler.
A crucial observation used in \cite{CGZ23} to avoid the seed reuse issue of \cite{KGW+23,scott}, is that if the tokens constituting the seed contain enough entropy, seed reuse becomes unlikely. 
Because the watermarking algorithm has access to the token distributions, it can compute the amount of entropy in a token sequence and use it as a seed only once this entropy has exceeded a certain threshold.
This allows \cite{CGZ23} to achieve undetectability and a robustness guarantee they call \emph{substring-completeness}, that any sufficiently high-entropy substring of a response will be detected as watermarked.
This is essentially the same robustness guarantee as \cite{KGW+23,scott}, but it is weaker than that of \cite{KTHL23,ZALW23}, since changing one token in every seed removes the watermark.

\cite{public} has similar quality to \cite{CGZ23}, and embeds a digital signature in the response to make it detectable by anyone with a public detection key. A separate secret key is required to embed the watermark.
Similarly to \cite{CGZ23}, their watermark generates the initial portion of the response from the original model. It then uses this portion as a seed for a PRF and encodes each bit of a digital signature by sampling a block of tokens such that the PRF evaluated on that block is equal to that bit.
Like \cite{CGZ23}, changing any token in the seed removes the watermark.

We compare some of these watermarking schemes in \Cref{tab:comparison}. We refer the reader to \cite{piet2023mark} for a more detailed empirical comparison of some of these schemes, and to \cite{kirchenbauer2023reliability} for an empirical study of the robustness of \cite{KGW+23}.

\begin{table}
    \centering
    \begin{threeparttable}
    \begin{tabular}{c|c|c|c}
         Paper & Single-query undetectable? & Undetectable? & Robust?\\
         \hline
         \cite{scott} & \cmark\tnote{*} & \xmark & \xmark \\
         \cite{KGW+23} & \xmark & \xmark & \xmark \\
         \cite{CGZ23} & \cmark & \cmark & \xmark \\
         \cite{ZALW23} & \xmark & \xmark & \cmark \\
         \cite{KTHL23} & \cmark & \xmark & \cmark \\
         \cite{public} & \cmark & \cmark & \xmark \\
         This work & \cmark & \cmark & \cmark
    \end{tabular}
    \caption{Comparison between various watermarking schemes for language models. ``Single-query undetectable'' means that a single query to the watermarked model is computationally indistinguishable from one to the original model, and ``undetectable'' is the multi-query analogue as defined in \cite{CGZ23}.
    In this table, ``robust'' means that the watermark is resilient to a constant rate of substitutions.\label{tab:comparison}}
    \begin{tablenotes} \footnotesize
      \item[*] \cite{scott} is single-query undetectable provided that every short contiguous sequence of response tokens has high enough entropy.
    \end{tablenotes}
    \end{threeparttable}
\end{table}

\paragraph{On removing watermarks.}
\cite{CGZ23} describe an attack that removes any undetectable watermark and preserves response quality, but this attack is very impractical because it involves querying the model once for each token in the response.
\cite{zhang2023watermarks} describe an attack that can remove any watermark in a way that preserves response quality, assuming that the attacker has access to a quality oracle and that random walks over the graph of responses mix sufficiently quickly.
These assumptions are quite strong, and it is not clear whether there are fast, reliable quality oracles that don't already yield a full language model.
Still, in light of these impossibility results, one cannot hope to have a watermarking scheme that is robust against all possible attacks --- indeed, an adversary with sufficient knowledge of the language could always write a high-quality text without even consulting the watermarked model.
Therefore, we instead define robustness against particular classes of adversaries, including those that delete and substitute tokens.

\paragraph{Other techniques for detecting machine-generated text.}
Another approach for detecting AI-generated text is to train a machine learning classifier to distinguish between AI-generated and natural text \cite{zellers2019defending,detectgpt,gptzero,openaidetector}.
However, these classifiers lack transparency, can be easily evaded, have high false-positive rates, and have unpredictable biases.
OpenAI retracted its classifier-based detector due to these issues.

Instead of relying on existing idiosyncrasies of AI-generated text, watermarks embed patterns themselves, making detection more reliable and transparent.
See, e.g., \cite{tang2023science,piet2023mark} for an overview of the various approaches to detecting AI-generated text.

\paragraph{Undetectable backdoors for models.}
\cite{backdoors} shows how to embed a hidden ``backdoor'' during training of a model. Using a secret key, one can ``activate'' the backdoor by slightly perturbing inputs to alter their classification under the backdoored model.
These backdoors are undetectable in the sense that, without the secret key, one cannot distinguish between an honestly trained or a backdoored model.
While not closely related to our work in technical content, \cite{backdoors} is similar in spirit as an application of cryptographic definitions and techniques to machine learning.

\paragraph{Universal steganography.}
Steganography was introduced by Simmons \cite{simmons1984prisoners}, who presented it as problem where two prisoners wish to communicate in the presence of a warden, hiding not only the content of their communication but also the fact that this communication is occurring.
\cite{HLvA02} first formalized steganography in the computational setting.
Here, there is some distribution of \emph{covertexts} in which the sender wishes to conceal its \emph{message}.
In \Cref{sec:stego}, we consider \emph{universal} steganography, where the covertext distribution is an arbitrary distribution to which the sender has sample access. The sender uses this sample access and a secret key to construct a \emph{stegotext} that encodes the message, which the receiver decodes using the secret key. 

Loosely speaking, steganographic security requires that an outside observer cannot distinguish stegotexts from covertexts, and furthermore cannot learn anything about the message.
There are several security variants in the literature, including information-theoretic security \cite{cachin1998information} and security against active attacks \cite{backes2005public,hopper2005steganographic}.
We focus on (computational) \emph{security against chosen message attacks} (CMA) as defined in \cite{HLvA02}, which is analogous to CPA security of encryption.
There are several constructions of CMA-secure secret-key steganography schemes under certain assumptions, including \cite{HLvA02,hopper2008provably,dedic2009upper}.
We present a comparison of these schemes' properties in \Cref{tab:stego-comparison}.
All of these schemes either make the very strong assumption that there is a known hash function that is \emph{unbiased} on the covertext distribution, or else rely on a synchronized shared state between the sender and receiver.

An additional desirable property of a stegosystem, and one that is challenging to achieve, is \emph{robustness}: Even if the stegotext is corrupted by an adversary, the receiver should be able to recover the message.
To the best of our knowledge, prior to this work there was no provably secure \emph{stateless} secret-key stegosystem with nontrivial robustness.

Our robustness notion is different from that of \cite{HLvA02}, called \emph{substitution-robustness}. 
In substitution-robustness, an adversary may make substitutions of some symbols of the stegotext before it is given to the receiver, which should still recover the message.
The set of substitutions that their adversary is allowed to make is parameterized by a relation $R$.
That is, the adversary can change any symbol $s$ to a symbol $s'$ provided that $(s, s') \in R$.
This definition breaks down when the alphabet is binary, since if $R$ is nontrivial, containing without loss of generality $(0,1)$, the adversary can change all stegotexts to the all-one string.
Furthermore, it does not capture an adversary that has a bound on the number of symbols it may change, but that can change each symbol to any other symbol.
Our definition is thus incomparable to theirs, and our channel may introduce deletions rather than just substitutions.

While \cite{HLvA02} constructs a robust steganography scheme under this relation-based definition, they assume that the sender and receiver can share some state; in their scheme, this state is a synchronized counter $N$ of the number of messages sent so far. See \Cref{subsec:techo-stego} for a discussion on synchronized states in steganography.

As with encryption, steganography has a public-key analogue, which was defined formally by \cite{vAH04}. That is, encoding is possible with a public key, and decoding requires a secret key. 
While \cite{vAH04} was the first to define and formally prove security of public-key steganography, public-key schemes existed in prior work \cite{And98,craver1998public}.
The main steganography construction in this work is public-key as well.

While we focus on steganography with provable security in the models of \cite{HLvA02,vAH04}, we note that there is a large body of work that constructs steganography schemes with heuristic guarantees.
We refer the interested reader to \cite{fridrich2009steganography}.

\begin{table}
    \centering
    \begin{threeparttable}
    \begin{tabular}{c|c|c|c|c}
         Scheme & Stateless? & No unbiased function? & Public-key? & Robust?\\
         \hline
         \cite{HLvA02} Construction 2\tnote{*} & \cmark & \xmark & \xmark & \xmark\\
         \cite{HLvA02} Constructions 3,4 & \xmark & \cmark & \xmark & \cmark\\
         \cite{hopper2008provably} $\texttt{NoState}$ Construction & \cmark & \xmark & \xmark & \xmark\\
         \cite{vAH04} Construction 2 & \cmark & \xmark & \cmark & \xmark\\
        This work & \cmark & \cmark & \cmark & \cmark 
    \end{tabular}
    \caption{Comparison of universal steganography schemes. ``Stateless'' means that the sender and receiver do not need to share a synchronized state. ``No unbiased function'' means that the scheme does not require a function that is unbiased over the covertext distribution; instead, schemes with a check mark in this column rely on an assumption about the entropy of the text. As in \Cref{tab:comparison}, ``robust'' means robustness to a constant rate of substitutions.} \label{tab:stego-comparison}
    \begin{tablenotes} \footnotesize
      \item[*] Although this scheme was also proposed by \cite{And98,cachin1998information}, we refer to the construction from \cite{HLvA02}, whose setting and terminology we follow.
    \end{tablenotes}
    \end{threeparttable}
\end{table}

\paragraph{Language model steganography.}
The formal setting of steganography from \cite{HLvA02} assumes that the sender has sample access to the covertext distribution.
It is unclear how to realize this assumption in practice: If the sender wishes to conceal its message in casual conversation, it must be able to sample a random casual conversation.
One solution is to use a language model as this sampler for the covertext distribution, and there are several \emph{language model steganography schemes} tailored to this setting where the sender interacts with a language model to craft its stegotexts \cite{meteor,de2022perfectly,ZieglerDR19,zamir2024excuse}.
This is a relaxation of universal steganography, since these language model steganography schemes leverage the sender's explicit access to the covertext distribution via the model.
Therefore, our universal steganography scheme in \Cref{sec:stego} is incomparable to these schemes.

Language model steganography and undetectable watermarks for language models are closely related, as both are concerned with secretly embedding a signal in the output of a language model.
Indeed, we note in \Cref{subsec:water-simple} that our watermarking scheme yields a language model steganography scheme, since undetectability implies steganographic secrecy. 
Our resulting steganography scheme has the strongest robustness of any existing scheme, and in particular is the only language model steganography scheme with robustness to a constant rate of random substitutions.

\cite{zamir2024excuse} presents a language model steganography scheme derived from the watermarking scheme of \cite{CGZ23}. This scheme is also stateless and relies on a minimal entropy assumption about the text, but it is not public-key or robust.

In Meteor \cite{meteor}, the sender and receiver share a generative language model, and they maintain a shared history of the prompt and all tokens output thus far by the model.
To encode a message, the sender queries the model to obtain the distribution $\p_i$ over the next token.
It then samples the next token in a way that encodes some information about the message.
Because the sender and receiver share the model description and the token history \emph{including the prompt}, the receiver can compute $\p_i$.
Meteor crucially uses the ability of the receiver to compute $\p_i$, in order to decode the message.
Any change to the text history at all (even removing the prompt) can destroy the receiver's ability to compute $\p_i$, and therefore the receiver's ability to decode the message.

\cite{de2022perfectly} constructs a steganography scheme using minimum entropy couplings, in the information theoretic setting of \cite{cachin1998information}.
In their scheme, the decoder also must know the prompt used, as it requires access to an explicit description of the covertext distribution which is determined by the prompt. 
Furthermore, the adversary is not allowed to tamper with the stegotexts.

\paragraph{Backdoored/trapdoor PRGs.}
A \emph{trapdoor} or \emph{backdoored} PRG \cite{trapdoor,dodis2015formal} is a pseudorandom generator that outputs a sequence of bits that are pseudorandom to an observer, but where a party holding a secret key can distinguish this sequence from random.
In the context of backdoored PRGs, the secret key is viewed as a potential vulnerability of the PRG.
A PRC is in particular a trapdoor/backdoored PRG, since codewords appear uniformly random to an outside observer but can be detected with a secret key.

However, a PRC comes with the additional property of robustness.
This is especially interesting in the context of backdoored PRGs, where an \emph{immunizer} may be applied to the PRG output in an attempt to thwart the adversary holding the backdoor.
For example, one might apply a hash to the PRG output; \cite{dodis2015formal} show that this immunizer is effective for certain PRGs.
In that work they consider three classes of immunizers: public immunizers where the adversary can construct the trapdoor PRG based on knowledge of the seed of the function to be applied, semi-private immunizers where the adversary does not know the seed when constructing the PRG but does know it at attack time, and private immunizers where the adversary does not know the seed at all.
Our PRCs show a strong impossibility result for immunizers against an adversary distinguishing PRG outputs from uniform randomness. For any immunizers in the class of channels that the PRC is robust to, the adversary can distinguish a single immunized output from uniform randomness with overwhelming probability. This applies even to private immunizers, since the randomness of the channel is not known to the code detection algorithm.

\end{document}